%% file: closest.tex
\documentclass{jfm}

\usepackage{graphicx}
\usepackage{natbib}
\usepackage{hyperref}
\usepackage{todonotes}
\hypersetup{
    colorlinks = true,
    urlcolor   = blue,
    citecolor  = black,
}

\usepackage{epstopdf, epsfig}
\usepackage{amsmath}
\usepackage{bm}

\usepackage[T1]{fontenc}
\usepackage[english]{babel}
\usepackage{float}
\usepackage{amssymb}
\usepackage{algorithm}
\usepackage{enumitem}

\graphicspath{{./figs/}{./figs/channel/}}
\DeclareGraphicsExtensions{.eps,.ps,.jpg}

\makeatletter

\newtheorem{theorem}{Theorem}

\newcommand{\RomanNumeralCaps}[1]

\makeatother

\providecommand{\definitionname}{Definition}

\shorttitle{Three-dimensional flow visualization via streamsurfaces of the closest first integral}
\shortauthor{S. Katsanoulis et al.}

\title{Approximate streamsurfaces \\ for flow visualization}

\author{Stergios Katsanoulis\aff{1}, 
Florian Kogelbauer\aff{2}, 
Roshan Kaundinya\aff{1}, 
Jesse Ault\aff{4}
\and George Haller\aff{1}\corresp{\email{georgehaller@ethz.ch}}}

\affiliation{%
	\aff{1}Institute for Mechanical Systems, ETH Zurich, 8092 Zurich, Switzerland
	\aff{2} Department of Mechanical and Process Engineering, ETH Zurich, 8092 Zurich, Switzerland
	\aff{3} School of Engineering, Brown University, Providence, RI 02912, USA
	}

\begin{document}
\maketitle

\begin{abstract}
Instantaneous features of three-dimensional velocity fields are most directly visualized via streamsurfaces.
It is generally unclear, however, which streamsurfaces one should pick for this purpose, given that infinitely many such surfaces pass through each point of the flow domain.
Exceptions to this rule are vector fields with a nondegenerate first integral whose level surfaces globally define a continuous,  one-parameter family of  streamsurfaces. While generic vector fields have no  first integrals, their  vortical regions may admit local first integrals over a discrete set of streamtubes, as Hamiltonian systems are known to do over Cantor sets of invariant tori.
Here we introduce a method to construct such  first integrals approximately from velocity data, and show that their level sets indeed frame vortical features of the velocity field in examples in which those features are  known from Lagrangian analysis.
Moreover, we test our method in numerical data sets, including a flow inside a V-junction and a turbulent channel flow. For the latter, we propound an algorithm to pin down the most salient barriers to momentum transport up to a given scale providing a way out of the occlusion conundrum that typically accompanies other vortex visualization methods.


\end{abstract}


\selectlanguage{english}
\input{ClosestFIntegral_new}
\bibliographystyle{jfm}
\bibliography{biblio}

\end{document}

%% file: ClosestFIntegral_new.tex
\section{Introduction}
Streamsurfaces of a three-dimensional flow are two-dimensional surfaces composed of streamlines. Even a small but judiciously chosen set of such surfaces can give an effective characterization of the global topology of the velocity field. In steady flows, streamsurfaces are also invariant manifolds for the particle motion and hence frame the Lagrangian  particle dynamics. For these reasons, streamsurfaces should be, in principle, the simplest tool for illustrating instantaneous features of a velocity field.

Analogously, to identify and visualize vortical features of a velocity field, \cite{yang2010lagrangian} generalized the notion of vortex tubes and sheets \citep{batchelor2000introduction} by defining the vortex-surface field (VSF), i.e., a smooth scalar field whose isosurfaces act as 2D invariant manifolds of the vorticity field. Initially developed for symmetric, inviscid flows \citep{yang2010lagrangian}, this work has been extended to capture approximations to the Lagrangian (material) evolution of VSFs in analytic viscous flows \citep{yang2011evolution}, shear flows \citep{xiong2017boundary}, compressible flows \citep{peng2018effects}, transitional wall flows \citep{zhao2016evolution,zhao2016vortex} and in homogeneous isotropic turbulence \citep{xiong2019identifying}. 
A different approach, based on the spherical Clebsch maps, was used by \cite{chern2017inside} to visualize vortex lines and surfaces in computer graphics.
However, neither VSFs nor streamsurfaces are objective (frame-indifferent; \cite{haller2005objective}), which renders their experimental detection ambiguous.

In contrast, recent results on the transport of dynamically active vector fields (such as the momentum and vorticity fields) have yielded objectively defined barrier vector fields whose distinguished invariant surfaces turn out to be frame-indifferent material barriers to active transport \citep{haller2020objective}. Specifically, the barrier vector fields are the velocity and vorticity Laplacian in the instantaneous (Eulerian) limit and the time-averaged pullbacks of these Laplacian fields in the Lagrangian case. The invariant manifolds of these barrier vector fields have been shown to highlight vortical features of the velocity field with increased accuracy in several 2D and 3D flows \citep{haller2020objective,aksamit2022objective}.

Irrespective of the underlying vector field, no general methodology is available for its efficient visualization via a well-placed set of invariant surfaces. This is because streamlines launched from any smooth curve form a streamsurface by definition. As a consequence, infinitely many streamsurfaces cross through any point of the flow domain. It is unclear which (if any) of these surfaces should be selected even locally as a representative of the vector field topology. As a consequence, flow visualization packages tend to rely on user input for seed points of streamlines and streamsurfaces.

Most related ongoing research in the scientific visualization community focuses either on the more accurate computation of streamsurfaces from select seed points (\cite{hultquist92}) or streamline selection for streamsurfaces based on visual optimization (\cite{born10}, \cite{schulze14}). A relatively recent realization is that the streamsurfaces framing the instantaneous flow behavior most efficiently are the key invariant manifolds of the instantaneous velocity field.  Local stable and unstable manifolds near stagnation points and closed streamlines indeed well illustrate the instantaneous local velocity field geometry (see, e.g., \cite{peikert2009topologically}) but generally stretch and fold globally. These globally filamented streamsurfaces then lose their ability to demarcate different flow regions efficiently, with  unavoidable inaccuracies also arising in their computation (\cite{sadlo07}). Closest in spirit to our work is the observation of \cite{van1993implicit}, who seeks streamsurfaces as level sets of a scalar function. After defining a scalar distribution within an inflow boundary, the level curves of this scalar field are propagated along streamlines into the flow. As a consequence, the resulting surfaces will generally stretch, fold and self-accumulate, resulting in filamenting streamsurfaces that also depend on the choice of the initial scalar distribution. Reviews of all these  approaches in scientific visualization community can be found, e.g., in \cite{peikert2009topologically} and  \cite{martinez13}.

A representative set of streamsurfaces for flow illustration should arguably include at least one surface for each typically observed streamline topology,  as well as  the (generally unobserved) surfaces separating these different topological classes. Finding such a set of streamlines is straightforward for 2D incompressible flows by the existence of a first integral (conserved quantity) for the equation of instantaneous streamlines. This first integral is the streamfunction, whose set of level curves contains all typically observed streamline families as well as separatrices among them. Therefore, one can systematically scan through the one-parameter family of level curves and select those that stand out to be included in the visualization. 

For the lack of a streamfunction in 3D flows, the above program can only be carried out for streamsurfaces of integrable 3D flows, such as steady Euler flows that do not satisfy the Beltrami property. For such flows, the Bernoulli function provides a non-degenerate first integral \citep{arnold1999topological} whose one-parameter family of level surfaces can be scanned and filtered to obtain the required representative set of streamsurfaces. A similar result is available for incompressible velocity fields with a volume-preserving symmetry group \citep{haller98}. For general 3D incompressible flows, however, no first integrals will exist. This is because even steady 3D flows can be chaotic (see, e.g., \cite{dombre1986chaotic}) and hence cannot have smooth nontrivial conserved quantities. This is equally valid for the autonomous 3D differential equation defining the instantaneous streamlines of a generic 3D unsteady flow, excluding the existence of a regular foliation  of the flow domain by streamsurfaces that are level sets of a smooth scalar function.

Often, however, the most important parts of a flow turn out to be vortical regions filled with tubular or toroidal streamsurfaces. One cannot expect these surfaces to necessarily form a continuous family, especially the toroidal ones. Indeed, families of 2D tori composed of streamlines typically form Cantor sets (as opposed to a continuous family) in the 3D set of differential equations generating the streamlines \citep{cheng89}. While this is also the case for classic Hamiltonian systems admitting families of invariant tori \citep{arnold89}, those systems nevertheless turn out to be integrable restricted to this Cantor set of tori in an appropriate sense (see \cite{poschel82} and \cite{chierchia82}). Specifically, smooth functions can be constructed that act as first integrals over the Cantor family of tori but not in the small gaps among those tori.

Motivated by this idea of integrability over Cantor sets in Hamiltonian systems, we seek here smooth scalar functions that serve as approximate first integrals over a set of streamsurfaces forming vortical (elliptic) regions of a given vector field. As there is no widely accepted definition of a vortex, we use the term "vortical" loosely to describe families of toroidal or cylindrical surfaces to which either the velocity, vorticity or barrier field is tangent. The approximate first integrals arising from this procedure will be steady for steady vector fields and time-varying for unsteady vector fields. We construct these (approximate) first integrals by seeking scalar functions whose gradient vector field is as close to being normal to the given vector field as possible. In order to avoid the trivial solution to this problem, we use a constrained minimization approach that does not allow for globally constant first integrals.

Our method resembles that of \cite{yang2010lagrangian} for the construction of VSFs in inviscid, analytic and highly symmetric flows. Ours, however, differs crucially in that we work with a grid in the physical space over which we expand the unknown approximate first integral in a Fourier series.
We then only use the known values of the vector field at the gridpoints.
We find that this approach results in a homogeneous linear system of equations whose unique, unit-norm least-squares solution yields the unknown Fourier coefficients of the approximate first integral. Thus, in contrast to \cite{yang2010lagrangian}, this procedure is free from any symmetry assumptions on the first integral and does not require rewriting the homogeneous system as an inhomogeneous one under further assumptions. As we show in one of the appendices, these features of our approach significantly enhance the quality of the final solution.

Outside elliptic regions (i.e., in hyperbolic streamline domains), the streamlines are generically chaotic and hence will admit only trivial approximate first integrals. Accordingly, we expect the approximate first integrals to be nearly constant outside vortical regions, while admitting nontrivial shapes inside such regions. In those elliptic regions, level sets of the approximate first integrals will be close to streamsurfaces that form vortical features. 

The result of this approach is an automated numerical visualization method that does not require the user to guess seed points in the flow for streamsurfaces in vortical regions. The simplicity and generality of the proposed method allows us to employ it to complex flows defined either analytically or through numerical data. Our examples include spatially periodic integrable and non-integrable flows, a non-periodic vortex ring flow, a V-junction flow and a fully developed turbulent channel flow. The latter flow exemplifies a case wherein exact streamsurfaces tend to obscure the visualization of the most prominent features of the barrier field. Indeed, for such a flow vortical regions have been delineated via diagnostic tools such as the active version of the finite-time Lyapunov exponents (aFTLE), as described in \cite{haller2020objective}. Yet, tracking barrier streamlines originating in the neighborhood of aFTLE ridges quickly results in the streamsurface falling apart, despite some initial vortical motion. In contrast, the structures based on the approximate first integral are able to follow closely the valleys around the aFTLE ridges allowing for a better visualization.

To obtain these results, we divide the computational domain into smaller subdomains and seek approximate first integrals in each one of them separately. In this way, we can capture the most salient structures up to a given scale without the problem of obstruction by smaller structures. These in turn can be captured by further refining the domain subdivision based on their signatures in the aFTLE field. Further, upon assuming mild convexity for the vortical structures to be extracted, we obtain families of barrier surfaces whose outermost members act as vortex boundaries.
Thus, we propound a way out of the isocontour value dilemma that besets the typical vortex identification criteria used in the literature \citep{kim1987turbulence,hussain1995,zhou1999mechanisms}, causing them to produce dissimilar structures for different isocontour values.




\section{Set-up of the minimization scheme and reconstruction algorithm}
\label{algorithm_setup}
Let $\mathbf{v}(\mathbf{x},t)$ be a smooth vector field defined on some open subset $U\subseteq\mathbb{R}^{3}$. The associated dynamical system for the instantaneous streamlines of this vector field at time $t$ is given by
\begin{equation}
{\mathbf{x}}^\prime = \mathbf{v}(\mathbf{x},t),\quad \mathbf{x} \in U,
\label{dyn}
\end{equation}
where the prime denotes differentiation with respect to a curve-parameter $s\in{\mathbb R}$ along the streamline. A continuously differentiable, scalar function $H(\mathbf{x},t)$ is called an (\textit{instantaneous}) \textit{first integral} for $\mathbf{v}(\mathbf{x},t)$ at time $t$ if it is constant along each solution of \eqref{dyn}, i.e.,  $\frac{\partial}{\partial s} H(\mathbf{x}(s),t) = 0$. This condition implies that 
\begin{equation}
\nabla H(\mathbf{x},t)\cdot \mathbf{v}(\mathbf{x},t)=0,
\label{inner}
\end{equation}
for all $\mathbf{x}\in U$, where $\nabla$ denotes the gradient with respect to $\mathbf{x}$. 

For an arbitrary vector field $\mathbf{v}$, no first integral will exist, in general, \textit{pointwise}, i.e., for all $\mathbf{x}\in U$. We may, however, relax the constraint \eqref{inner} by seeking a function $ H(\mathbf{x},t)$ that minimizes the functional 
\begin{equation}
J[H]=\frac{1}{2}\int_{U}|\nabla H \cdot \mathbf{v}|^{2}dV,
\label{eq:closest_functional} 
\end{equation}
which measures the average deviation of $H(\mathbf{x},t)$ from being an exact, pointwise first integral in the domain $U$ at time $t$. Any minimizer of \eqref{eq:closest_functional} is called an \emph{approximate first integral}. 

First, let us assume that the domain $U$ is triply-periodic. This allows us to expand $H$ in a Fourier series and define its modal truncation of order $N$ as
\begin{equation}
H^{\leq N}(\mathbf{x})=\sum_{0<|\mathbf{k}|\leq N}\hat{H}_{\mathbf{k}}e^{i\mathbf{k}\cdot\mathbf{x}}.
\label{trunc}
\end{equation}
For the application to the examples in later sections, let us remark that the truncated expansion \eqref{trunc} can also be used locally on non-periodic domains, as long as we stay away from the domain boundaries. The integrand of the functional (\ref{eq:closest_functional}) in Fourier basis takes the form
\begin{equation}
\nabla H\cdot\mathbf{v}=\sum_{0<|\mathbf{k}|\leq N}\hat{H}_{\mathbf{k}}e^{i\mathbf{k}\cdot\mathbf{x}}\mathbf{k}\cdot\mathbf{v}.
\label{eq:integrand} 
\end{equation}
Because of the linearity of the gradient operator, the expression in eq.~(\ref{eq:integrand}) is linear in the unknown Fourier coefficients. 

We assume that the vector field $\mathbf{v}$ is known on a discrete, three-dimensional grid of points which we enumerate from $1$ through $m$, where $m$ is the total number of points. Also, to target data-driven applications specifically, we work with the discretized version of eq.~(\ref{eq:closest_functional}) using the vector of pointwise inner products defined as
\begin{equation}
\left[\begin{array}{cccc}
\nabla H_{1}\cdot\mathbf{v}_{1} & \nabla H_{2}\cdot\mathbf{v}_{2} & \cdots & \nabla H_{m}\cdot\mathbf{v}_{m}\end{array}\right]^{T}=\mathbf{C} \boldsymbol{\mathbf{h}},
\label{eq:closest_homo_system}
\end{equation}
with $\mathbf{C}\in{\mathbb C}^{m\times n}$. Here, $n$ is the number of modes used, $C_{ij}=e^{i\mathbf{k}_{j}\cdot\mathbf{x}_{i}}\mathbf{k}_{j}\cdot\mathbf{v}_{i}$ is the $(i,j)$ entry of $\mathbf{C}$
and $\boldsymbol{\mathbf{h}}=\left\{ \hat{H}_{\mathbf{k}}\left|\mathbf{k}\in\right.\mathbb{Z}^{3},\,\, 0<|\mathbf{k}|\leq N\right\} $ is the vector comprising the Fourier coefficients to be determined. Approximating a first integral through the functional \eqref{eq:closest_functional} then amounts to minimizing the squared vector norm  $|\mathbf{C}\boldsymbol{\mathbf{h}}|^2 $. To exclude the trivial solution $H\equiv const.$ from our analysis, we add the constraint
$|\boldsymbol{\mathbf{h}}|^2=1$. 

Solving this optimization problem
is equivalent to finding the eigenvector corresponding to the minimal eigenvalue of the symmetric matrix $\mathbf{A} = \mathbf{C}^{*}\mathbf{C}$, where $\mathbf{C}^*$ denotes the conjugate transpose of $\mathbf{C}$ (see Appendix \ref{app:least_squares_homo}). Since the eigenvectors of $\mathbf{A}$ are the right-singular vectors of $\mathbf{C}$, the solution to our algorithm can also be calculated from the singular-value decomposition (SVD) of $\mathbf{C}$.

We refrain from expanding the known vector field $\mathbf{v}$ into a Fourier series and we write down eq.~(\ref{inner}) explicitly for every point in the computational grid without working with the coefficients of each Fourier mode.
These features distinguish our method from the one presented in \cite{yang2010lagrangian}, as already noted in the Introduction. This difference will allow us to obtain unique solutions as well as apply our approach even to turbulent flows without additional assumptions.

In the absence of further constraints, the resulting approximate first integral $H$ will, generally, be a complex-valued scalar field. Denoting by $H_{r}$ and $H_{i}$ its real and imaginary parts, respectively, we have  $\langle \nabla H,\mathbf{v} \rangle_{\mathbb C} = \langle \nabla H_{r},\mathbf{v} \rangle - \langle \nabla H_{i},\mathbf{v} \rangle \mkern3mu i$.
This implies that by considering the expression in eq.~(\ref{eq:closest_homo_system}) we essentially optimize both the real and imaginary parts of $H$.
This allows us to use $\left| H \right|$ as the approximate first integral.
Alternatively, we can require $H$ to be real a priori. We discuss the latter procedure in Appendix \ref{app:realConstraint} and show that its results are similar to those obtained without imposing this constraint.

As already noted, for generic flows, a non-trivial approximate first integral is only expected to exist in vortical regions. Outside such regions, we expect our algorithm to yield almost flat first integrals. Such a function, locally supported on several vortical regions, would generally require a large number of Fourier basis functions, which in turn would lead to numerical inefficiencies and cost. To avoid this, we work with a comparably small number of basis functions and only use the level surfaces of the emerging approximate first integral in regions where those surfaces are indeed nearly tangent to the given vector field. To identify such regions, we introduce the {\em invariance error} as
\begin{equation}
E_{A} =\frac{1}{A} \sum_{i=1}^{p} \left| \frac{\nabla \left| H_{i} \right|  \cdot \mathbf{v}_{i}}{\left|\nabla \left|H_{i}\right| \right| \left| \mathbf{v}_{i} \right|} \right|,
\label{eq:norm_error_surface} 
\end{equation} 
where $A$ is the  surface area of a level set and $p$ the number of points in the level set. This type of error estimate was first introduced in \cite{yang2010lagrangian}. In our visualizations, we will discard streamsurface candidates with invariance errors exceeding a certain threshold value.

Finally, we note that our minimization procedure can also be viewed as finding for $\mathbf{v}$, in the appropriate norm,  the closest member of the integrable, 3D incompressible vector field family
\begin{equation}
{\mathbf{x}}^\prime = \mathbf{J}(\mathbf{x})\nabla H(\mathbf{x},t),\quad \nabla\cdot \mathbf{J}\equiv\mathbf{0},
\label{closest_integrable_velocity_field}
\end{equation}
with $\mathbf{J}=-\mathbf{J}^T$ and $\nabla\cdot \mathbf{J}$ denoting the divergence of the tensor field $\mathbf{J}$. Indeed, all these vector fields in eq.~(\ref{closest_integrable_velocity_field}) share the same streamsurfaces, the level sets of $H$. Working with  (\ref{closest_integrable_velocity_field}) directly, however, is much more demanding  numerically in our experience and would also require specific assumptions on the form of $\mathbf{J}$.

Before moving to specific examples, we note that finding an exact, pointwise first integral in eq.~(\ref{inner}) is not a well-posed problem by itself. Indeed, if $H(\mathbf{x},t)$ is a solution, then, for any sufficiently smooth function $F$, $F(H(\mathbf{x},t))$ will also be a solution due to the homogeneity of eq.~(\ref{inner}). This would not be an issue for the detection of streamsurfaces, if the isocontours of $H$ and $F(H)$, expressed through the gradient of these fields, represent the same geometric and topological features for the given vector field $\mathbf{v}$. Unfortunately, we can construct simple counter-examples where this is not the case. For instance, if we denote by $v_x, v_y$ and $v_z$ the three components of $\mathbf{v}$ and assume that $v_x = 0$, then the function $G(x) H(y,z)$ will be an exact, pointwise first integral as long as $v_y \theta_y H + v_z \theta_z H = 0$. By tweaking $G$, one can easily obtain markedly different topological features resulting from the corresponding streamlines. We also refer to \cite{pullin2014whither} for more examples of first integrals with different topology for the vorticity field of the Taylor--Green flow stemming from the superposition of independent solutions to eq.~(\ref{inner}).

Consequently, our variational principle (\ref{eq:closest_functional}) will exhibit the same non-uniqueness issues whenever an exact, pointwise first integral is admitted by the underlying vector field. To resolve this, we will only consider approximate first integrals for which the weakest eigenvalue of $\mathbf{A}=\mathbf{C}^*\mathbf{C}$ is not considered (numerically) zero.
In addition, we will only retain the streamsurfaces whose topology remains unaltered when, for the same set of grid points, a larger number of Fourier modes is used allowing for small geometric refinements due to the increased accuracy of the solution.
If these two conditions are met, we will consider the resulting structures as robust and they will be included in the visualization.
Furthermore, we will see that the more complex a flow is the larger the spectral gap to the second-smallest eigenvalue will be.
Irrespective of this gap, however, in all the examples that follow we will build our solution based only on the eigenvector corresponding to the smallest eigenvalue of $\mathbf{A}$.
We close this section by noting that our approach is in agreement with the findings of \cite{xiong2017boundary,xiong2019identifying} who showed that the construction of VSFs in turbulent flows through the use of PDEs leads to robust structures despite the non-uniqueness issues.

\section{Approximate first integrals for explicit solutions of the Euler equations}
In this section, we illustrate our minimization algorithm to construct approximate first integrals and use their level sets as approximate streamsurfaces in analytic examples.

\subsection{ABC flow}
\label{ABC_Flow}
As a first test case, we investigate the ABC (Arnold--Beltrami--Childress) class of flows \citep{dombre1986chaotic,henon1966topologie}, defined as
\begin{equation}
\begin{gathered}
\displaystyle \dot{x} = A \sin z + C \cos y, \\
\displaystyle \dot{y} = B \sin x + A \cos z, \\
\displaystyle \dot{z} = C \sin y + B \cos x,
\label{eq:ABC} 
\end{gathered}
\end{equation}
for $A,B,C\in\mathbb{R}$ and $(x,y,z) \in [0,2\pi]^3$, together with periodic boundary conditions. The right-hand side of \eqref{eq:ABC} defines an exact steady solution to the incompressible Euler equations. For $ABC \neq 0$, the flow exhibits chaotic behavior \citep{dombre1986chaotic,henon1966topologie}, whereas some analytic non-integrability results can be found in \cite{ziglin1988splitting,ziglin1998absence}.

\subsubsection{Integrable case}
We first analyze the ABC flow with $A=0$ for which \eqref{eq:ABC} is completely integrable. For $BC \neq 0$, an exact, pointwise first integral is given by $H_1(x,y) = C \sin y + B \cos x$, while another independent first integral can be constructed through the use of elliptic functions \citep{llibre2012note}. The level sets of $H_1$ are depicted in Fig. \ref{fig:closest_ABC_integrable}(a) on one cross-section of the computational domain for $B=\sqrt{2}$ and $C = 1$. These curves, therefore, represent the intersections of a representative set of streamsurfaces with the $z=0$ plane. These streamsurfaces are also exact invariant manifolds for the Lagrangian particle motions in this steady flow.

To test our proposed algorithm for constructing approximate first integrals, we start from a discretized version of the full 3D velocity field (\ref{eq:ABC}) along $100$ points per direction, using approximately $9\,000$ Fourier modes in formulas (\ref{trunc})-(\ref{eq:closest_homo_system}).  Solving the underlying optimisation problem using the algorithm in Appendix \ref{app:least_squares_homo}, we obtain the results shown in Fig.~\ref{fig:closest_ABC_integrable}(b) at the same cross-section as in Fig.~\ref{fig:closest_ABC_integrable}(a). The numerically constructed approximate level sets match the analytic first integral closely. To measure the proximity of streamsurfaces and approximate streamsurfaces along the $z=0$ plane, we introduce a planar version of the general invariance error (\ref{eq:norm_error_surface}) by defining
\begin{equation}
E_{l} =\frac{1}{l} \sum_{j=1}^{p} \left| \frac{\nabla \left| H_{j} \right|  \cdot \nabla H_{1j}}{\left| \nabla \left|H_{j}\right| \right| \left| \nabla H_{1j} \right|} \right|
\label{eq:norm_error_curve} 
\end{equation} 
where $l$ is the length of the level and $p$ is the number of points for each level set. We use this metric to remove level curves with fewer than $30$ points and those with invariance errors $E_{l} > 10^{-5}$. The results shown in Fig. \ref{fig:closest_ABC_integrable}(c) confirm that choosing these thresholds removes small-scale artifacts arising from numerical inaccuracies.

\begin{figure}
	\centering
	\includegraphics[]{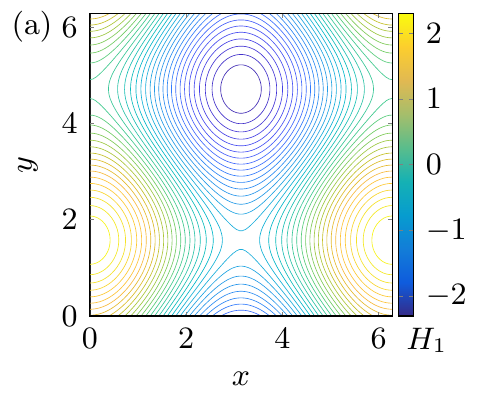}\hfill\includegraphics[]{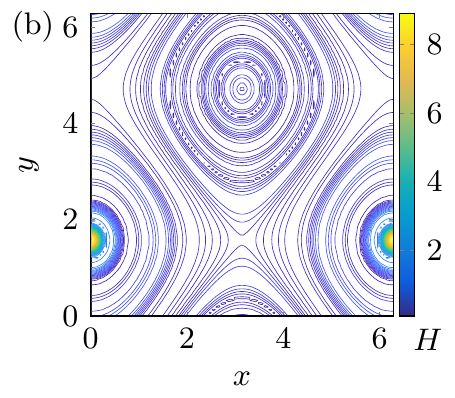}\hfill\includegraphics[]{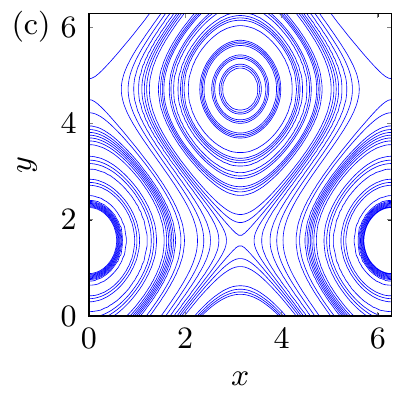}
	\caption{Analysis of the integrable ABC flow using a computational grid of $100^3$ points and about $9\,000$ Fourier modes. (a) Intersections of the level surfaces  of $H_1$ with the $z=0$ plane (b) Intersections of the level surfaces of the approximate first integral $H$  with the $z=0$  plane. (c) Same as (b) but after the removal  of  small-scale structures of (b) as well as the structures with $E_{l} > 10^{-5}$.}
 	\label{fig:closest_ABC_integrable}
 \end{figure}

We perform the same analysis on the $z=0$ plane using $150$ points per direction but only about $1\,200$ Fourier modes. The results are depicted in Fig.~\ref{fig:closest_ABC_integrable_2}. Again, we observe close agreement between the known first integral and the reconstructed curves. At some points, the agreement is even closer when compared to Fig.~\ref{fig:closest_ABC_integrable}(c) due to the higher spatial resolution, even though the number of Fourier modes is significantly smaller.

\begin{figure}
	\centering
	\includegraphics[]{./figs/combinedFigures/IntegrableABC/ABC_integral_gs.pdf}
	\hfill\includegraphics[]{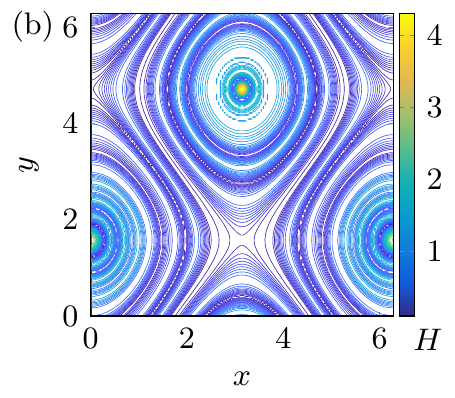}
	\hfill\includegraphics[]{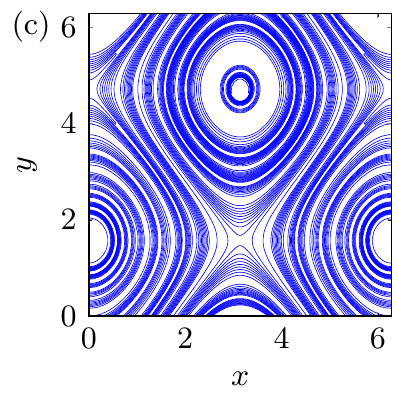}
	\caption{Same as Fig.~\ref{fig:closest_ABC_integrable} but with  a computational grid of $150^3$ points and  $1\,200$ Fourier modes.}
	\label{fig:closest_ABC_integrable_2}
 \end{figure}

\subsubsection{Non-integrable case}
\label{nonintegrable_ABC}
For a different set of parameter values ($A = \sqrt{3}$, $B = \sqrt{2}$ and $C = 1$), the ABC flow \eqref{eq:ABC} is non-integrable and shows chaotic behavior in some regions. The dynamic behavior of trajectories for this set of parameter values is well-studied, including the KAM-type tori highlighted by Poincaré maps \citep{dombre1986chaotic} and elliptic LCS techniques \citep{blazevski2014hyperbolic,oettinger2016global}. We use this velocity field as a benchmark  to test different solution algorithms for finding an approximate first integral for a non-integrable flow. Also, this will serve as a proof of concept for finding elliptical regions. 

We have already noted that the unit-norm least-squares solution to the homogeneous system of eq.~(\ref{eq:closest_homo_system}) coincides with the right-singular vector of $\mathbf{C}$ associated to its smallest singular value or, equivalently, with the eigenvector associated to the smallest eigenvalue of $\mathbf{A} = \mathbf{C}^{*}\mathbf{C}$. To improve numerical stability, the SVD-based solution  is preferred \citep{golub1973differentiation}. Indeed, the eigenvalue calculation requires a matrix multiplication to form $\mathbf{A}$, which invariably squares the condition number of $\mathbf{C}$. For comparison, we calculate both the singular-vector-based and the eigenvector-based solutions on the triply-periodic box $[0,2 \pi]^3$ with $100$ points per direction and $N=13$ (or $9\,170$ Fourier modes). In this setting, running the SVD algorithm of MATLAB on $\mathbf{C}$, which is a $(10^2)^3 \times 9\,170$ matrix, would require an exorbitant amount of memory (more than $400$ GB), indicating that the classical SVD algorithm is not optimized for tall-skinny matrices (such as our coefficient matrix). To proceed with our comparison test, we instead follow the modified SVD method discussed in Appendix \ref{app:skinnySVD} for tall-skinny matrices.

With this modification to the SVD-based solution, the results from the two approaches for the non-integrable ABC flow are presented in Fig.~\ref{fig:closest_ABC_eigen_vs_psvd} on the $z=0$, $y=0$ and $x=0$ planes. We observe that the differences between the eigenvector-based and  singular-vector-based computations are marginal, indicating that the larger condition number of $\mathbf{A}$ does not affect the results. Furthermore, we use the same three planes as Poincaré sections to integrate trajectories up to an arclength of $10^4$ from a uniform grid of $20 \times 20$ initial conditions on each section. We observe a very good agreement between the predicted structures and the intersections of the KAM surfaces with each of these sections. This is highlighted perhaps even better by the reconstructed KAM surfaces as approximate streamsurfaces in Fig.~\ref{fig:IsoHoudini_ABC}, which are to be contrasted with the $\lambda_2$-based structures in Appendix \ref{app:faultyCases}. Since the ABC flow is a Beltrami flow, its velocity $\mathbf{v}$ is parallel to the vorticity $\boldsymbol{\omega} = \boldsymbol{\nabla} \times \mathbf{v}$ (in fact, $\mathbf{v}=\boldsymbol{\omega}$); consequently, the approximate-first-integral-based tori we have constructed are also VSFs. This illustrates that our algorithm can identify VSFs in flows where the methodology of \cite{yang2010lagrangian} is inapplicable. Indeed, as already noted, the non-integrable ABC flow has chaotic streamlines and, thus, no symmetry assumptions regarding these streamlines can be utilized to accelerate the convergence rate for the optimization technique presented in \cite{yang2010lagrangian}. Even if this rate was irrelevant, however, expanding the known velocity field in a Fourier series would result in an optimization problem with many (numerically) zero eigenvalues and, thus, infinitely many possible minimizers.

\begin{figure}
	\centering
	\includegraphics[]{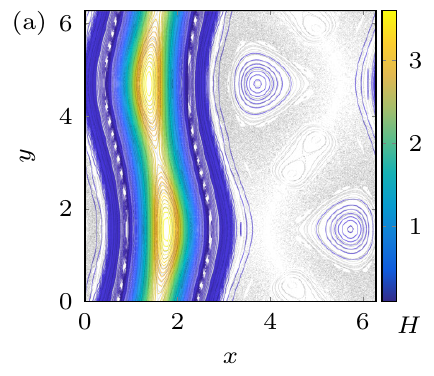}
	\includegraphics[]{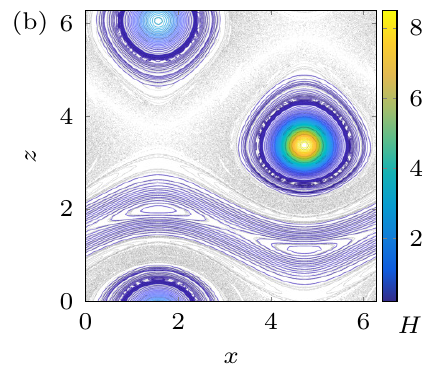}
	\includegraphics[]{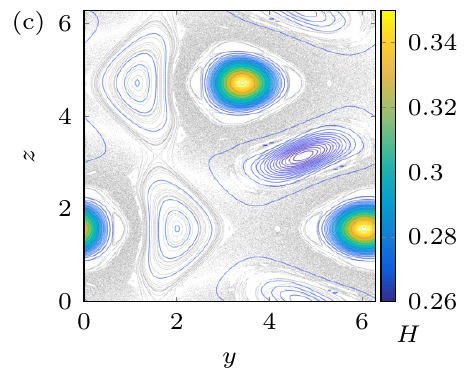}\\
	\includegraphics[]{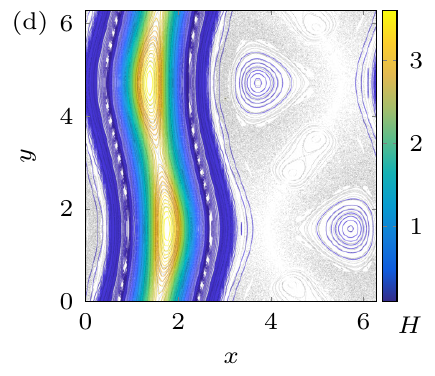}
	\includegraphics[]{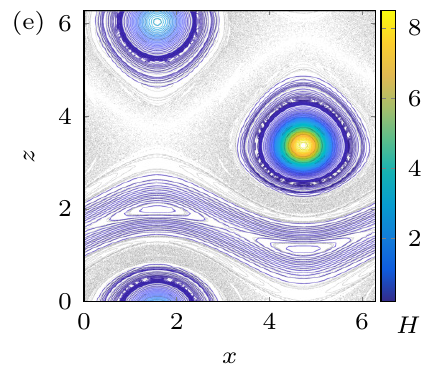}
	\includegraphics[]{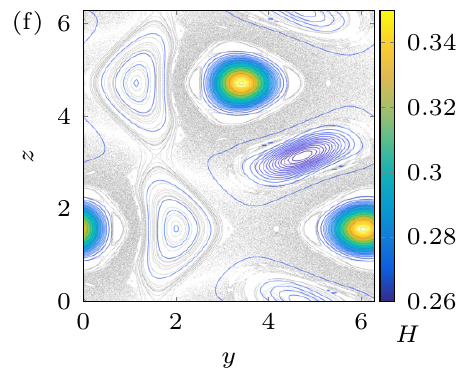}
	\caption{Analysis of the non-integrable ABC flow using a computational grid of $100^3$ points and $9\,170$ Fourier modes. Level sets of the approximate first integral at $z=0$ (a and d), $y=0$ (b and e) and $x=0$ (c and f). The first row is constructed from the eigenvector of $\mathbf{A}$ corresponding to the smallest eigenvalue, whereas the second row is produced using the SVD of $\mathbf{C}$. The overlaid Poincaré map (black dots) on each section is based on a uniform grid of $20 \times 20$ initial conditions.}
	\label{fig:closest_ABC_eigen_vs_psvd}
 \end{figure}
\begin{figure}
	\centering
	\includegraphics[width=.48\linewidth]{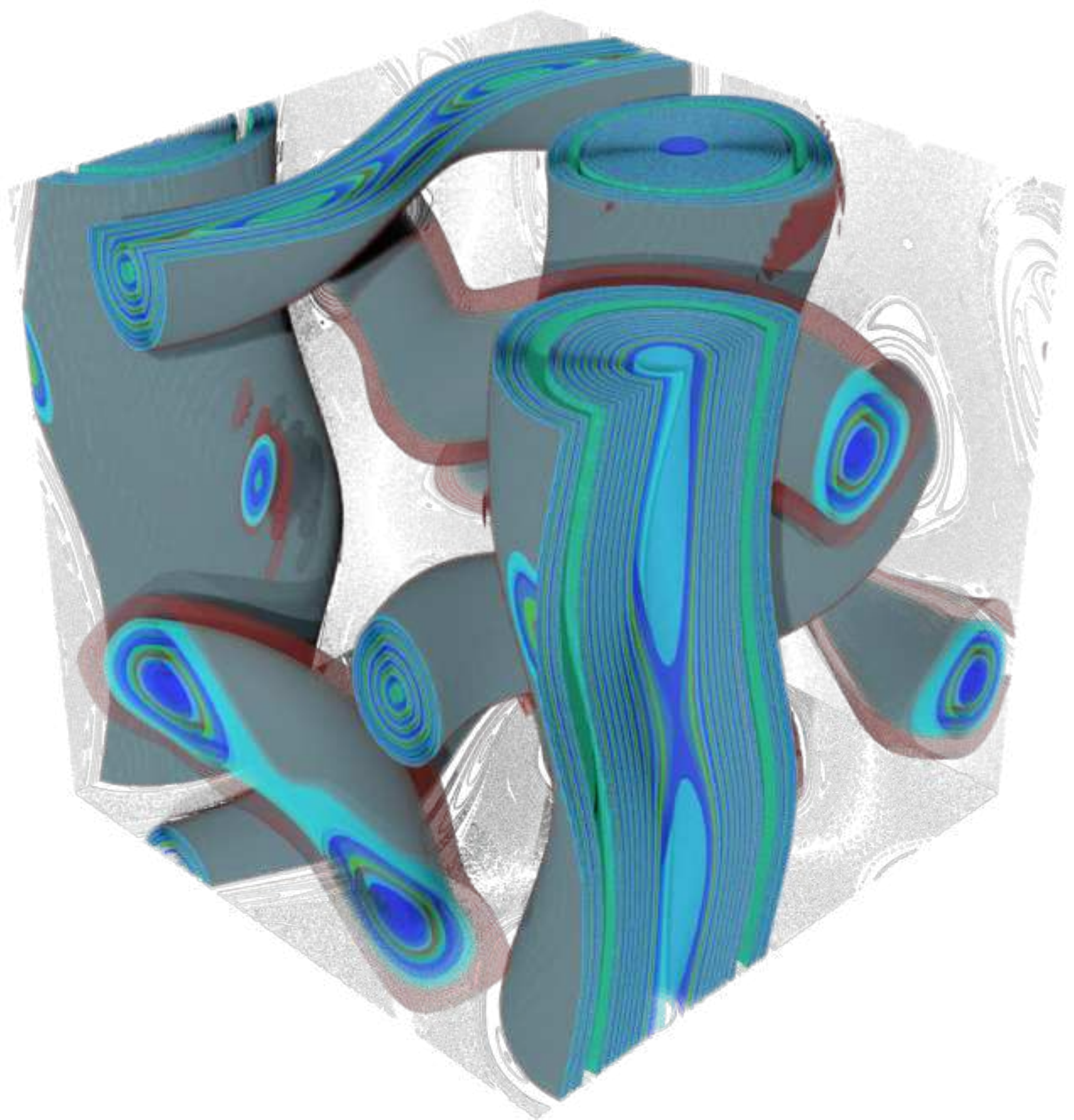}
	\includegraphics[width=.48\linewidth]{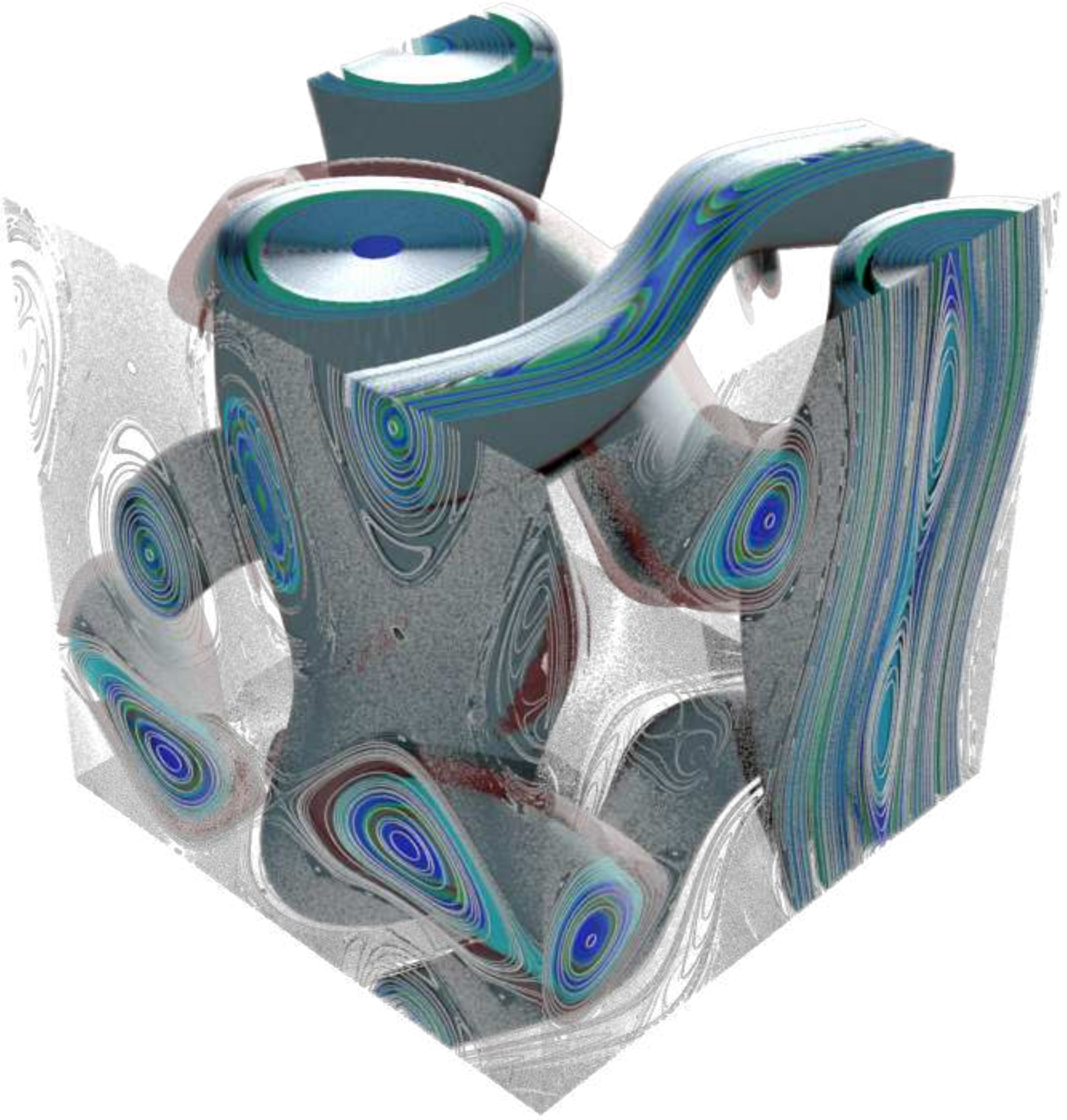}
	\vspace{-0.2cm}
	\caption{Two different views of the approximate streamsurfaces (level sets of the approximate first integral) closely approximate the KAM-type surfaces of the non-integrable ABC flow in elliptic regions. Also shown are iterations of the  Poincaré map (black dots) on three orthogonal planes. The results were obtained using the weakest eigenvector of the positive definite matrix $\mathbf{A}$.  See also the supplementary Movie 1.mp4.}
	\label{fig:IsoHoudini_ABC}
\end{figure}
Upon taking a closer look at the results of Fig.~\ref{fig:closest_ABC_eigen_vs_psvd}, we notice that the reconstructed level sets attain their values in a longer range (i.e., $[0,3.5]$) for the larger KAM surfaces, whereas, in the vicinity of the smaller structures, they are confined to a narrow band (i.e., $[0.25,0.35]$). Here the adjectives larger and smaller are used to refer to either the area (Fig.~\ref{fig:closest_ABC_eigen_vs_psvd}) or the volume (Fig.~\ref{fig:IsoHoudini_ABC}) these families enclose. This is a type of overfitting that we would like to mitigate. One way to achieve this is by considering a slightly different optimization problem (see Appendix \ref{app:inhomoSol}) which resembles the one put forward by \cite{yang2010lagrangian}. This different approach, however, turns out to be computationally intense, posing severe limitations to its use for typical grid sizes while its results are arguably of inferior quality. All in all, obtaining the solution to the proposed algorithm as the eigenvector corresponding to the smallest eigenvalue of $\mathbf{A}$ is computationally superior to all the other techniques used and, thus, it is the one that we will follow in the rest of this article.

We conclude this section by performing a convergence analysis for different numbers of modes in Fig.~\ref{fig:closest_ABC_spectrum}. We observe that the least-squares error (as the smallest eigenvalue of $\mathbf{A}$) approaches zero as the number of modes used in the analysis increases.
Furthermore, the second smallest eigenvalue converges to the smallest one for higher modes.
In Appendix \ref{app:least_squares_homo}, we show that when the smallest eigenvalues are almost equal, we can construct the solution as a linear combination of the eigenvectors corresponding to these near-identical eigenvalues. Here, however, we can base our solution only on the weakest eigenvector since we have $d_1 = 0.015415$ and $d_2=0.01596$ for $N=13$. This gap between $d_1$ and $d_2$ will prove significantly larger in the following, data-based examples, confirming the uniqueness of solutions to the optimization problem.
\begin{figure}
	\begin{centering}
		\includegraphics[scale=1.1]{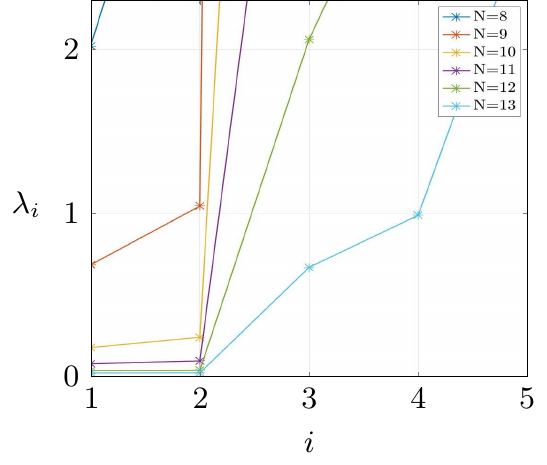}
		\par\end{centering}
		\vspace{-0.3cm}
	\caption{Five smallest eigenvalues of $\mathbf{A} = \mathbf{C}^{*}\mathbf{C}$ for different numbers of modes ($2\,108, 3\,070, 4\,168, 5\,574, 7\,152, 9\,170$ modes for $N=8,9,10,11,12,13$, respectively).}
	\label{fig:closest_ABC_spectrum}
\end{figure}

\subsection{Further analytic solution to the Euler equations}
A set of analytic, unsteady, tri-periodic laminar solutions of the Navier–Stokes equations was put forward in \cite{antuono2020tri}. Here, we consider only the steady part of these solutions, which is a Beltrami solution to the Euler equations with no known first integral. We therefore expect the streamlines of this velocity field to be chaotic and the overall dynamics to be non-integrable. Representative (approximate) streamsurfaces have not yet been constructed for this flow in the literature.

The velocity field is given by
\begin{equation}
\mathbf{v}=\frac{4\sqrt{2}}{3\sqrt{3}}\left(\begin{array}{c}
\sin(x-\frac{5\pi}{6})\cos(y-\frac{\pi}{6})\sin(z)-\cos(z-\frac{5\pi}{6})\sin(x-\frac{\pi}{6})\sin(y)\\
\sin(y-\frac{5\pi}{6})\cos(z-\frac{\pi}{6})\sin(x)-\cos(x-\frac{5\pi}{6})\sin(y-\frac{\pi}{6})\sin(z)\\
\sin(z-\frac{5\pi}{6})\cos(x-\frac{\pi}{6})\sin(y)-\cos(y-\frac{5\pi}{6})\sin(z-\frac{\pi}{6})\sin(x)
\end{array}\right).
\label{eq:Nflow_def}
\end{equation}
Using a uniform grid of $20 \times 20$ initial conditions on the $y=0$ plane, we integrate trajectories up to arclength of $10^4$ and, upon retaining their long-term behavior from the interval $[5 \cdot 10^3, 10^4]$, we show the resulting Poincaré map in Fig.~\ref{fig:closest_NFlow_Poin_N13_percentiles}(a) where a plethora of KAM tori is discernible. Selecting the triply periodic box $[0,2 \pi]^3$ and sampling it with $100^3$ points, we run our algorithm for $N=13$. Upon constructing isosurfaces for $10$ different isovalues of the resulting approximate first integral, we locate their intersections with the $y=0$ plane and superimpose them on Fig.~\ref{fig:closest_NFlow_Poin_N13_percentiles}(a).
We then compute the invariance error based on eq.~(\ref{eq:norm_error_surface}) and color the isocontours of Fig.~\ref{fig:closest_NFlow_Poin_N13_percentiles}(a) blue or red depending on whether their average error corresponds to an angle of less or more than $5^{\circ}$. Based on this, contours lying inside the chaotic sea of the Poincaré map show the largest invariance error despite our algorithm not using any knowledge of the long-term, chaotic dynamics.

For every isocontour of Fig.~\ref{fig:closest_NFlow_Poin_N13_percentiles}(a), we launch trajectories following the vector field $\mathbf{v}$ of eq.~(\ref{eq:Nflow_def}) until they leave the computational domain $[0,2\pi]^3$, and then calculate their distances from the corresponding isosurfaces. Thus, in Fig.~\ref{fig:closest_NFlow_Poin_N13_percentiles}(b), we present the 95th percentile of those distances for each streamsurface after separating them according to Fig.~\ref{fig:closest_NFlow_Poin_N13_percentiles}(a), i.e., we depict the trajectories corresponding to the blue contours in the left side of the dash-dotted line, whereas the ones for the red contours to the right side. We note the correlation between the retained (blue) isocontours of Fig.~\ref{fig:closest_NFlow_Poin_N13_percentiles}(a) and the significantly smaller percentiles of Fig.~\ref{fig:closest_NFlow_Poin_N13_percentiles}(b).

\begin{figure}
	\centering
	\includegraphics[width=.4\textwidth]{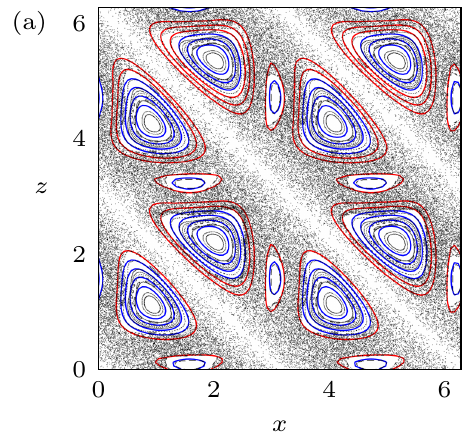}
	\includegraphics[width=.435\textwidth]{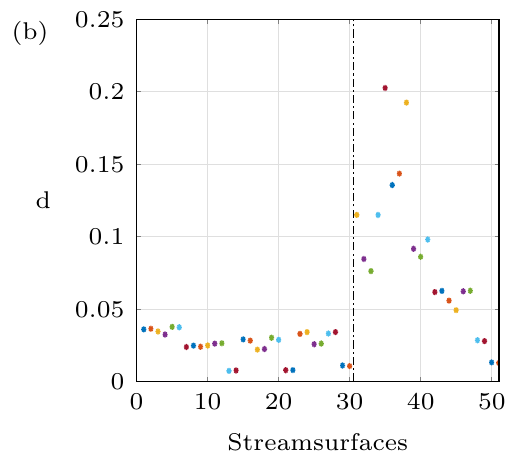}
		
	\vspace{-0.3cm}
	\caption{(a) Comparison of the Poincaré map on the plane $y=0$ for the steady Euler flow (\ref{eq:Nflow_def}) overlaid on the intersections of the tori obtained from an approximate first integral with the same plane for $N = 13$.
	The blue (red) isocontours depict tori whose  invariance error (see eq.~(\ref{eq:norm_error_surface})) is smaller (larger, respectively) than $5^{\circ}$ on average.
	(b) 95th percentile of the distance (in non-dimensional units) between trajectories emanating from the isocontours of (a) with the corresponding 2D tori inside the computational box $[0,2\pi]^3$. The points to the left (right) of the dash-dotted line correspond to trajectories originating in the blue (red) contours of (a).}
	\label{fig:closest_NFlow_Poin_N13_percentiles}
\end{figure}

Similarly, we present the reconstructed tori for $N = 15$ or $14\,146$ modes in Fig.~\ref{fig:closest_NFlow_Poin_N15_N19_closeup}(a). We observe that, for $N=13$, the approximate first integral captures virtually all the KAM surfaces indicated by the Poincar{\' e} map. In contrast, for $N=15$, some of the structures are captured more accurately while others are missed completely.
The convergence analysis depicted in Fig. \ref{fig:closest_NFlow_spectrum_convergence}(a) shows that the least-squares error follows a declining trend as $N$ grows. 
This prompts us to consider an error measure similar to the one in eq.~(\ref{eq:norm_error_surface}), defined as
\begin{equation}
E_m=\frac{1}{m}\sum_{i=1}^{m}\left| \frac{\nabla \left| H_{i} \right| \cdot \mathbf{v}_{i}}{\left| \nabla \left| H_{i} \right|\right| \left| \mathbf{v}_{i}\right|  } \right|.
\label{eq:normalized_error} 
\end{equation}
In this expression, the summation is taken over all the grid points, providing an estimate for the mean invariance error of the entire solution. This allows us to make a direct comparison among solutions corresponding to different numbers of modes.
Specifically, excluding points that lie in the vicinity of fixed points of either
$\mathbf{v}$ or $\nabla \left| H \right|$, we show the dependence of the invariance error $E_m$ on the number of Fourier modes used in our algorithm in Fig. \ref{fig:closest_NFlow_spectrum_convergence}(b). We observe that the error attains the minimum value for $N=13$, in agreement with what is inferred from Fig.~\ref{fig:closest_NFlow_Poin_N13_percentiles}(a).

Moreover, to mitigate minor discrepancies between the reconstructed tori and those outlined by the Poincaré map of Fig.~\ref{fig:closest_NFlow_Poin_N13_percentiles}(a), we increase the number of Fourier modes to $28\,670$ ($N=19$), while keeping the same set of grid points. A closeup view of a region filled with invariant tori in Fig.~\ref{fig:closest_NFlow_spectrum_convergence} confirms that there is a close agreement with the tori obtained from an approximate first integral. There is, however, a trade-off between the increased accuracy and the ensuing computational burden that an end user must consider.

We also note that, while for $N=15$ we have $d_1 \approx d_2 = 839.8079$, for $N=19$ we have $d_1 = 204.5165$ and $d_2 = 218.2078$, further corroborating that our approach leads to a unique solution. Finally, we conclude this section with Fig.~\ref{fig:isoHoudini_NFlow} showing 3D rendered images of (a) the three Poincaré maps on the planes $x=y=z=0$, (b) representative streamsurfaces obtained as level surfaces of an approximate first integral for $N=13$ and (c) the superimposition of (a) on (b),  confirming the close agreement between the expected and reconstructed structures.
\begin{figure}
	\centering
	\includegraphics[width=.4\textwidth]{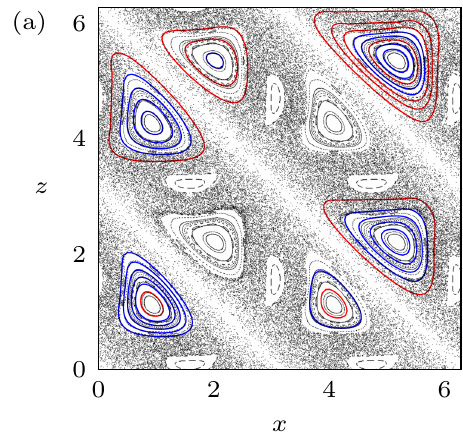}
	\includegraphics[width=.4\textwidth]{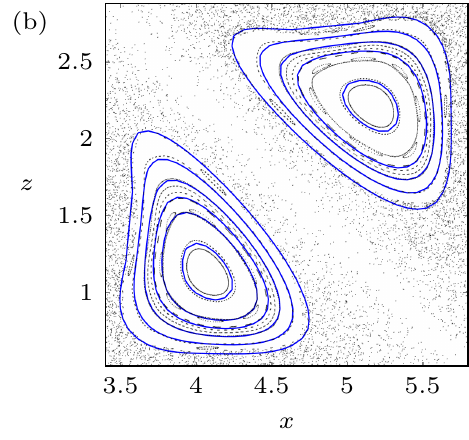}
		
	\vspace{-0.3cm}
	\caption{(a) Same as in Fig.~\ref{fig:closest_NFlow_Poin_N13_percentiles}(a) with the tori reconstructed for $N = 15$.
	(b) Closeup view on a region filled with two families of invariant tori. Overlaid on the Poincaré map are the tori obtained from an approximate first integral for $N = 19$.}
	\label{fig:closest_NFlow_Poin_N15_N19_closeup}
\end{figure}
\begin{figure}
	\centering
	\includegraphics[]{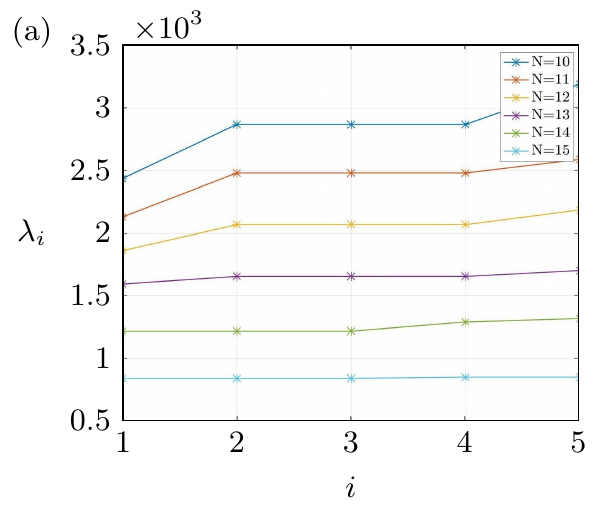}
	\includegraphics[]{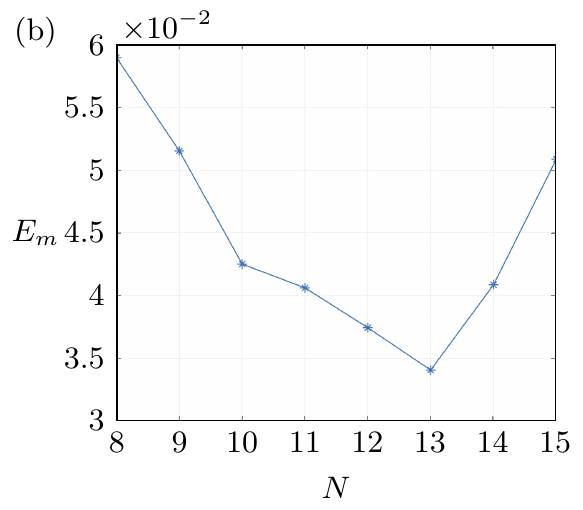}
	\vspace{-0.28cm}
	\caption{Numerical details of the approximate first integral calculation for the steady Euler flow (\ref{eq:Nflow_def}). (a) Five smallest eigenvalues of $\mathbf{A} = \mathbf{C}^{*}\mathbf{C}$ and (b) normalized error estimate (eq.~(\ref{eq:normalized_error})) for different numbers of modes.}
	\label{fig:closest_NFlow_spectrum_convergence}
\end{figure}
\begin{figure}
	\centering
		\centering
		\includegraphics[width=.32\linewidth]{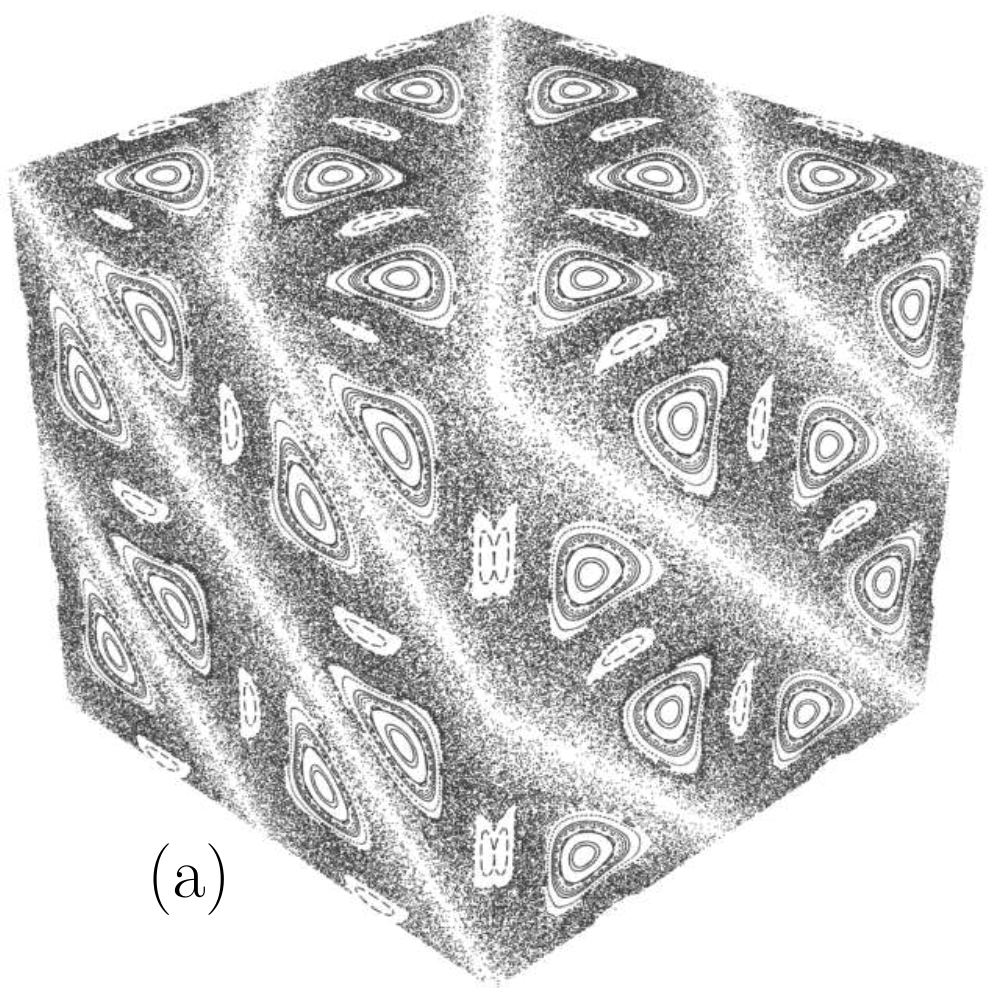}
		\includegraphics[width=.32\linewidth]{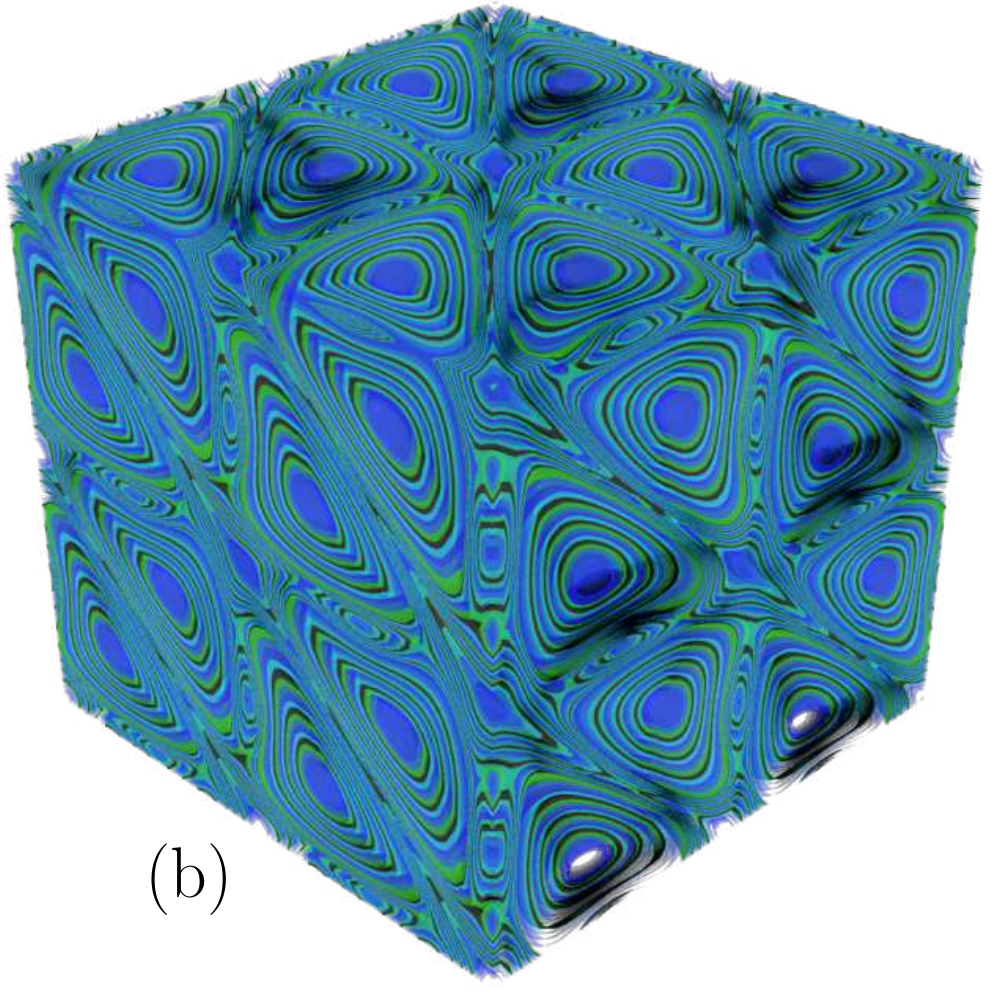}
		\includegraphics[width=.32\linewidth]{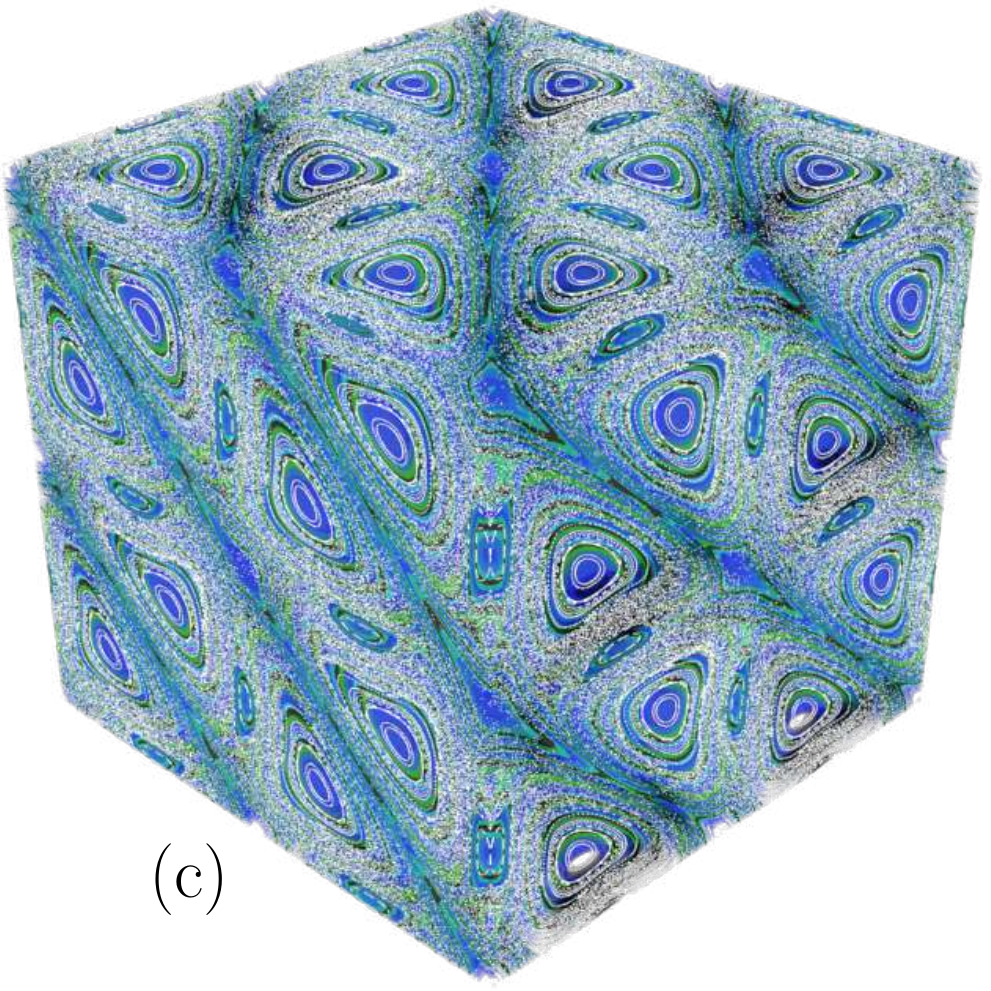}
	
	\vspace{-0.2cm}
	\caption{Results for the steady Euler flow (\ref{eq:Nflow_def}). (a) Poincaré maps on $x=y=z=0$. (b) Streamsurfaces approximating the KAM surfaces of eq.~(\ref{eq:Nflow_def}). (c) (a) superimposed on (b).}
	\label{fig:isoHoudini_NFlow}
\end{figure}

\subsection{Hill's spherical vortex}
\label{Hills_vortex}
We now turn to a spatially non-periodic, integrable flow given by Hill's spherical vortex \citep{hill1894vi}. The axial-symmetric (about the $z$-axis) stream function for Hill's solution to the Euler equations is
\begin{equation}
\psi(h,z)=\begin{cases}
\begin{array}{c}
\frac{3}{4}U_{0}h^{2}\left(1-r^{2}\right),\quad r\leq1,\\
-\frac{1}{2}U_{0}h^{2}\left(1-\frac{1}{r^{3}}\right),\quad r>1.
\end{array}\end{cases}\label{eq:def_psi}
\end{equation}
where $h$ is the distance from the axis of symmetry ($h^2 = x^2 +y^2$) and
$r^{2}=h^{2}+z^{2}$. Using the relations $u_{h}=-\frac{1}{h}\frac{\partial\psi}{\partial z}$ and $u_{z}=\frac{1}{h}\frac{\partial\psi}{\partial h}$, we obtain the corresponding velocity field in Cartesian coordinates
\begin{equation}
\displaystyle
\mathbf{v}(x,y,z)=\begin{cases}
\begin{array}{c}
\frac{3}{2}U_{0} \left(xz, \quad yz, \quad 1-(r^2+h^2)\right) ,\quad r\leq1,\\
\frac{3}{2}U_{0} \left(\frac{xz}{r^5}, \quad \frac{yz}{r^5}, \quad \frac{2}{3} \frac{1}{r^3} - \frac{h^2}{r^5}\right),\quad r>1.
\end{array}\end{cases}\label{eq:def_vel_Hill}
\end{equation}

As for periodic domains, we can compute a Fourier expansion on a bounded subdomain, bearing in mind that the convergence of the partial Fourier sum will be slow in general \citep{gottlieb1997gibbs}. To mitigate this issue, we will have to consider a sufficiently high number of modes. Furthermore, due to the well-known Gibbs phenomenon, there will be sizable spurious oscillations in the approximate first integral near the box boundary that do not diminish after an increase in the number of modes. To damp this effect in our algorithm, we discard all the reconstructed surfaces that fall within $5\%$ of the box size in all three directions.

The velocity field of eq.~(\ref{eq:def_vel_Hill}) has toroidal streamsurfaces inside the spherical domain $\{\mathbf{x} \in \mathbb{R}: \left|\mathbf{x}\right| \leq 1\}$ (see Fig.~\ref{fig:Hills_tori_per}(a)). We reconstruct these streamsurfaces over the computational domain $[-2,2]^3$ with $60$ points per direction, running our algorithm for different numbers of Fourier modes with $U_{0}=1$. In some cases with $N<15$, the non-periodic nature of the velocity field results in the breakdown of the reconstructed toroidal structures. For $N\ge15$, however, we obtain fully symmetric solutions. Specifically, for illustration purposes, we work with $N=17$ (or $20\,478$ Fourier modes) which yields a unique solution corresponding to the smallest eigenvalue of $\mathbf{A}$, $d_{1} = 23.5169$ ($d_2 = 25.3118$). Based on this we create $10$ isosurfaces in the interval $\big[ |H_{min}|,|H_{max}| \big]$. Upon locating the intersections of these surfaces with eight radially equidistant planes, we launch trajectories corresponding to these intersections and compute the pointwise distance of the solution curves to the reconstructed surfaces. The results (in terms of percentiles) are presented in Fig.~\ref{fig:Hills_tori_per}(b). A comparison between the reconstructed streamsurfaces and indicative solution curves is given in Fig.~\ref{fig:Hills_tori_3D}.

\begin{figure}
	\centering
		\centering
		\includegraphics[width=.49\linewidth]{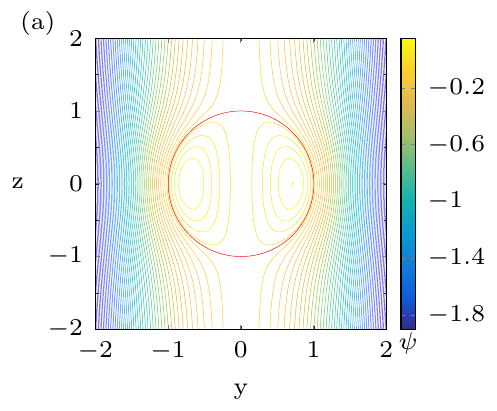}
		\includegraphics[width=.48\linewidth]{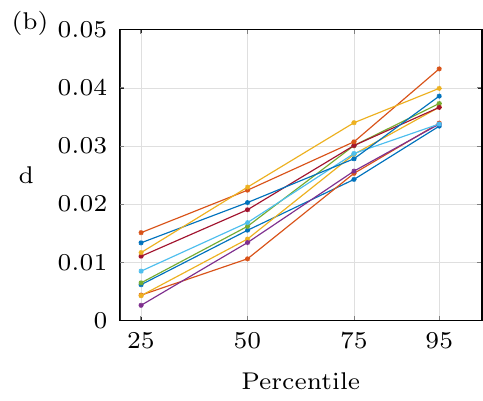}

	\vspace{-0.3cm}
	\caption{(a) Level sets of Hill's stream function on the $x=0$ plane. (b) 25th, 50th, 75th and 95th percentile of the pointwise distances (in non-dimensional units) between solution curves and $10$ different reconstructed streamsurfaces for $N=17$.}
	\label{fig:Hills_tori_per}
\end{figure}

\begin{figure}
	\centering
		\centering
		\includegraphics[width=.49\linewidth]{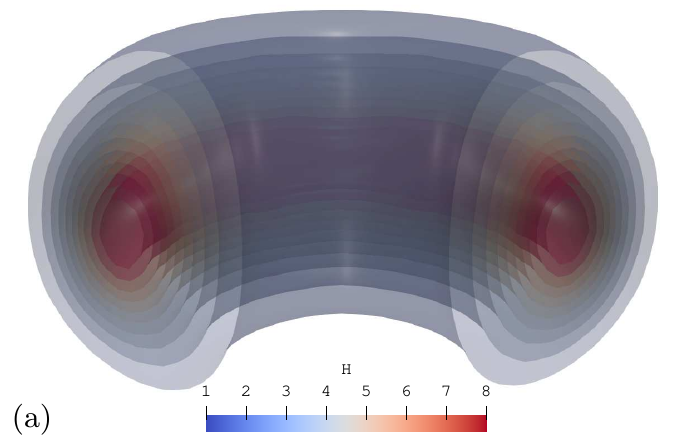}
		\includegraphics[width=.48\linewidth]{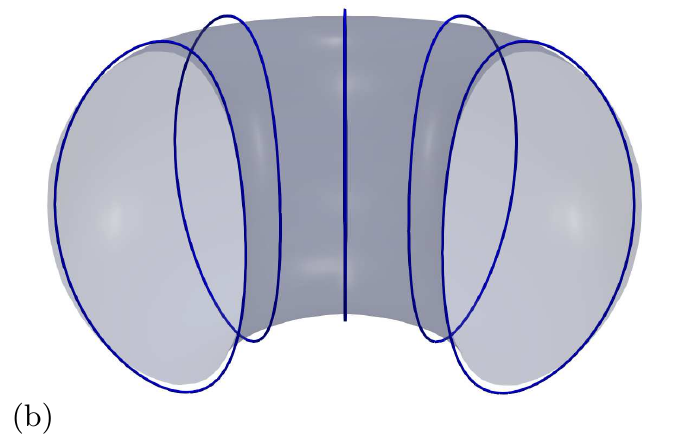}

	\vspace{-0.3cm}
	\caption{(a) Streamsurfaces of Fig.~\ref{fig:Hills_tori_per} that have a 95th percentile less than $0.05$ for $x>0$. (b) Solution curves of eq.~(\ref{eq:def_vel_Hill}) for $5$ different points lying on the outer streamsurface of (a).}
	\label{fig:Hills_tori_3D}
\end{figure}

\subsection{Flow inside a V junction}
\label{V_junction}
Next we investigate the flow inside a V-junction as depicted in Fig. \ref{fig:junction_streamlines}. Despite the simple geometry, recent experiments have suggested that pumping a particle-laden fluid into such a configuration allows light particles, such as gas bubbles in water, to be permanently trapped in the junction \citep{vigolo2014unexpected}. This phenomenon arises for a wide range of Reynolds numbers and for various junction angles \citep{ault2016vortex}. Here we consider one such flow with junction angle $70^{\circ}$ and $\mathrm{Re} = \left( \bar{u}L/\nu \right) = 230$ \citep{shin2015flow}, where $\bar{u}$ is the average inlet flow speed, $L$ is the side length of the square channel and $\nu$ is the kinematic viscosity. 
\begin{figure}
	\begin{centering}
		\includegraphics[scale=0.16]{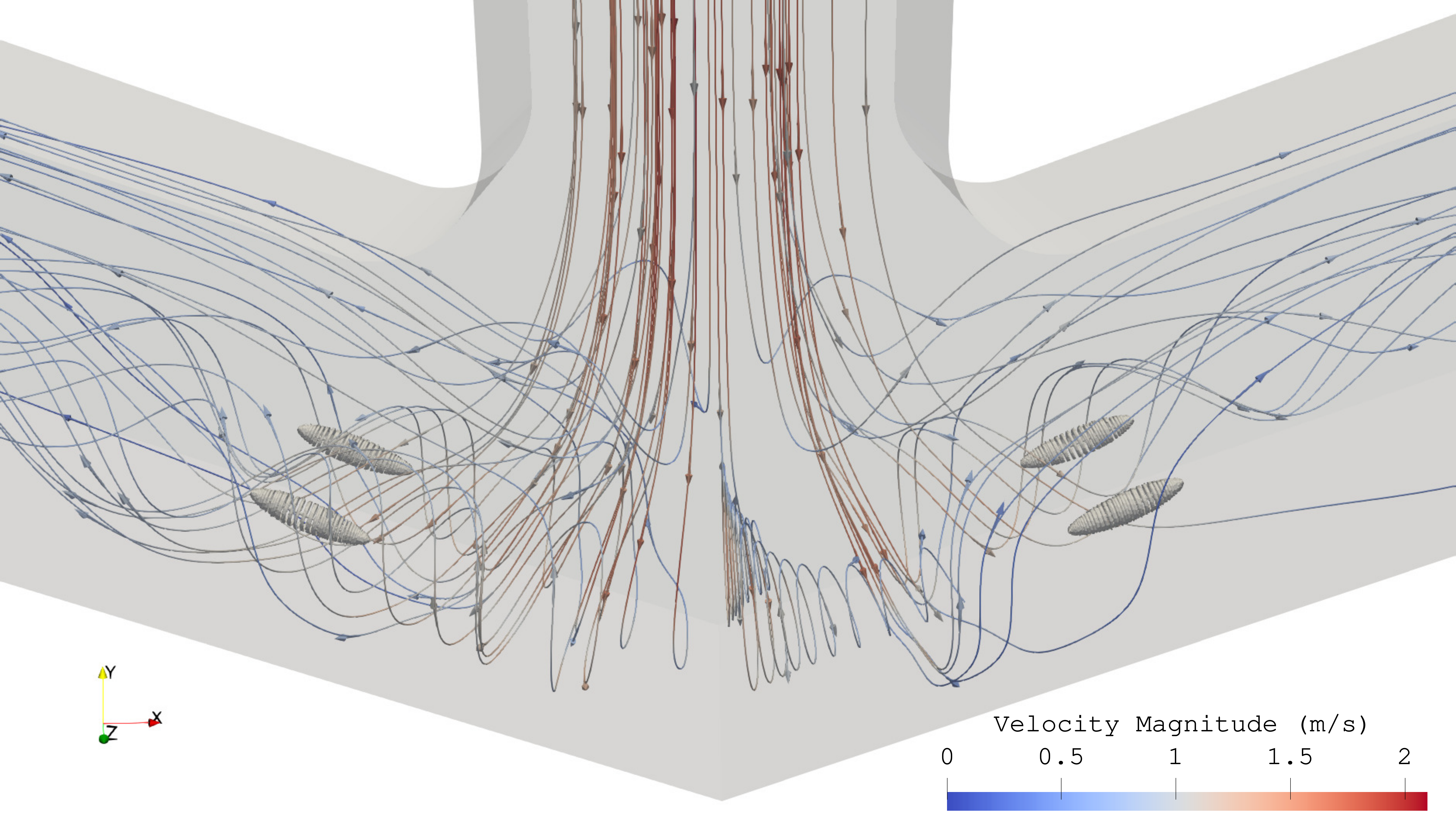}
		\par\end{centering}
	\caption{Velocity streamlines, color-coded with their magnitude, emanating from the inlet of the V-junction and leaving the domain from the two outlets. For $\mathrm{Re} = 230$ and a junction angle of $70^{\circ}$, four symmetric vortex-breakdown bubbles are portrayed in gray.}
	\label{fig:junction_streamlines}
\end{figure}

We use a finite-volume solver from the \texttt{OpenFOAM} library to obtain a steady solution to the 3D incompressible Navier--Stokes equation \citep{weller1998tensorial}.
The same numerical solution has recently been analyzed using methods from dynamical systems theory, which revealed large, anchor-shaped trapping regions for light particles  \citep{oettinger2018invisible}. These trapping regions, however, have been invariably linked to bubble-type vortex breakdown structures in the fluid flow, which are formed downstream in the junction \citep{ault2016vortex}. These structures are depicted as dark gray blobs in Fig.~\ref{fig:junction_streamlines}, obtained from a careful advection of streamlines from the vicinity of known stagnation points. Their construction, therefore, is by no means automated and assumes a detailed knowledge of the streamline geometry.

Each vortex breakdown bubble is demarcated by the 2D stable and unstable manifolds of two saddle-type fixed points. In a generic 3D flow, these two manifolds do not coincide because that configuration would not be structurally stable. Instead, they are expected to intersect transversely. In such a scenario, the original streamsurface breaks down allowing fluid particles from the main stream to be entrained into the bubble and, conversely, particles to return from the bubble to the main stream \citep{holmes1984some,peikert2007visualization}. \cite{sotiropoulos2001chaotic} showed that a very careful mesh refinement 
is required to reveal the splitting of the manifolds.
Here, we will only be interested in reconstructing a closed streamsurface that manifests minimal fluid exchange (see \cite{peikert2007visualization} for a discussion of this surface).

We select as computational domain the box depicted in Fig.~\ref{fig:computationalBoxes} with $110, 90, 80$ points in the $x,y$ and $z$ direction, respectively. This box is chosen so that the bubble-like vortical structure lies approximately in its center. The dimensions of the box are $L_{x} = 1 \textrm{m}$, $L_{y} = 0.6 \textrm{m}$ and $L_{z} = 0.4 \textrm{m}$. This procedure incorporates a priori information about the rough location of the streamsurface of interest. We will see in the next section how the same algorithm can be extended to uncover a priori unknown structures in a turbulent flow.
\begin{figure}
	\centering
	\includegraphics[width=0.45\linewidth]{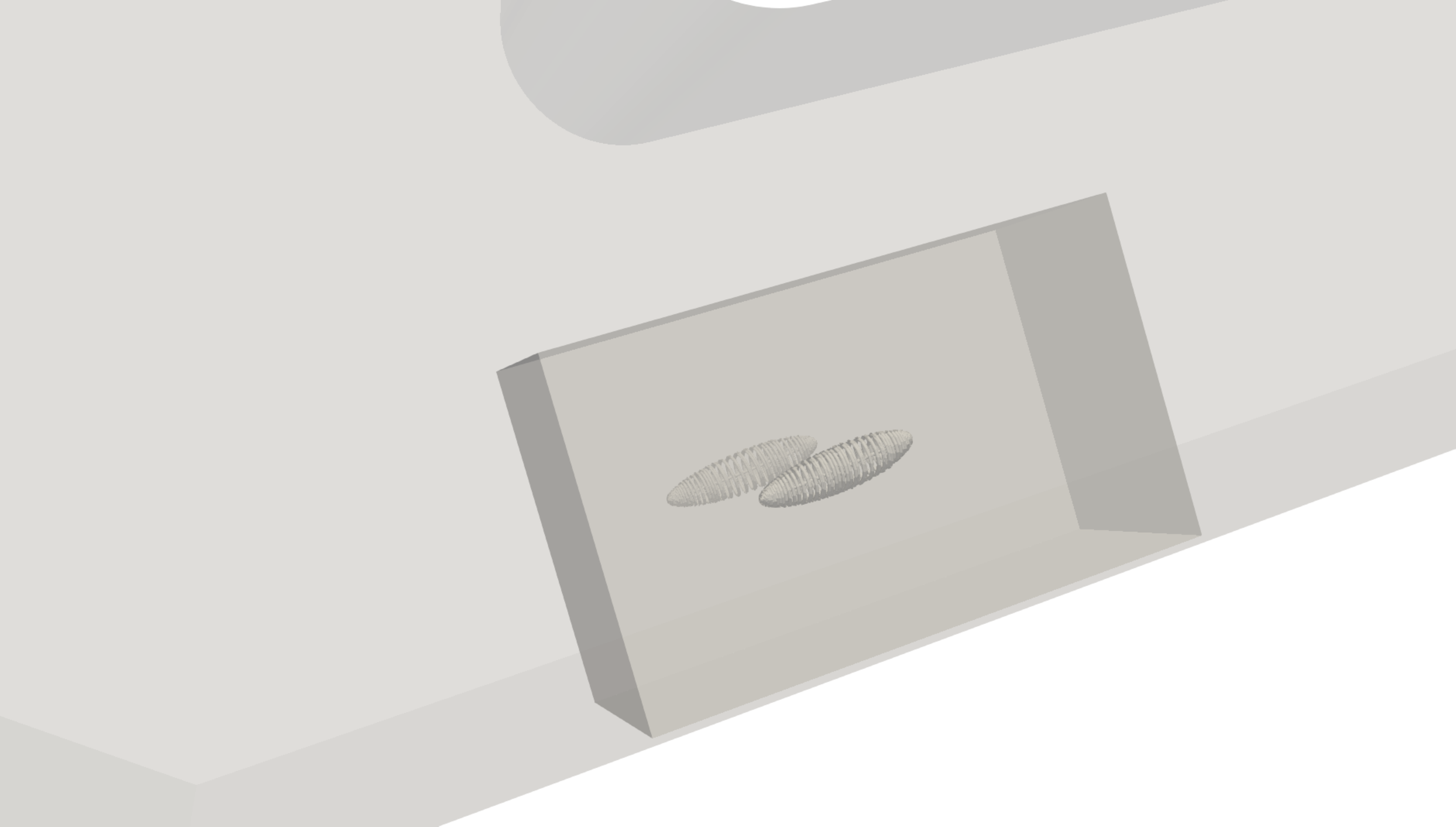}
	\includegraphics[width=0.45\linewidth]{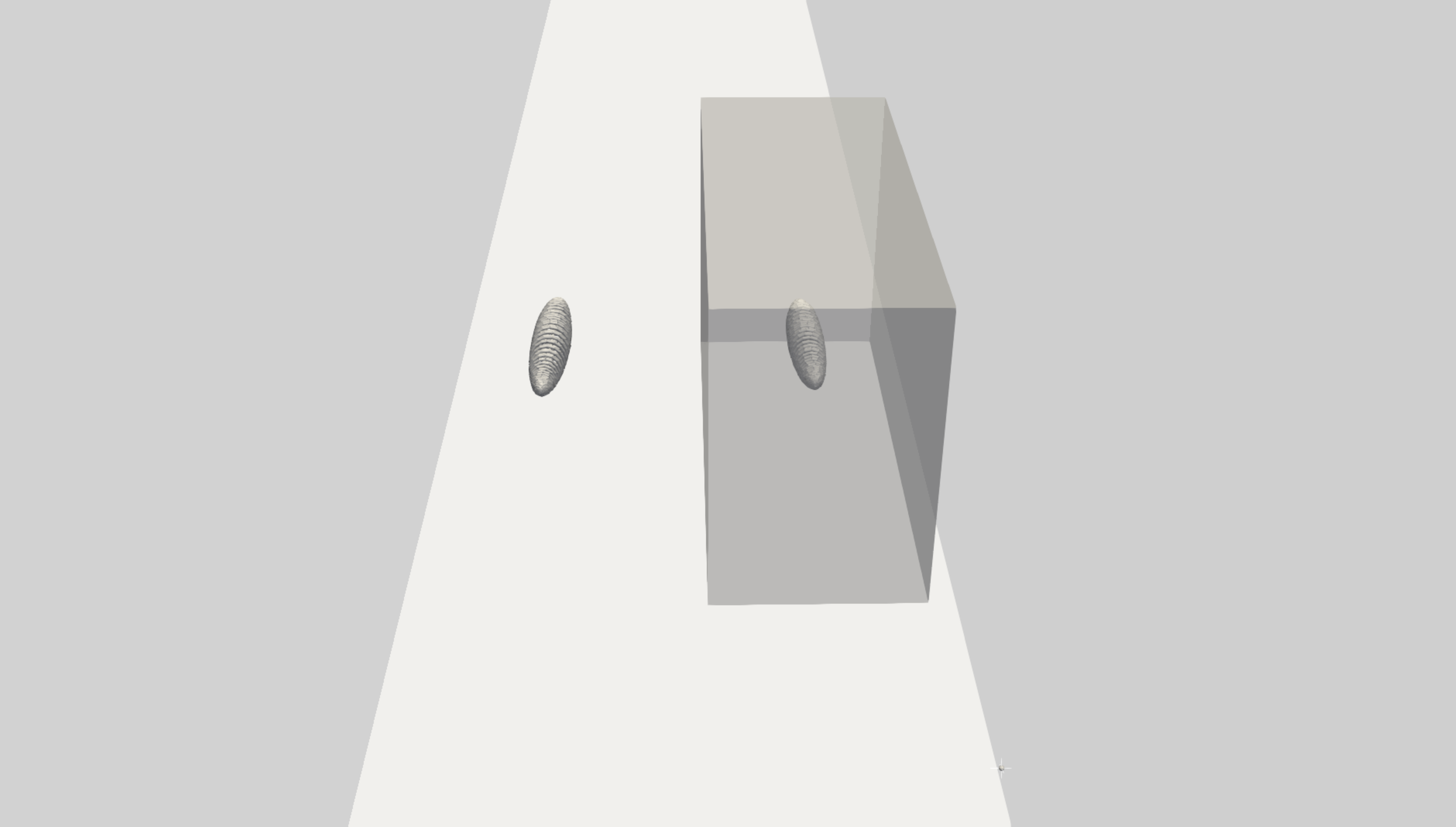}
	\caption{Different views of the computational box used to construct an approximate first integral for one of the bubble-like structures in the V junction flow.}
	\label{fig:computationalBoxes}
\end{figure}

We use our algorithm with $N=15$ or $14\,146$ Fourier modes and show a 2D cross-section at the middle of the computational domain in Fig.~\ref{fig:crossPlusInvariance}(a). For the smallest eigenvalue of $\mathbf{A}$, $d_1 = 58.1538$ ($d_2 = 376.9696$), a very pronounced circular structure is clearly visible approximately in the middle of the domain bordering an otherwise flat landscape. To obtain the full, 3D reconstruction, we use $40$ equidistant values between $\left| H \right|_{min}$ and $\left| H \right|_{max}$ and extract $28$ different level surfaces of the approximate first integral. We then sort them in decreasing order based on the volume they enclose. 
\begin{figure}
	\centering
	\includegraphics[width=0.45\linewidth]{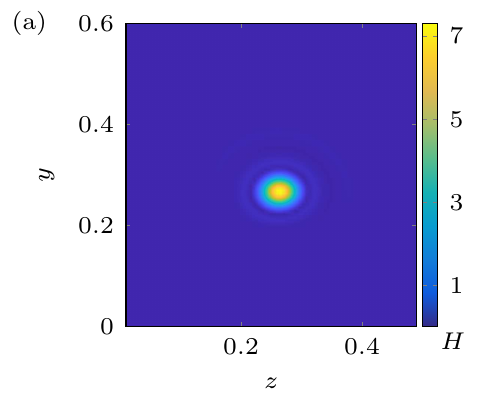}
	\includegraphics[width=0.45\linewidth]{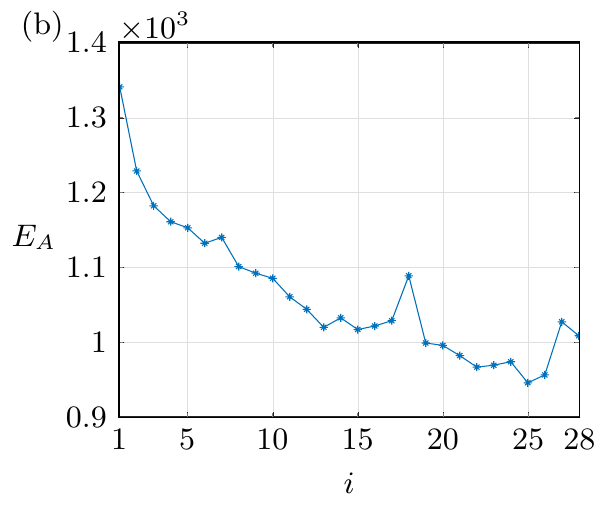}
	\caption{Results for an approximate first integral in the V-junction flow, framing one of the elliptical vortex regions. (a) Approximate first integral distribution on a plane which coincides with the middle, in the $x$ direction, of the computational domain presented in Fig.~\ref{fig:computationalBoxes}. (b) Normalized invariance error as a function of the extracted isosurfaces sorted in descending order with respect to their volume.}
	\label{fig:crossPlusInvariance}
\end{figure}

Plotting this error estimate for the extracted isosurfaces (Fig.~\ref{fig:crossPlusInvariance}(b)) yields a global minimum for the surface $i=25$ and a local minimum for the surface $i=15$. Rendered depictions of the extracted isosurfaces corresponding to these two minima are shown in Fig.~\ref{fig:junction_structures}. We note the close agreement between the reconstructed structures and the structure delineated by the judiciously chosen streamlines. This is in stark contrast to the structures suggested by the $Q$-criterion as demonstrated in Appendix \ref{app:faultyCases}.

\begin{figure}
	\centering
	\includegraphics[width=0.45\linewidth]{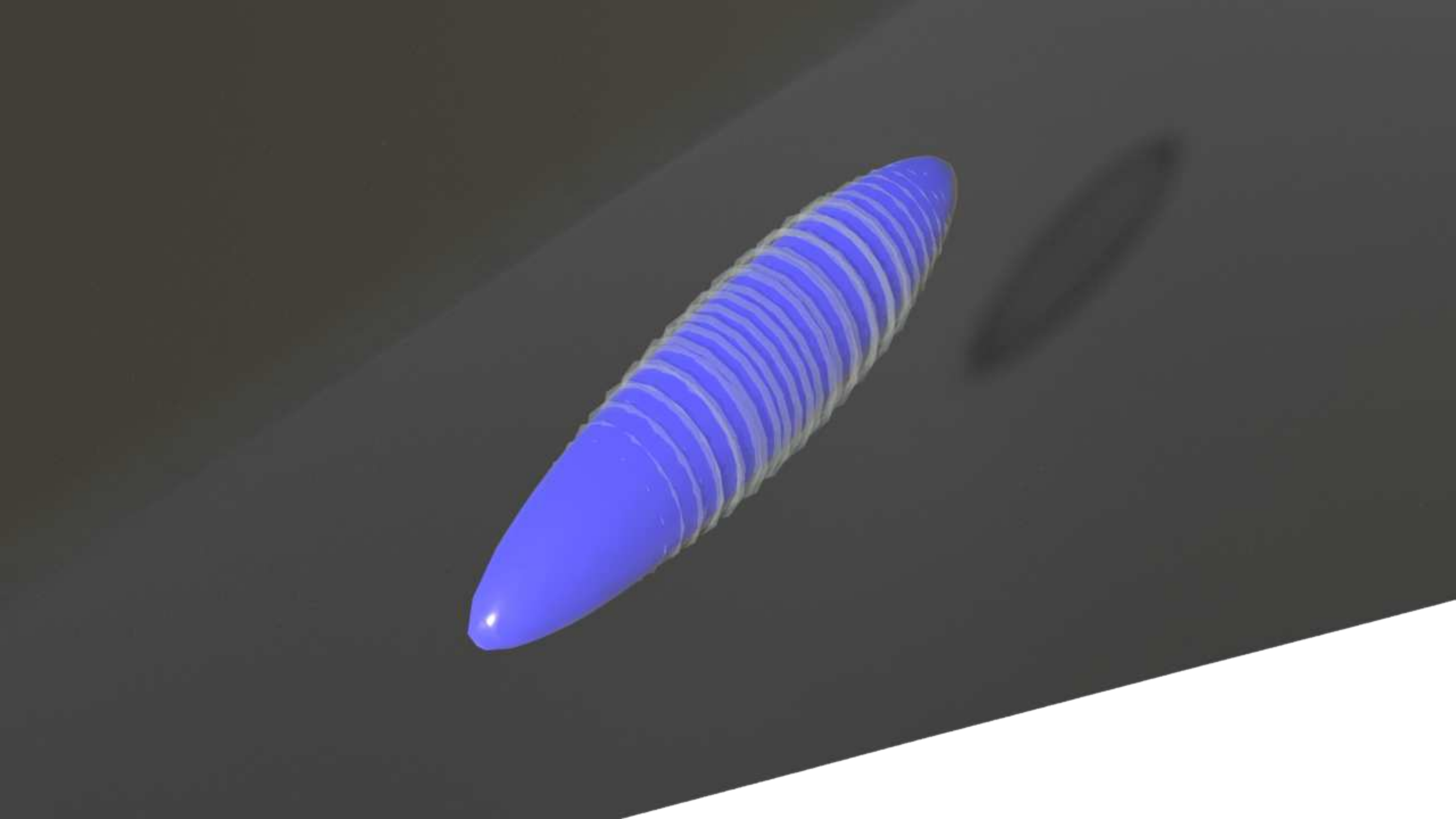}
	\includegraphics[width=0.45\linewidth]{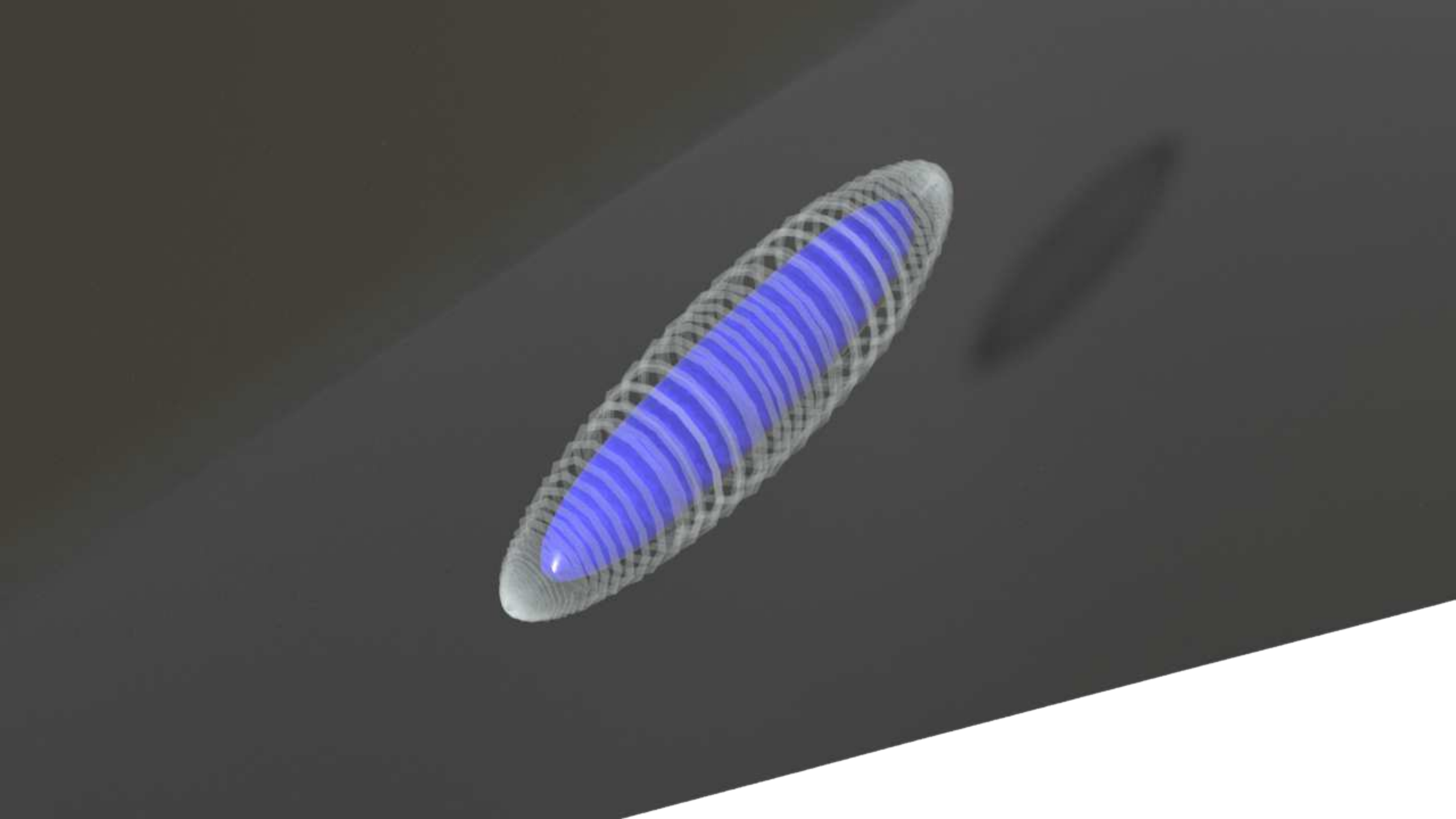}
	\caption{Streamsurfaces as level sets of an approximate first integral in the V-junction flow,  corresponding to one local $(i = 15)$ and the global minimum $(i = 25)$ of the normalized error $E_{A}$ depicted in Fig.~\ref{fig:crossPlusInvariance}(b). The duct is cut transversely to help the visualization.}
	\label{fig:junction_structures}
\end{figure}

\section{Momentum transport barriers in a turbulent channel flow}
\label{turbulentChannel}

As we have mentioned in the Introduction, the shape of streamsurfaces is observer-dependent: their geometry changes under general time-dependent rotations and translations of the observer. Therefore, unless the flow has a distinguished frame in which streamsurfaces coincide with material surfaces, streamsurfaces simply highlight features of the velocity field in a given frame, as opposed to intrinsic, observer-indifferent features of the flow of fluid particles. For studies seeking to be consistent with observed flow physics, the latter features are relevant \citep{haller2021can}. This is because flow visualization experiments with dye  particles highlight material (and hence objectively defined) transport barriers, which generally differ substantially from streamsurfaces in unsteady flows. At the same time, classic studies focused on momentum, energy or vorticity transport are inherently frame-dependent by the frame-dependence of all these Eulerian fields.

To reconcile these two objectives, \cite{haller2020objective} developed a theory of objectively defined barriers to the transport of active vector fields, such as the vorticity and momentum. \cite{aksamit2022objective} used this theory to locate instantaneous (Eulerian) frame-indifferent momentum transport barriers in 3D turbulent channel flows. Active barriers turn out to be distinguished streamsurfaces of appropriately defined steady, 3D incompressible vector fields \citep{haller2020objective}. Specifically, momentum transport barriers of a Navier--Stokes velocity field $\mathbf{v}(\mathbf{x},t)$ at time $t$ are  streamsurfaces of the barrier equation
\begin{equation}
	\mathbf{x}^{\prime}(s)=\Delta \mathbf{v}\left(\mathbf{x}(s),t\right),
	\label{eq:barrier equation_laplacian}
\end{equation}
with $s\in{\mathbb R}$ denoting a parametrization of streamlines forming the streamsurfaces and $\Delta$ the Laplace operator. For an incompressible velocity field we further note that $\Delta \mathbf{v} = -\nabla \times \boldsymbol \omega$ holds.

\cite{haller2020objective} and \cite{aksamit2022objective} detected distinguished streamsurfaces of the momentum barrier equation (\ref{eq:barrier equation_laplacian}) using active versions of some of the passive hyperbolic and elliptic LCS diagnostics reviewed by \cite{haller2015lagrangian}. These Lagrangian calculations involve arrays of trajectories and return diagnostic scalar fields for visual inspection rather than explicit streamsurfaces families.  In the following, we will use our approach for finding approximate first integrals to obtain vortical momentum barriers as level surfaces of an approximate first integral for eq.~(\ref{eq:barrier equation_laplacian}).

\subsection{Numerical data set}
\label{numericalDataset}
We study the 3D incompressible, turbulent channel flow which can be found \href{https://www.tu-ilmenau.de/universitaet/fakultaeten/fakultaet-maschinenbau/profil/institute-und-fachgebiete/fachgebiet-stroemungsmechanik/dfg-priority-programme-1881-turbulent-superstructures/benchmark-cases}{here}. The friction Reynolds number is $\mathrm{Re}_{\tau} = u_{\tau}h/ \nu= 150$, where $u_{\tau}$ is the friction velocity, $h$ is the channel half-height and $\nu$ is the kinematic viscosity. We denote by $x$, $z$ and $y$ the streamwise, spanwise and wall-normal directions, respectively. The computational domain is $L_x = 2.5\pi h$ long and $L_z = \pi h$ wide.

The number of Fourier modes used in the simulation was $192$ in both the streamwise and the spanwise direction. The number of
points in the wall-normal direction is $194$, inhomogeneously spaced so that the grid becomes more refined closer to the walls. No-slip boundary conditions were applied at the two channel walls and the governing equations were integrated forward in time with a constant pressure gradient $dp/dx = -1$ enforced so as to drive the flow through the channel. $10\,000$ velocity field snapshots were stored at multiples of the simulation time step  $\Delta t = 0.001$ once the flow reached  a statistically stationary state. We will identify the $5\,000$th snapshot of this time series with the time $t = 0$.

\subsection{Momentum barrier extraction}
We seek to uncover objectively defined instantaneous momentum transport barriers corresponding to the time $t = 0$ as specific streamsurfaces of the vector field $\Delta \mathbf{v}\left(\mathbf{x},0\right)$ in eq.~(\ref{eq:barrier equation_laplacian}). First, as illustration of prior results by \cite{haller2020objective} on this problem,  we show the  FTLE field computed for $\Delta \mathbf{v}\left(\mathbf{x},0\right)$ (called active FTLE or aFTLE) from a grid of $800 \times 1000$ initial conditions uniformly placed in the wall-normal and spanwise directions, respectively, up to $s = 10^{-3}$ (see Fig.~\ref{fig:aFTLE_recon}(a)). This plot indicates the signatures of two larger barriers to momentum transport, with the first located approximately around $[z/h,y/h] = [0.75,1.25] \times [1.5,1.75]$. The second, a mushroom-type barrier, is located around $[2.25,2.75] \times [1.25,1.75]$. A plethora of other smaller-scale structures are also present in the aFTLE plots.

\begin{figure}
	\centering
	\includegraphics[]{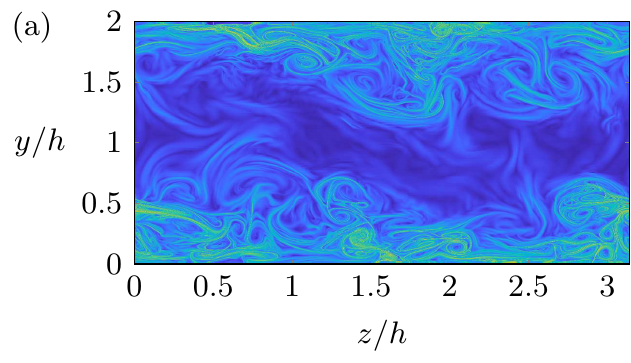}
	\includegraphics[]{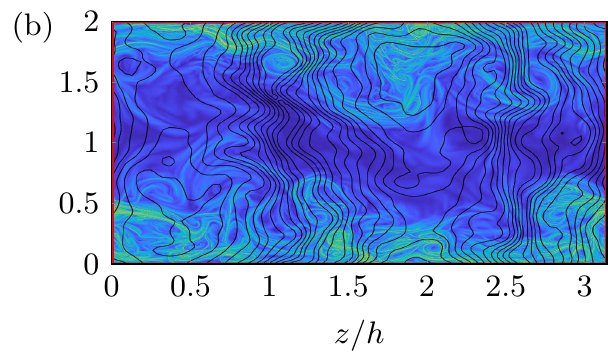}\\
	\includegraphics[]{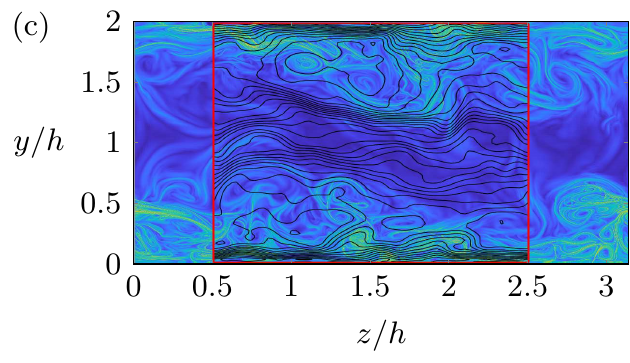}
	\includegraphics[]{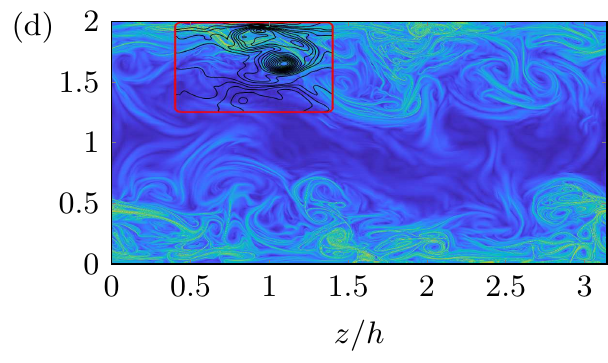}\\
	\includegraphics[width=0.39\linewidth]{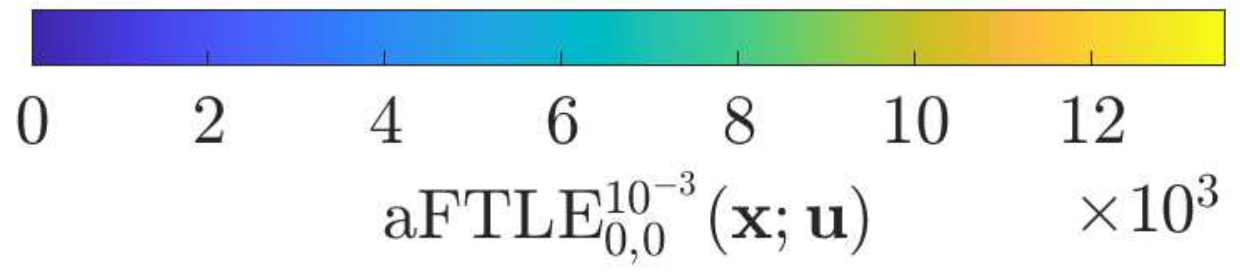}
	
	\vspace{-0.3cm}
	\caption{(a) Active FTLE (aFTLE) for the momentum barrier velocity field (\ref{eq:barrier equation_laplacian})  at $t = 0$ over a 2D cross-section of the channel at $x/h = 2$. (b), (c), (d) Projections of the approximate first integral (black lines) on the same cross-section using as computational domains for the analysis the (red) boxes $[1,3] \times [0,2] \times [0,\pi]$, $[1,3] \times [0,2] \times [0.5,2.5]$ and $[1,3] \times [1.25,2] \times [0.4,1.4]$, respectively, superimposed on the aFTLE landscape.}
	\label{fig:aFTLE_recon}
\end{figure}

We use three different computational domains to illustrate the potential of the algorithm presented in section \ref{algorithm_setup}: the boxes $[1,3] \times [0,2] \times [0,\pi]$, $[1,3] \times [0,2] \times [0.5,2.5]$ and $[1,3] \times [1.25,2] \times [0.4,1.4]$ with $80 \times 85 \times 80$ points in the streamwise, wall-normal and spanwise direction, respectively. In each of these boxes, the points are evenly spaced and a tri-linear interpolation scheme is employed to obtain the necessary $\Delta \mathbf{u}$ values. Moreover, we perform another rescaling of the dummy time $s$ to pointwise normalize the right-hand side of eq.~(\ref{eq:barrier equation_laplacian}). This eliminates the high norms of $\Delta \mathbf{u}$ near the wall in the expression (\ref{eq:closest_homo_system}). If this normalization is not used, our algorithm tends to miss the turbulent regions close to the wall and touches more on the central part of the channel.

For $N = 13$ $(9\,170$ Fourier modes$)$, $20$ isovalues of the approximate first integral are depicted in Figs.~\ref{fig:aFTLE_recon}(b), (c) and (d) on the $x/h = 2$ plane. We note the progressive improvement in the  reconstruction of the first structure discussed above, as the computational domain starts to close in on it. This observation prompts us to utilize a greedy algorithm compartmentalizing the domain into smaller, overlapping domains. The size of these domains will be dictated by the scale of the structures to be extracted.

Our algorithm so far generates families of streamsurfaces. We now seek to locate outermost members of nested elliptical barrier families. To this end, for each of the overlapping computational boxes, we use $n_{iso}$ different values to produce isosurfaces in the interval $\left[ \left| H \right|_{min},\left| H \right|_{max}\right]$. We then group the resulting isosurfaces into foliations (see Fig.~\ref{fig:AFTLE_mush_cut}(f) for two such foliations) or more generally families by finding the ones that either lie entirely inside others or have an intersection volume above a certain threshold, respectively. Subsequently, we discard the foliations or families that have a number of members less than a fixed percentage of $n_{iso}$. We note that for the isosurface extraction we use MATLAB's built-in function which is based on the "Marching cubes" algorithm \citep{lorensen1987marching}. Other techniques generating isosurfaces (see, e.g., the concept of contour trees \citep{carr2003computing}) are known to produce robust structures in a more efficient fashion and could certainly be used as viable alternatives.

We also follow \cite{haller16} and define the convexity deficiency of a closed surface as the ratio of the volume between the surface and its convex hull to the volume enclosed by the surface. Here, however, we will only discard the most non-convex surfaces. 

We summarize the main steps of the computations we described above in Algorithm \ref{algo}. In the next section, we will use this algorithm to uncover momentum transport barriers tied to different scales in the turbulent channel flow.

\noindent \begin{algorithm}
		\caption{Extraction of instantaneous barriers to momentum transport} 	\label{algo}{} 	\textbf{Input:} A snapshot of the $\Delta \mathbf{v}$ field defined over a 3D Cartesian grid. Lengths $L_x$, $L_y$ and $L_z$ as well as number of points $N_x$, $N_y$ and $N_z$. Different parameters $n_{iso}$, $p_b$, $d_{max}$ and $p_f$.
		\begin{enumerate}[leftmargin=.75cm,labelsep=.25cm]
			\item Partition the original grid into smaller, overlapping rectangular domains with dimension $L_x \times L_y  \times L_z$.
			\item Use tri-linear interpolation to obtain the values of $\Delta \mathbf{v}/ \left| \Delta \mathbf{v} \right|$ on a grid of $N_x \times N_y \times N_z$ points inside every domain. 
			\item Form the coefficient matrix $\mathbf{C}$ and obtain the eigenvector of $\mathbf{A}$ corresponding to its smallest eigenvalue.
			\item Obtain a 3D scalar field $\left| H \right|$ serving as an approximate first integral inside each domain and construct $n_{iso}$ isosurfaces out of it.
			\item Discard the surfaces that have at least one point within $p_b$ percent of either $L_x$, $L_y$ or $L_z$ of each domain's face boundaries, respectively. Also, discard all the surfaces that have a convexity deficiency ratio larger than a threshold $d_{max}$.
			\item Classify the remaining structures into foliations or families after locating the ones that either lie entirely inside others or have an intersection volume above a certain threshold, respectively. Keep only the foliations or families with more than $p_f \cdot n_{iso}$ members.
			\item Plot the member of each foliation or family with the largest volume.
		\end{enumerate} 
			\textbf{Output:} Almost convex, frame-indifferent 2D vortical barriers to momentum transport in the velocity field $\mathbf{v}$.
	\end{algorithm} 

\subsection{Results}
\subsubsection{Channel partition into large subdomains}
\label{partitionLarge}
To illustrate our algorithm, we use the entire computational domain of the Direct Numerical Simulation (DNS) described in section \ref{numericalDataset}. First we partition it in the following overlapping subdomains
$ [x_{1_i},x_{2_i}] \times [y_{1_j},y_{2_j}] \times [z_{1_k},z_{2_k}]$ with $x_1 \in \{-0.5,0.5,... ,5.5\}$, $x_2 \in \{1.5,2.5,... ,7.5\}$, $y_1 \in \{0,0.5,1\}$, $y_2 \in \{1,1.5,2\}$, $z_1 \in \{0,1,2\}$, $z_2 \in \{1.5,2.5,3.5\}$,
$i = 1, \ldots , 7$, $j = 1, 2, 3$, $k = 1, 2, 3$. For each of these domains, we use a grid of $70 \times 75 \times 70$ evenly spaced points on which we compute the normalized velocity Laplacian. For the approximate first integral we use $N=13$ (or $9\,170$ Fourier modes). This comes as a necessary trade-off between computational time and surface accuracy after we observed that the reconstructed structures do not vary significantly when we use $N=12$ or $N=14$. With it, we extract $n_{iso} = 40$ isosurfaces in every domain. Of them, we discard the surfaces described in steps $5$ and $6$ of the Algorithm \ref{algo} by employing $p_b = 2 \%$, $p_f = 10 \%$ and $d_{max} = 20 \%$.
 
The results of this computation are shown in Fig.~\ref{fig:structures_box_2}. We observe that the entire channel is populated by a host of different momentum transport barriers, the majority of which have a clear quasi-streamwise direction. These structures are reminiscent of those investigated in studies of wall-bounded flows \citep{robinson1991coherent}. The structures we obtain, however, are more evenly scattered throughout the channel, appearing not only near the channel walls but also in the more quiescent central part.
This is in stark contrast to typical predictions from classic vortex  criteria  for these flows \citep{jcr1988eddies, chong1990general, hussain1995}. This phenomenon, i.e., structures penetrating into and spanning the bulk flow region, has already been noted in the literature \citep{haller2020objective,aksamit2022objective} by analyzing the 2D signatures of these structures via diagnostic fields. Here, however, we explicitly construct these structures via the streamsurfaces surrounding them.
\begin{figure}
	\begin{centering}
		\includegraphics[width=0.85\linewidth]{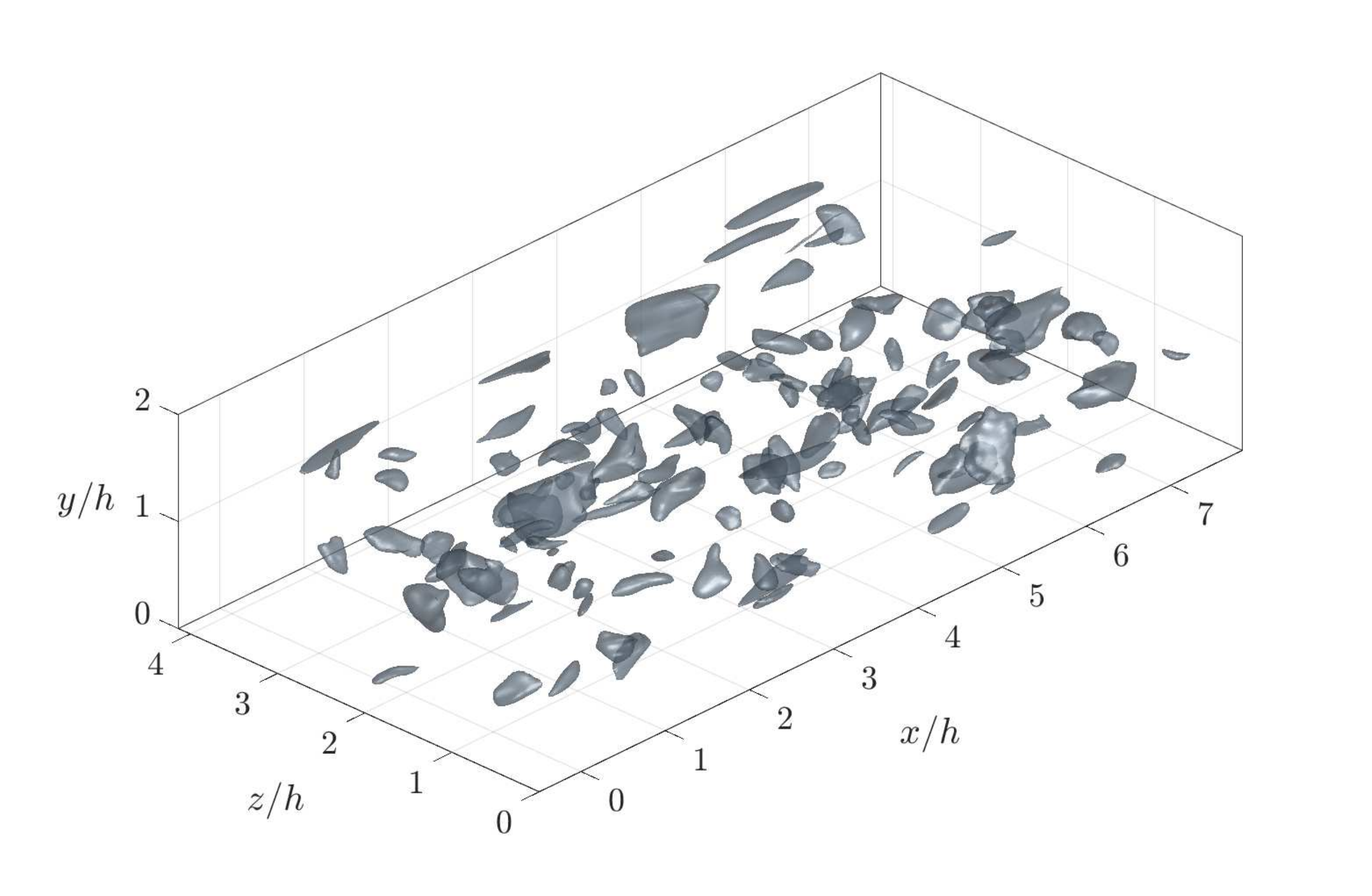}
		\par\end{centering}
	\caption{Instantaneous barriers to momentum transport in the 3D turbulent channel flow, derived using the larger partition described in section \ref{partitionLarge}.}
	\label{fig:structures_box_2}
\end{figure}

\subsubsection{Channel partition into small subdomains}
\label{partitionSmall}
We now employ smaller partitions of the channel, using the overlapping subdomains $ [x_{1_i},x_{2_i}] \times [y_{1_j},y_{2_j}] \times [z_{1_k},z_{2_k}]$ with $x_1 \in \{-0.5,0,... ,6.5\}$, $x_2 \in \{0.5,1,... ,7.5\}$, $y_1 \in \{0,0.5,1\}$, $y_2 \in \{1,1.5,2\}$, $z_1 \in \{0,1,2\}$, $z_2 \in \{1.5,2.5,3.5\}$,
$i = 1, \ldots , 15$, $j = 1, 2, 3$, $k = 1, 2, 3$, while keeping all the other parameters the same as in section \ref{partitionLarge}.

The results for this computation are shown in Fig.~\ref{fig:structures_box_1}. In this case, we again observe momentum transport barriers that tend to align with the streamwise direction but in larger numbers, when compared with Fig.~\ref{fig:structures_box_2}. Another notable difference is that the reconstructed structures show generally smaller scales than before. This demonstrates that our domain partition algorithm acts as a filter of various scales. This provides a natural way out of the occlusion quandary that besets many visualization approaches while retaining the algorithm's capacity to capture the smallest structures, delineated in the aFTLE plots, by using finer partitions. In the next subsection, we will demonstrate this capacity by pinpointing structures that the active diagnostics-based methods in \cite{haller2020objective} and \cite{aksamit2022objective} would likely miss. We conclude this section with Fig.~\ref{fig:d1_2_channel} depicting a well-formed spectral gap between the smallest and the second smallest eigenvalue of $\mathbf{A}$ for all the computational boxes. Finally, we carry out an unsteady barrier analysis for times varying over the interval $[t_{0},t_{1}] = [0,0.5]$, as presented in the supplementary Movie 2.mp4. 
In it, the fact that some barrier surfaces seem to appear, then disappear and perhaps reappear at a later time is a direct consequence of the flat convexity deficiency threshold we allow for these surfaces (20\%) throughout the channel. We expect this phenomenon to be resolved when the purely Lagrangian approach of \cite{haller2020objective} is combined with our algorithm in follow-up work.
\begin{figure}
	\begin{centering}
		\includegraphics[width=0.85\linewidth]{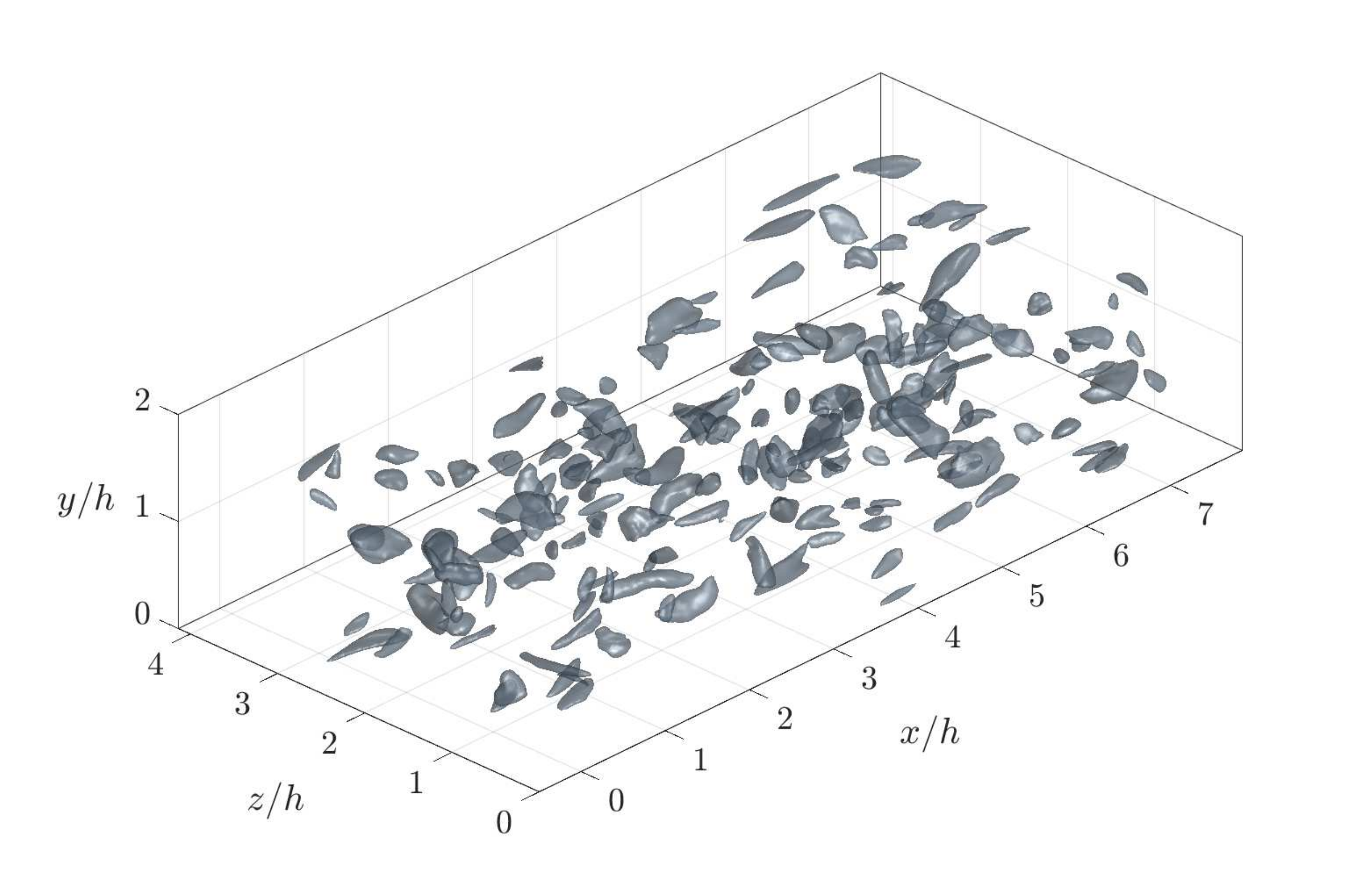}
		\par\end{centering}
	\caption{Instantaneous barriers to momentum transport in the 3D turbulent channel flow, derived using the smaller partition described in section \ref{partitionSmall}. The file Movie 2.mov of the supplementary materials portrays the extracted structures for the time interval $[t_{0},t_{1}] = [0,0.5]$.}
	\label{fig:structures_box_1}
\end{figure}

\begin{figure}
	\begin{centering}
		\includegraphics[width=0.5\linewidth]{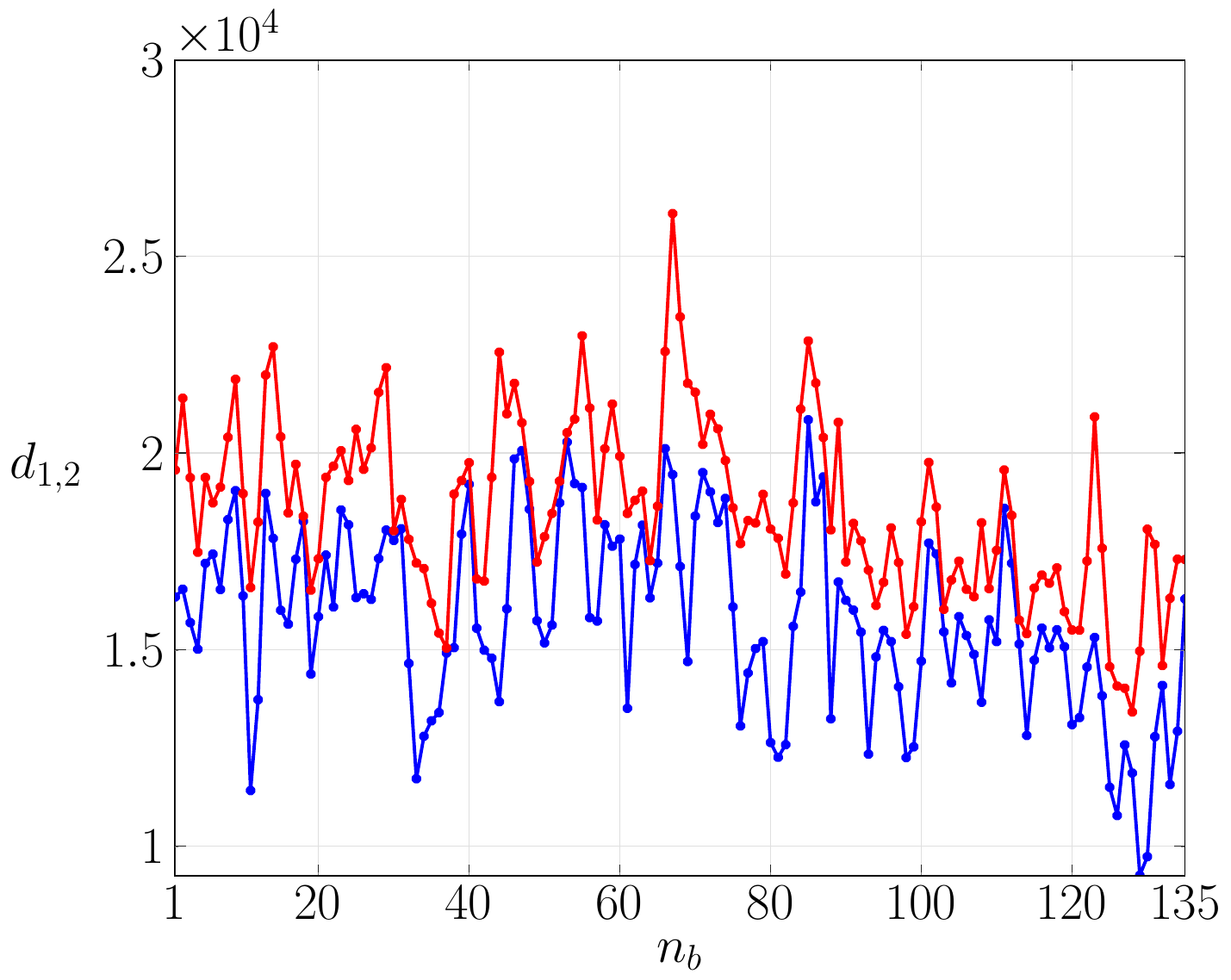}
		\par\end{centering}
	\caption{Distribution of the smallest (blue) and the second smallest (red) eigenvalues of $\mathbf{A}$ over the different computational domains used for the barriers surfaces presented in Fig.~\ref{fig:structures_box_1}.}
	\label{fig:d1_2_channel}
\end{figure}

\subsubsection{Uncovering hidden structures}
\label{hiddenStructures}
We now focus on the region we highlighted in Fig.~\ref{fig:aFTLE_recon}(d). The signature of a vortical structure is evident from the aFTLE landscape and it is corroborated by the isolines of an approximate first integral. Zooming in on the vicinity of this structure in Fig.~\ref{fig:structures_box_1}, however, reveals the existence of another structure (see Fig.~\ref{fig:AFTLE_mush_cut}), showing no imprint on the aFTLE landscape.

To investigate this, we populate the neighborhood around these two structures with 2D cross-sections and compute the aFTLE field for each of them. At $x/h = 1$ (Fig.~\ref{fig:AFTLE_mush_cut}(a)), we note two mushroom-like pairs of structures, one larger and one smaller, but neither of the two display prominent closed regions that could signal a vortex. As we move towards larger $x/h$ values, however, the aFTLE topography changes drastically. At $x/h = 1.4$ (Fig.~\ref{fig:AFTLE_mush_cut}(b)), the larger structure starts to fall apart before it completely disintegrates in the following cross-sections. In contrast, the smaller structure develops a salient closed loop, delineated by an aFTLE ridge, which is coincident with the largest of the reconstructed surfaces. In the next two cross-sections (Figs.~\ref{fig:AFTLE_mush_cut}(c) and (d)), the smaller of the two reconstructed surfaces comes into play as we further observe both parts of the vortex pair following closely the aFTLE ridges downstream. Finally, in the last cross section (Fig.~\ref{fig:AFTLE_mush_cut}(e)), the vortex pair pattern is broken up in agreement with the reconstructed vortex pair coming to a halt.

This example epitomizes two notable features of Algorithm \ref{algo}. First, using finer partitions results in the reconstruction of smaller structures consistent with the well-known hierarchy of coherent structures in turbulence. Second, this reconstruction takes place in an almost automated fashion with minimal reliance on user-defined parameters and with no need to advect trajectories of the barrier equation.
\begin{figure}
	\centering
	\includegraphics[width=0.45\linewidth]{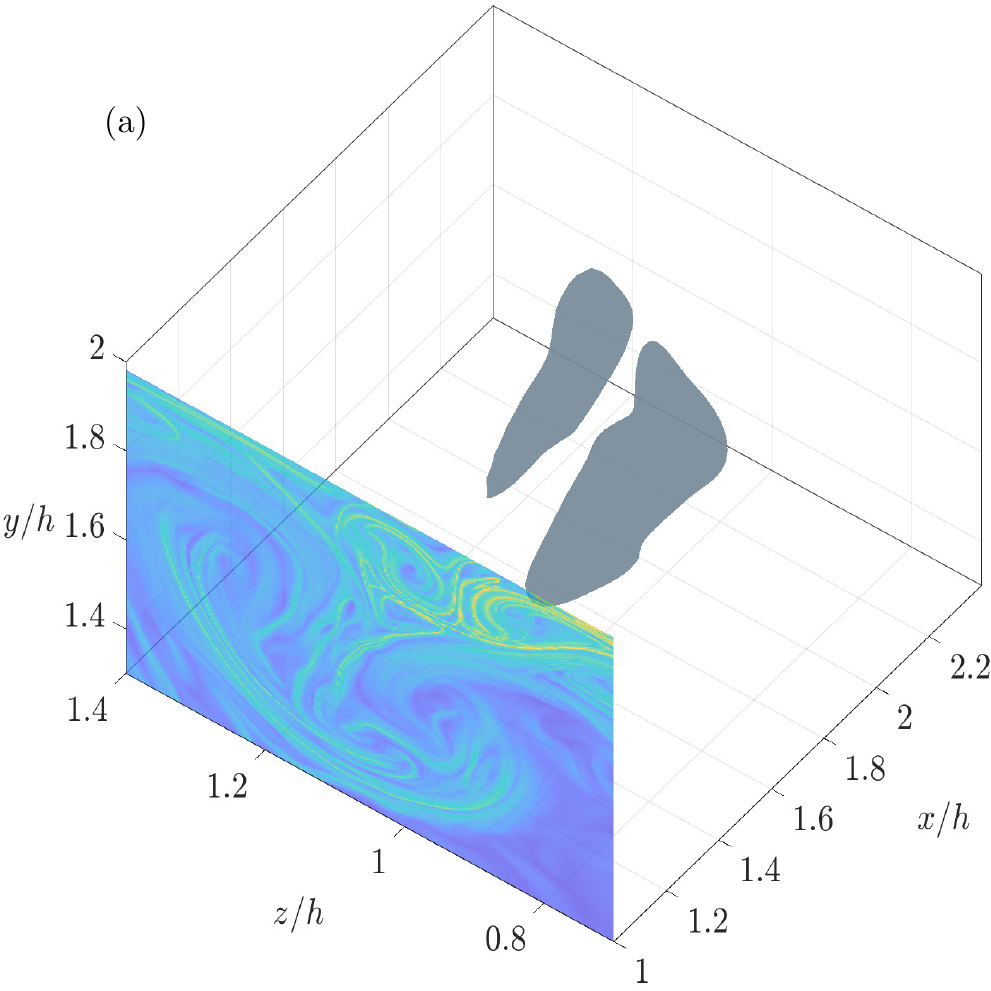}
	\includegraphics[width=0.45\linewidth]{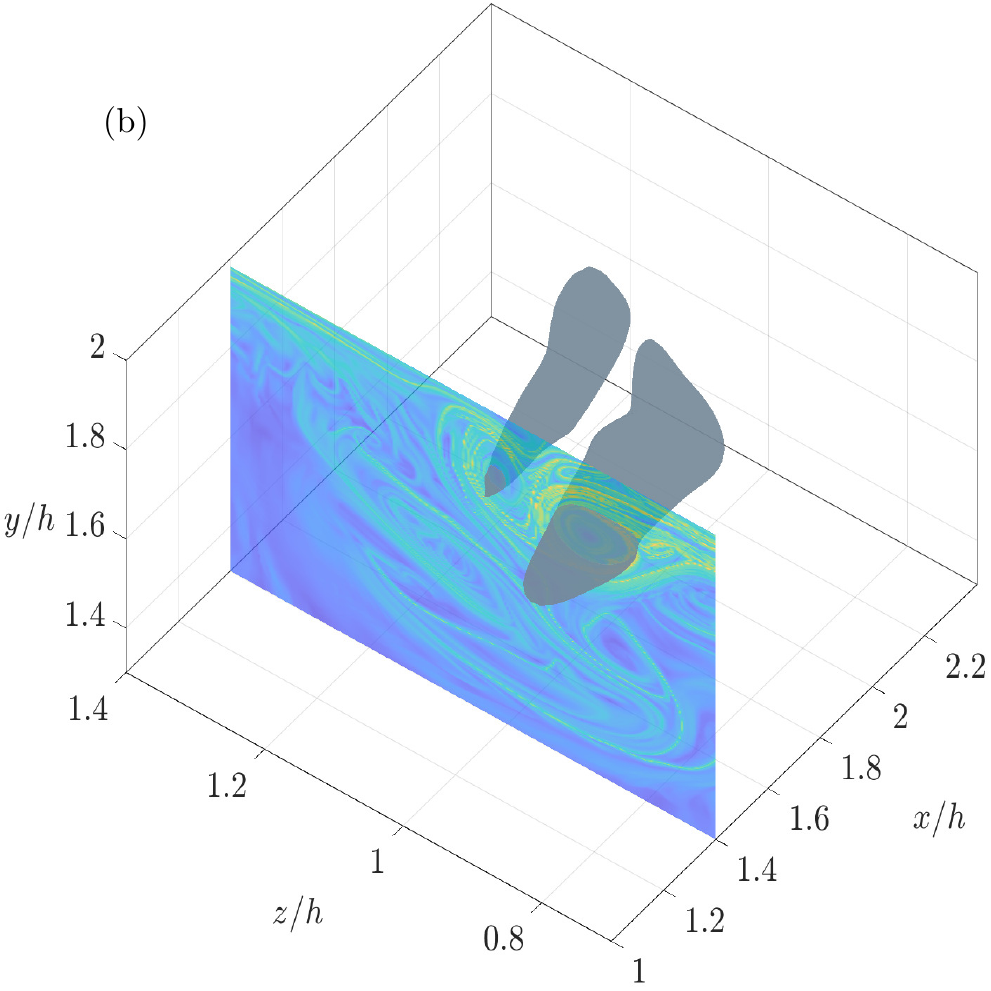}\\
	\includegraphics[width=0.45\linewidth]{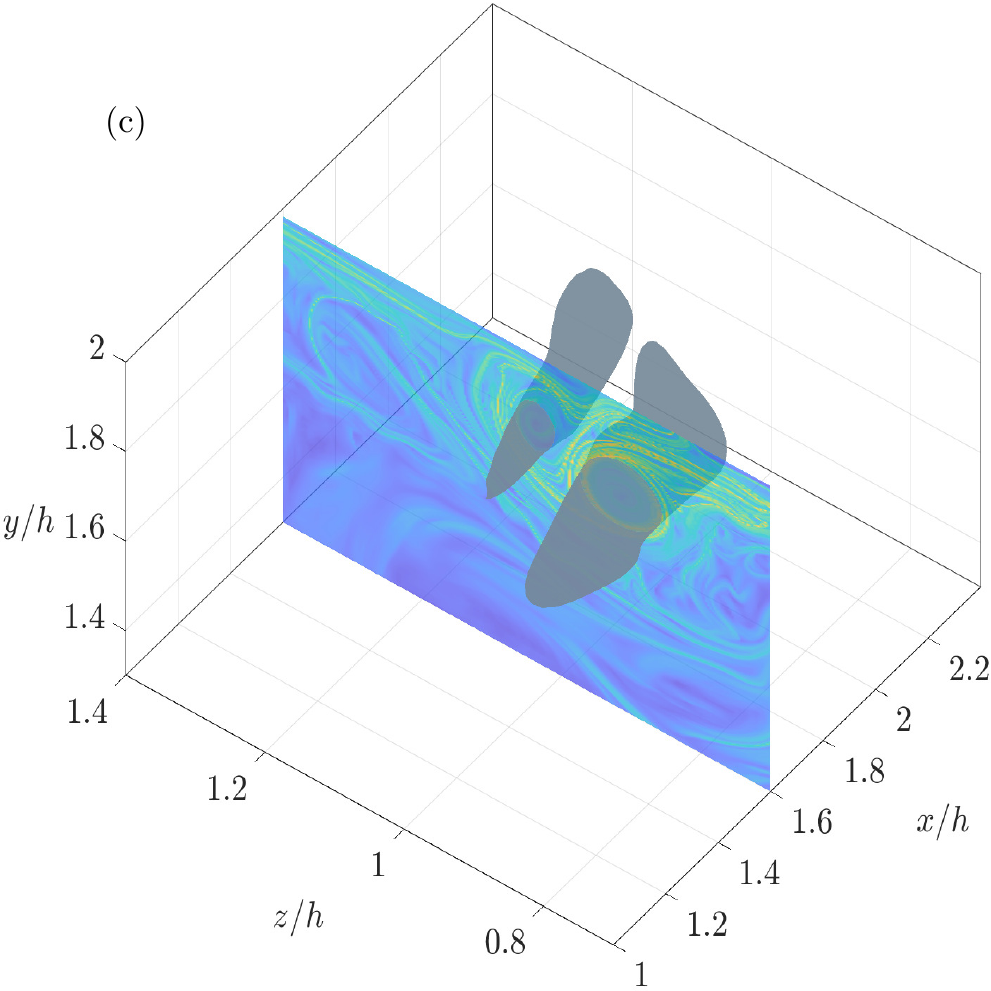}
	\includegraphics[width=0.45\linewidth]{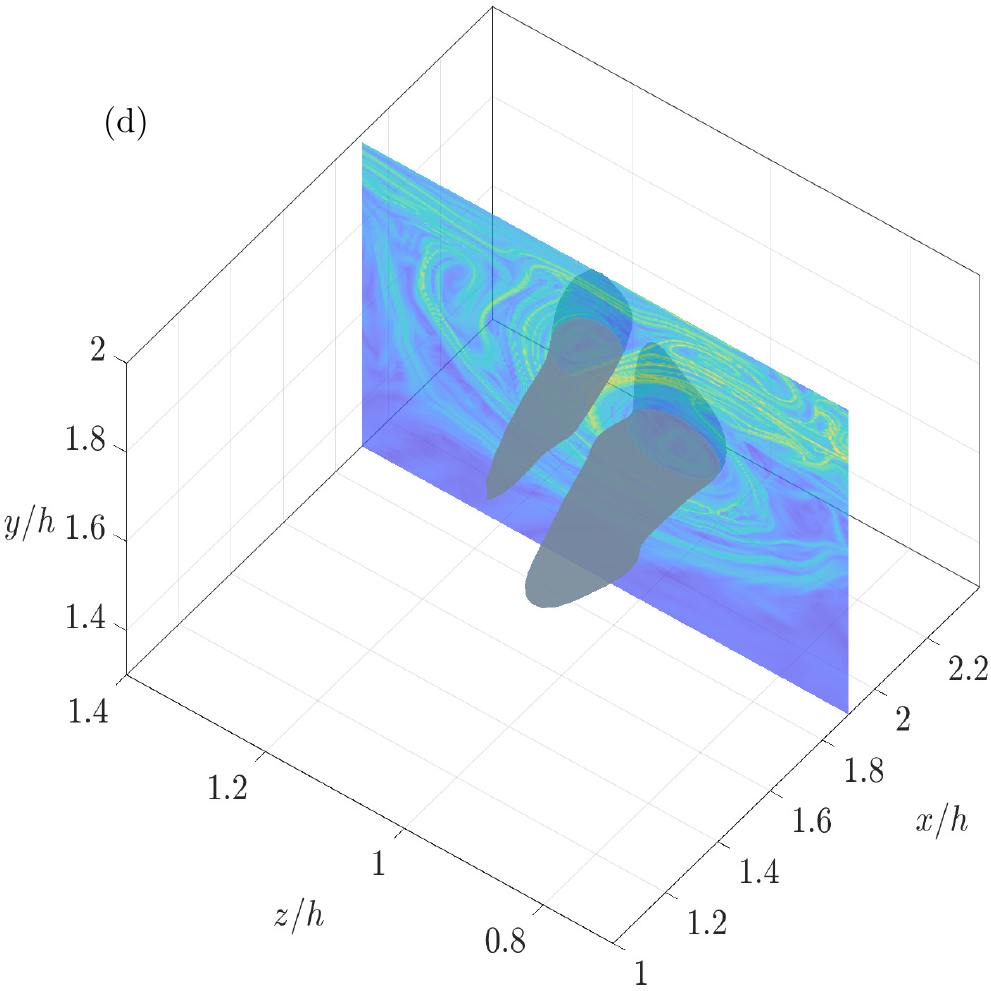}\\
	\includegraphics[width=0.45\linewidth]{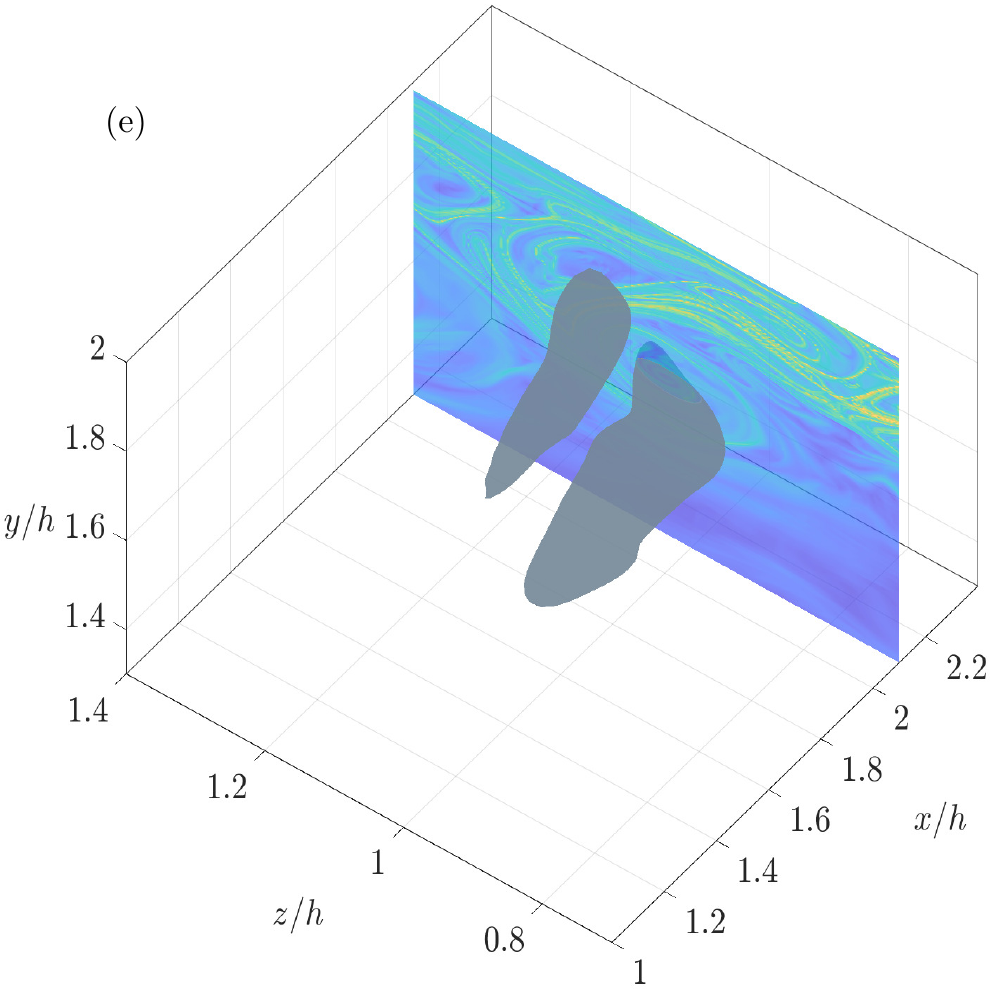}
	\includegraphics[width=0.45\linewidth]{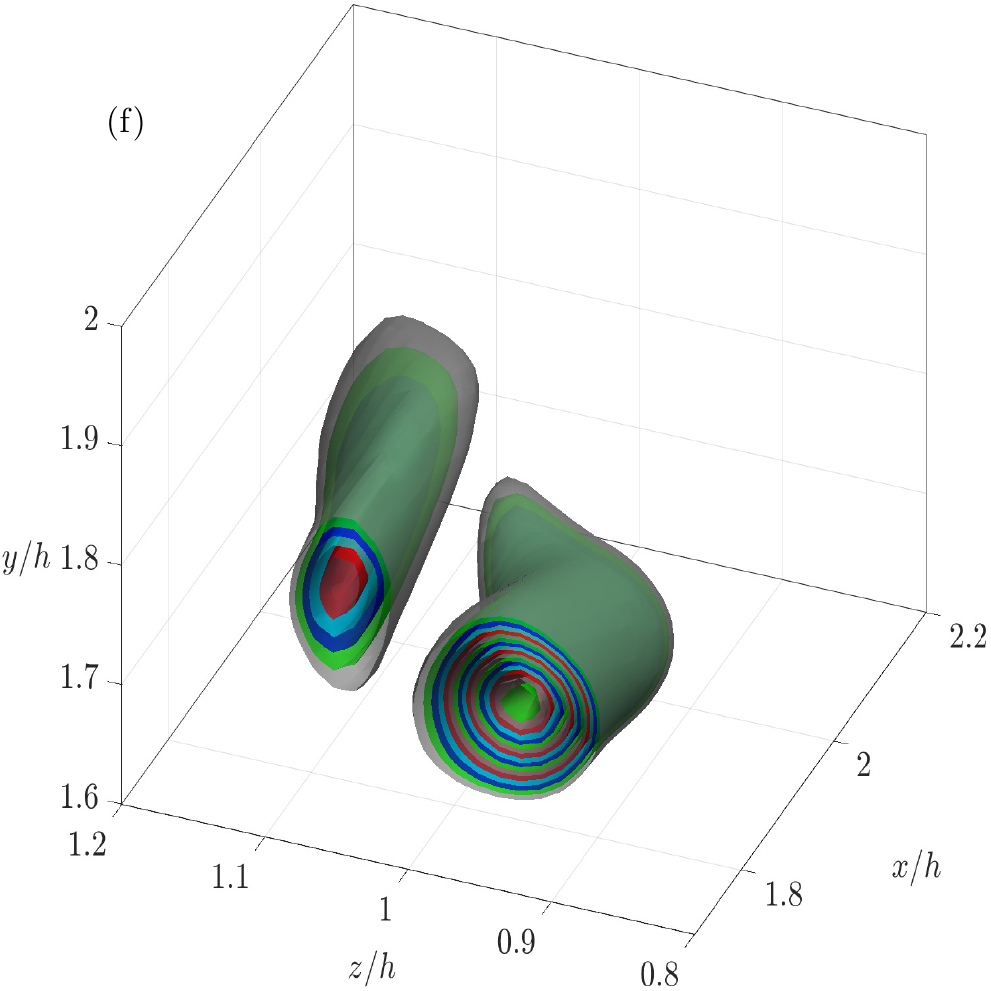}

	\caption{Two branches of a mushroom-like, objective vortex (in gray) captured by Algorithm \ref{algo} in the 3D turbulent channel flow. The branches are superimposed on 2D cross-sections of the aFTLE field at $x/h = 1, 1.4, 1.6, 1.9$ and $2.1$ in subfigs.~(a), (b), (c), (d) and (e), respectively. (f) A transverse cut at $x/h = 1.7$ revealing the foliations of structures constituting the two branches of the mushroom-like vortex.}
	\label{fig:AFTLE_mush_cut}
\end{figure}

\section{Conclusions}
We have introduced a minimization principle to find an instantaneous approximate first integral of a given 3D vector field. Level surfaces of this approximate first integral are expected to provide automated approximations to vortical streamsurfaces highlighting the instantaneous elliptic regions of the vector field. We have also proposed and tested a numerical algorithm  for solving this minimization problem by means of total least squares and ridge regression.  This algorithm has performed well on  various solutions of the steady 3D Euler equations on triply periodic and non-periodic domains. For steady velocity fields, the parametrized family of streamsurfaces rendered by our approach approximate  elliptic (toroidal or cylindrical) LCSs without the  need to advect a large number of trajectories required by Lagrangian methods \citep{haller16}. We have also illustrated on a 3D junction flow the applicability of our algorithm for velocity fields defined as numerical datasets produced from CFD simulations.

We have additionally used the described algorithm to extract objectively defined instantaneous momentum transport barriers in a 3D turbulent channel flow in an almost-automated fashion. Such barriers are streamsurfaces of an incompressible barrier equation defined by the Laplacian of the velocity field \citep{haller2020objective,aksamit2022objective}. We have extracted vortical momentum transport barriers without the need to advect arrays of trajectories and found such barriers across multiple spatial scales.

The generality of the presented algorithm provides a fertile ground for its use in the visualization of the important features in a number of diverse applications. A straightforward example is the barrier equation derived by \cite{haller2020objective} for the transport of vorticity.
Furthermore, the purely Lagrangian approach of \cite{haller2020objective} could be combined with our algorithm to produce LCSs for unsteady simulations. This method, unlike the instantaneous approach used for the production of Movie 2.mp4, will yield smoothly varying structures that are also experimentally verifiable.
Another applicatin of our algorithm could be the visualization of the magnetic field resulting from magnetohydrodynamic turbulence simulations \citep{biskamp2003magnetohydrodynamic}. Future research can also apply the algorithm to direction fields like those emerging from studies of 3D LCSs \citep{oettinger2016autonomous}, to 3D material barriers to diffusive transport \citep{haller2018material} or to uncover LCSs acting as barriers in diverse applications \citep{chem2021}. Such studies could extend the automated extraction of 2D LCSs and transport barriers \citep{2020vortex,barrierToolManual} to 3D flows, facilitating their integration to numerical simulation codes. 

Alternative solutions to the same optimization problem for an approximate first integral are certainly viable to investigate in future work. These could  involve different basis functions (such as Chebyshev polynomials, wavelets  or radial basis functions) instead of Fourier expansions. Techniques to mitigate the Gibbs phenomenon tied to non-periodic domains (Gegenbauer polynomials) \citep{gottlieb1997gibbs} are also feasible to consider. Similarly, the effect of using Chebyshev nodes, instead of the uniform grids we considered here, to reduce the impact Runge's phenomenon has near the boundaries could be examined.

\section*{Acknowledgements}
S.K.~and G.H.~acknowledge financial support from the Priority Program SPP 1881 (Turbulent Superstructures) of the German National Science Foundation (DFG). We are grateful to Dr.~Davide Gatti for providing us with the 3D turbulent channel flow data set he originally generated as a benchmark case for the same program.


\section*{Code and data availability}
Example scripts generating the approximate first integral for the ABC flow are available at \href{https://github.com/katsanoulis/Approximate\_First\_Integral/}{https://github.com/katsanoulis/Approximate\_First\_Integral/}. The V junction and turbulent channel flow data sets as well as their respective surface extraction scripts are available upon request from the first author.

\section*{Declaration of interests}
The authors report no conflict of interest.

\section*{Supplementary movies}
\label{supplementaryMaterials}
Supplementary movies are available \href{https://drive.google.com/drive/folders/100RvCuSiVFVdCFdnsqUjU0dLrfETHn87?usp=sharing}{here}.

\appendix
    \input{faultyCases}
	\input{least_squares_homo}
	\input{realConstraint}
	\input{skinnySVD}
	\input{algorithm2}

%% file: faultyCases.tex
\section{Typical flow visualization techniques}
\label{app:faultyCases}
In this section, we apply classic vortex visualization methods in some of the flows for which we have constructed approximate first integrals. First, we test two isosurface-based criteria that are widely used in the literature to locate vortical structures, i.e., the $\lambda_2$-criterion of \cite{hussain1995} and the $Q$-criterion of \cite{jcr1988eddies}. Both of them are defined from the decomposition
\begin{equation}
\nabla \mathbf{v} = \mathbf{S} + \boldsymbol{\Omega},
\label{eq:vel_grad}
\end{equation}
where $\mathbf{S}=\frac{1}{2}[\nabla \mathbf{v} + (\nabla \mathbf{v})^\top]$ is the rate-of-strain tensor and $\bm{\Omega}=\frac{1}{2}[\nabla \mathbf{v} - (\nabla \mathbf{v})^\top]$ is the vorticity tensor.

According to the $\lambda_2$-criterion, vortical regions coincide with domains where
\begin{equation}
\lambda_2(\mathbf{S}^2 + \boldsymbol{\Omega}^2) < 0
\label{eq:lambda2_cr}
\end{equation}
with $\lambda_2(\mathbf{B})$ denoting the intermediate eigenvalue of the symmetric tensor $\mathbf{B}$. Similarly, the $Q$-criterion identifies vortical regions as those where
\begin{equation}
Q = \frac{1}{2} \left[ \lVert \bm{\Omega} \rVert^2 - \lVert \mathbf{S} \rVert^2 \right] > 0
\label{eq:Q_cr}
\end{equation}
with $\lVert \mathbf{B} \rVert$ denoting the Frobenius norm of $\mathbf{B}$.

We stress here that both of these criteria are only Galilean-invariant and, thus, depend on the frame of reference they are used on. Of equal importance is that, according to their definitions, they highlight vortical \emph{regions} rather than surfaces. To bypass this, a specific value is usually chosen and the resulting isosurfaces are identified as the so-called vortices. Such a choice, however, would be justifiable, if the chosen isovalue was close to zero, as the original definitions postulated. Unfortunately, in typical flow visualizations, these isovalues are tuned to arbitrary values to match user expectations regarding the vortical features.

We highlight this in Fig.~\ref{fig:ABC_Lambda2_failure} using the $\lambda_2$-criterion for the non-integrable ABC flow.
Specifically, we use one value close to zero (Fig.~\ref{fig:ABC_Lambda2_failure}(a)) and one which corresponds to a drastically different topology for the resulting structures (Fig.~\ref{fig:ABC_Lambda2_failure}(b)).
We note that both of these values fail to capture even a single family of invariant tori.
Indeed, the reconstructed surfaces tend to either have holes in the vicinity of the invariant tori (Fig.~\ref{fig:ABC_Lambda2_failure}(a)) or be misaligned with the vortex cores (Fig.~\ref{fig:ABC_Lambda2_failure}(b)).
This is in stark contrast to the approximate-first-integral-based tori we have found in the subsection \ref{nonintegrable_ABC}.
\begin{figure}
	\centering
	\includegraphics[width=.48\linewidth]{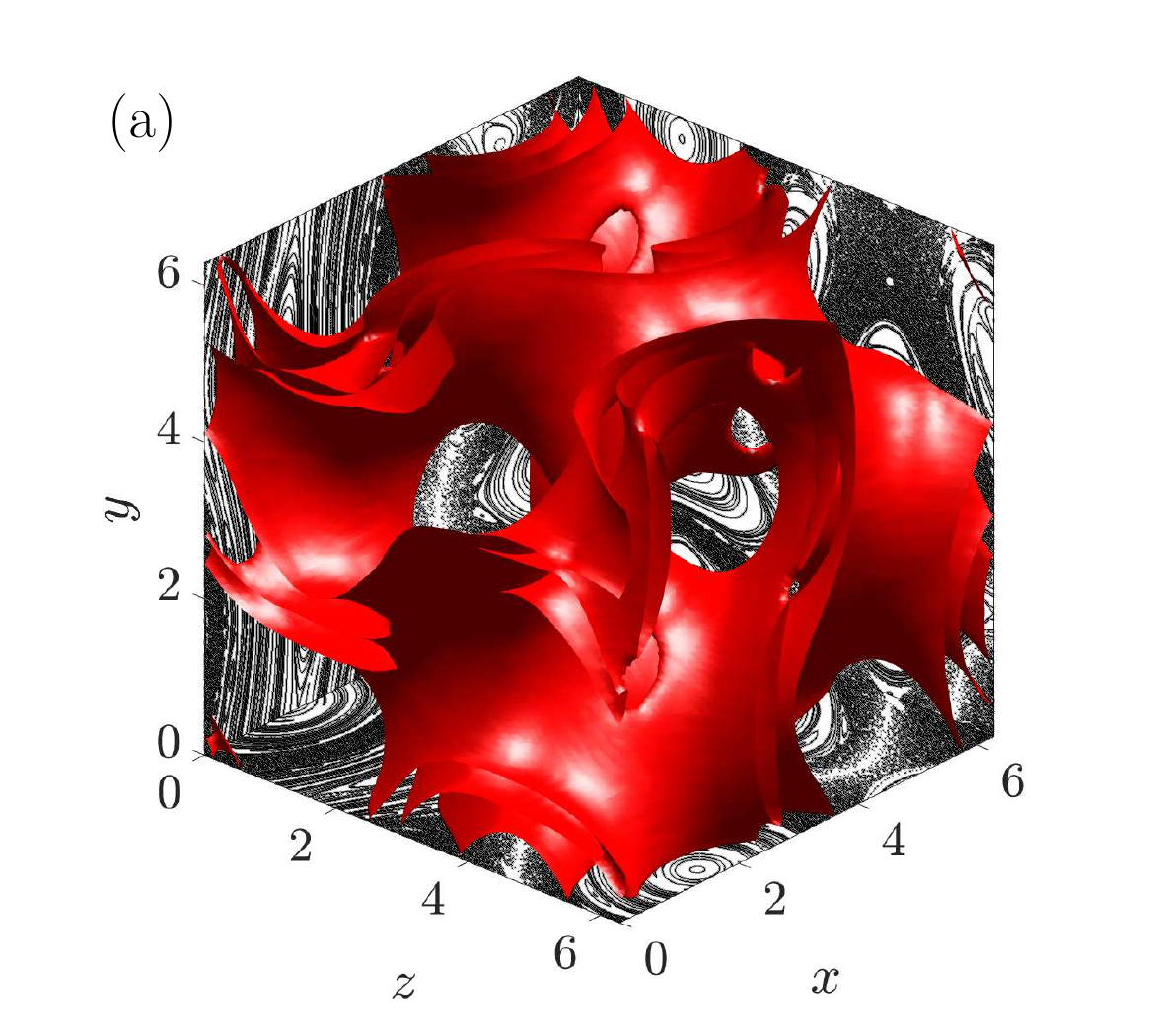}
	\includegraphics[width=.48\linewidth]{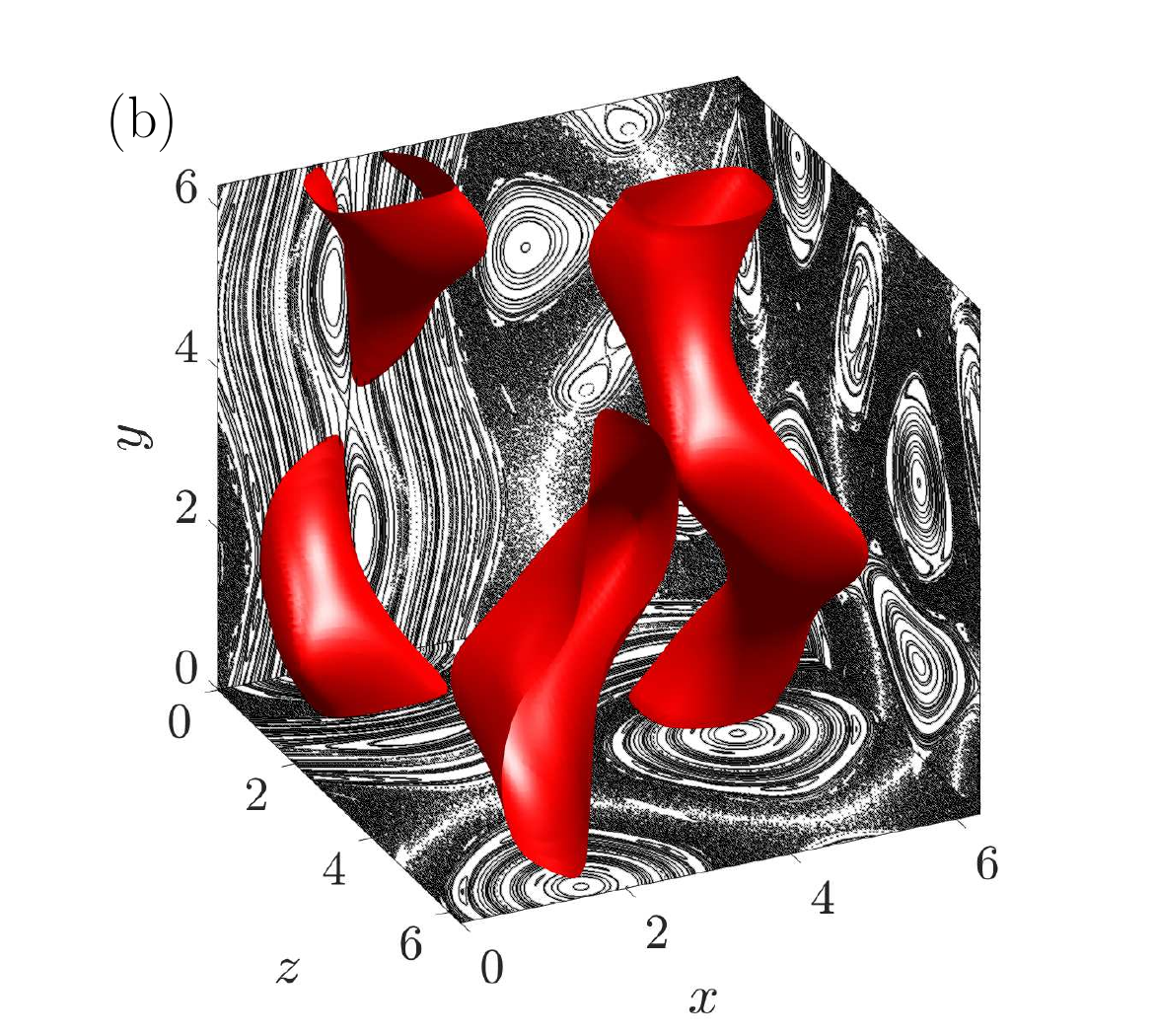}\\
	
	\vspace{-0.3cm}
	\caption{$\lambda_2$-criterion-based isosurfaces against Poincaré maps in the non-integrable ABC flow. The isosurfaces correspond to (a) $\lambda_2=-0.02$ and (b) $\lambda_2=-2.4$. The Poincaré maps are computed from a grid of $20 \times 20$ initial conditions on the $x=0$, $y=0$, and $z=0$ planes running up to arclength $10^4$.}
	\label{fig:ABC_Lambda2_failure}
\end{figure}

A similar conclusion can be drawn for the flow inside the V junction that we tested in the subsection \ref{V_junction}.
Our approximate-first-integral algorithm correctly pinpoints the exact position of the recirculation bubbles in this flow.
In contrast, we show in Fig.~\ref{fig:junction_Q_failure} the drastically different vortical features we obtain for different isovalues of the $Q$-criterion.
Despite the four recirculation bubbles that are formed downstream, the isosurfaces based on the $Q$-criterion locate invariably two vortical features whose length also varies substantially for different $Q$ values.
What is even more remarkable is that such Q-criterion-based vortical features persist even for similar flows with slightly different Reynolds numbers that exhibit no fluid recirculation at all (see Fig.~\ref{fig:junction_Q_failure}(c)). Importantly, our method finds no structure in the same spatial  domain for these Reynolds numbers  and hence correctly signals the lack of a recirculation bubble.
\begin{figure}
	\centering
	\includegraphics[width=.48\linewidth]{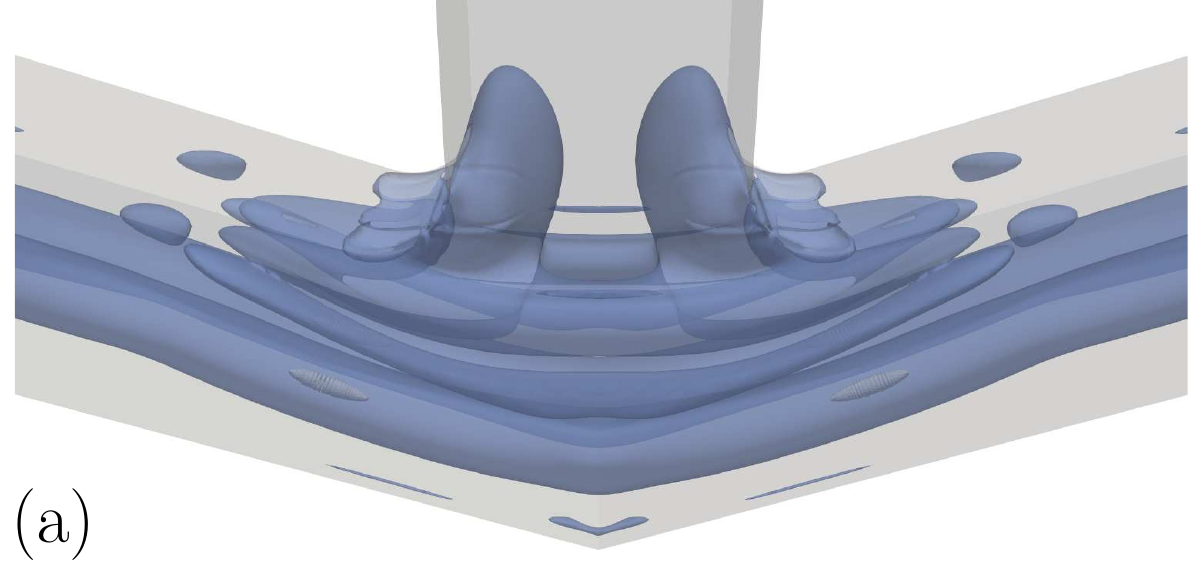}
	\includegraphics[width=.48\linewidth]{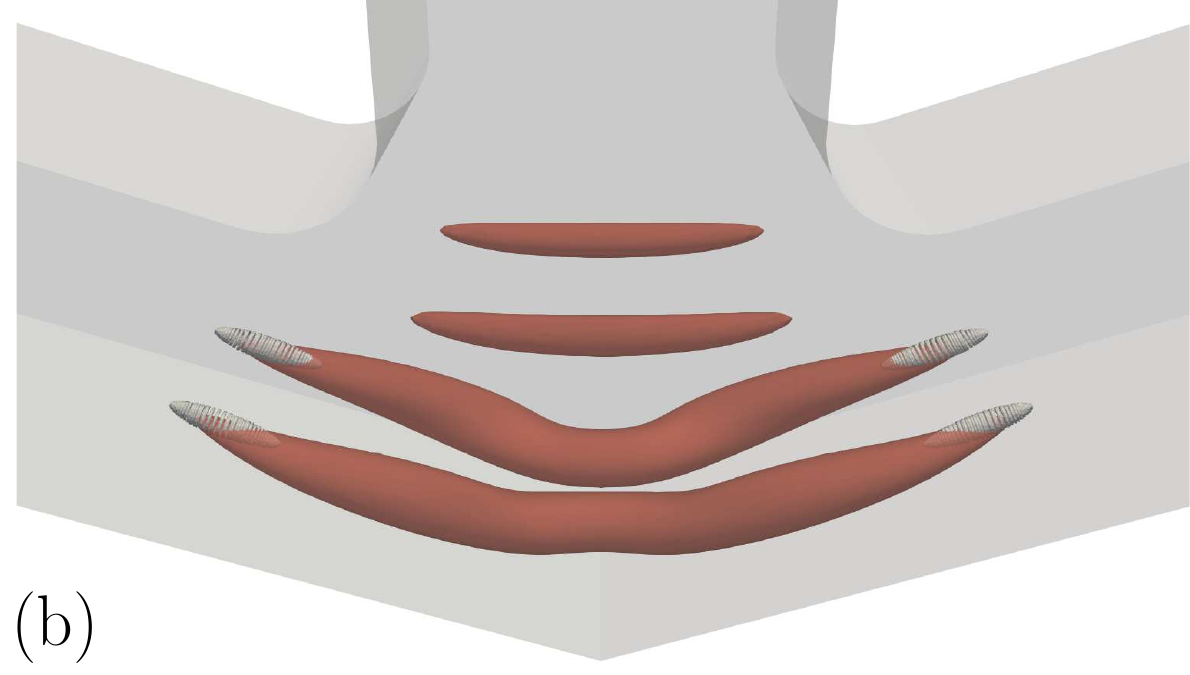}\\
	\includegraphics[width=.48\linewidth]{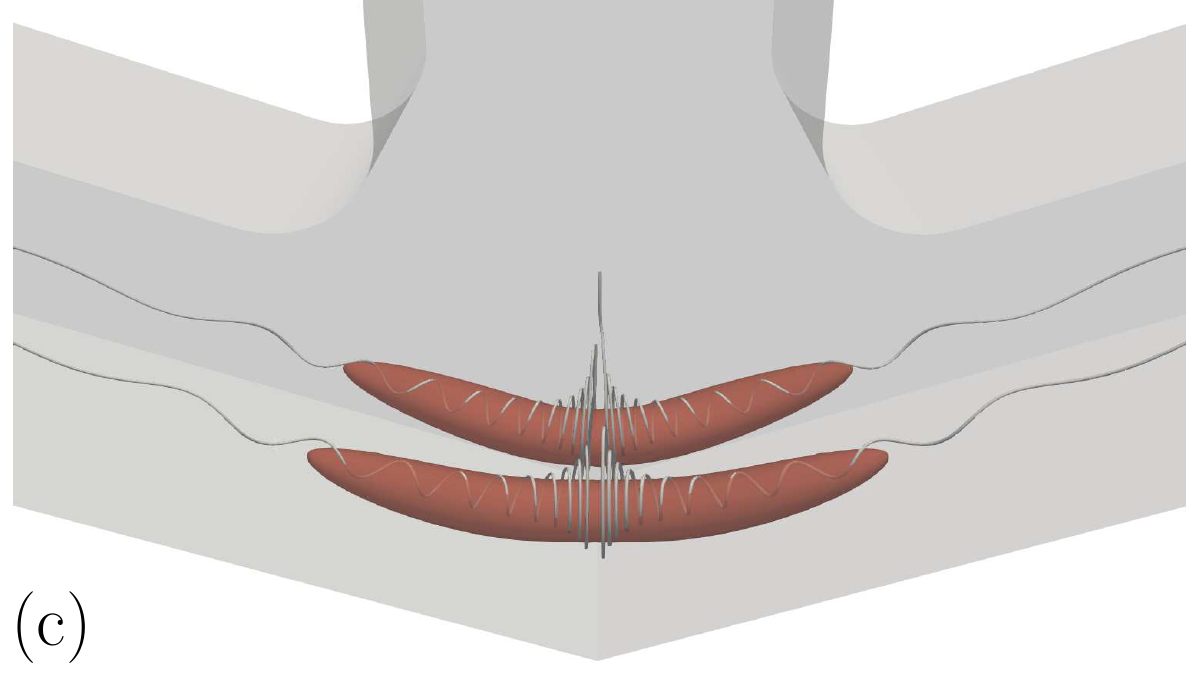}
	
	\vspace{-0.3cm}
	\caption{(a,b) $Q$-criterion-based isosurfaces against the recirculation bubbles for the flow inside the V junction of subsection \ref{V_junction} $(\mathrm{Re} = 230)$. The isosurfaces correspond to (a) $Q=0.02$ and (b) $Q=50$. (c) $Q$-criterion-based isosurfaces for $Q = 50$ for a perturbed solution $(\mathrm{Re} = 180)$ where no recirculation bubbles are formed.}
	\label{fig:junction_Q_failure}
\end{figure}

\begin{figure}
	\centering
	\includegraphics[width=0.45\linewidth]{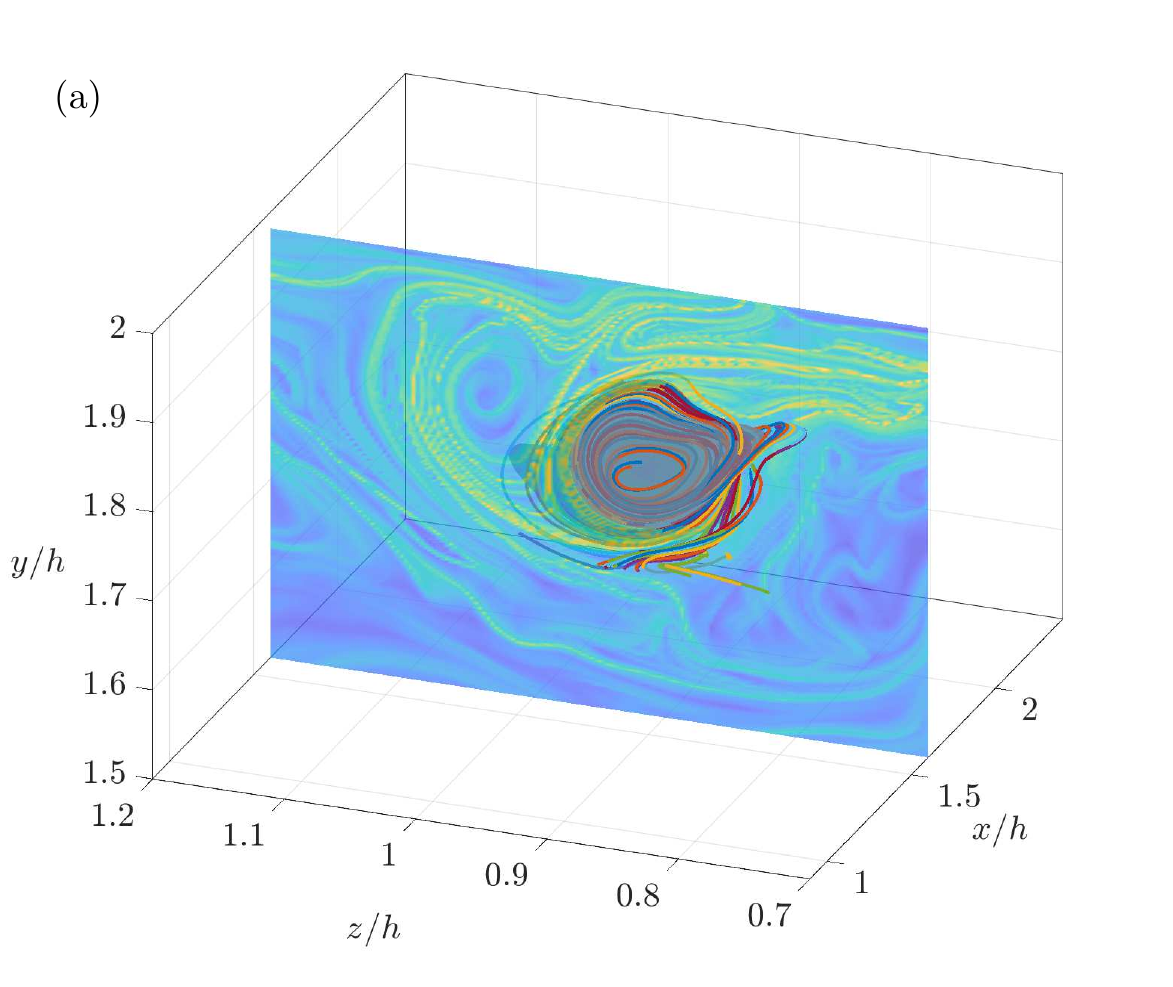}
	\includegraphics[width=0.45\linewidth]{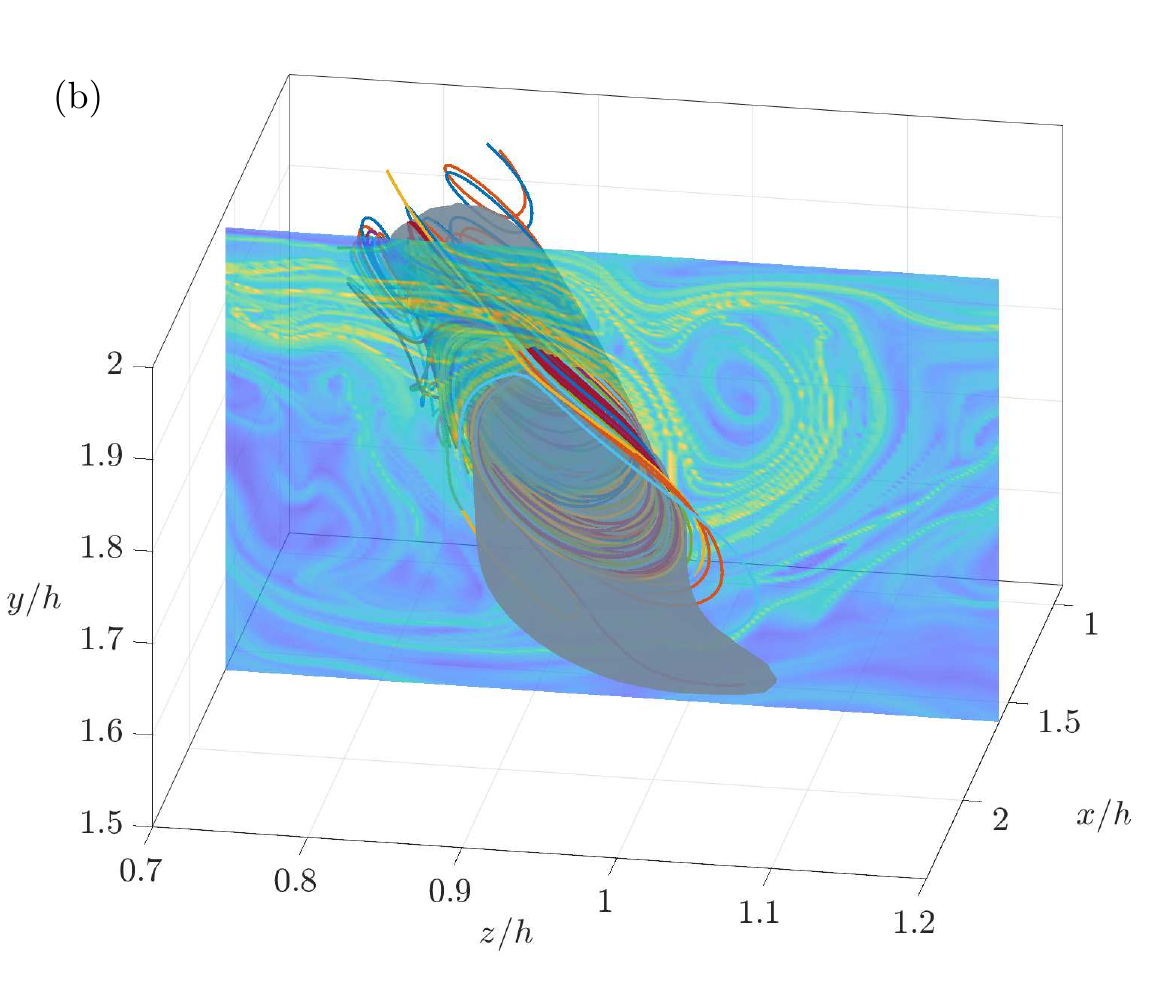}\\
	\includegraphics[width=0.45\linewidth]{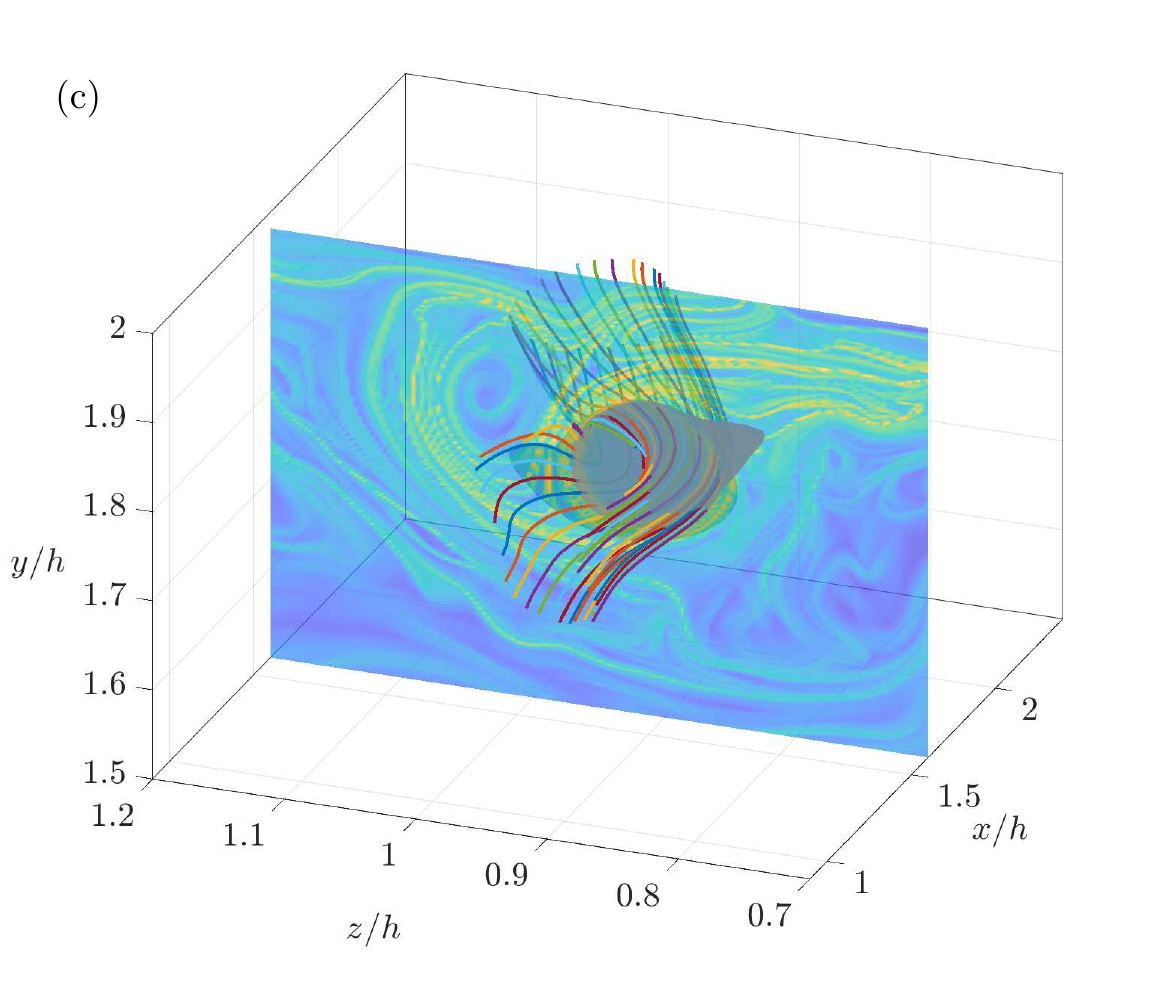}
	\includegraphics[width=0.45\linewidth]{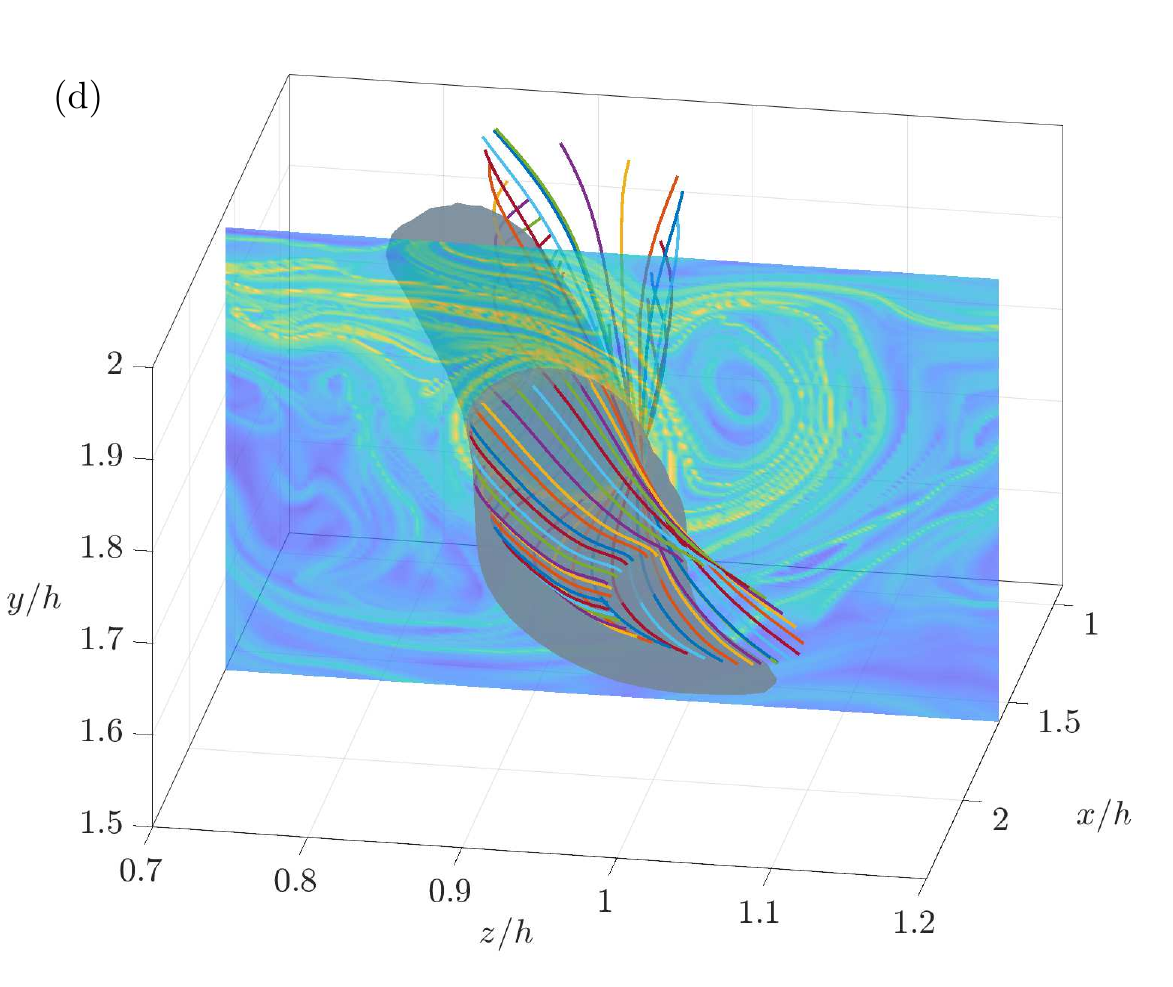}\\
	\includegraphics[width=0.45\linewidth]{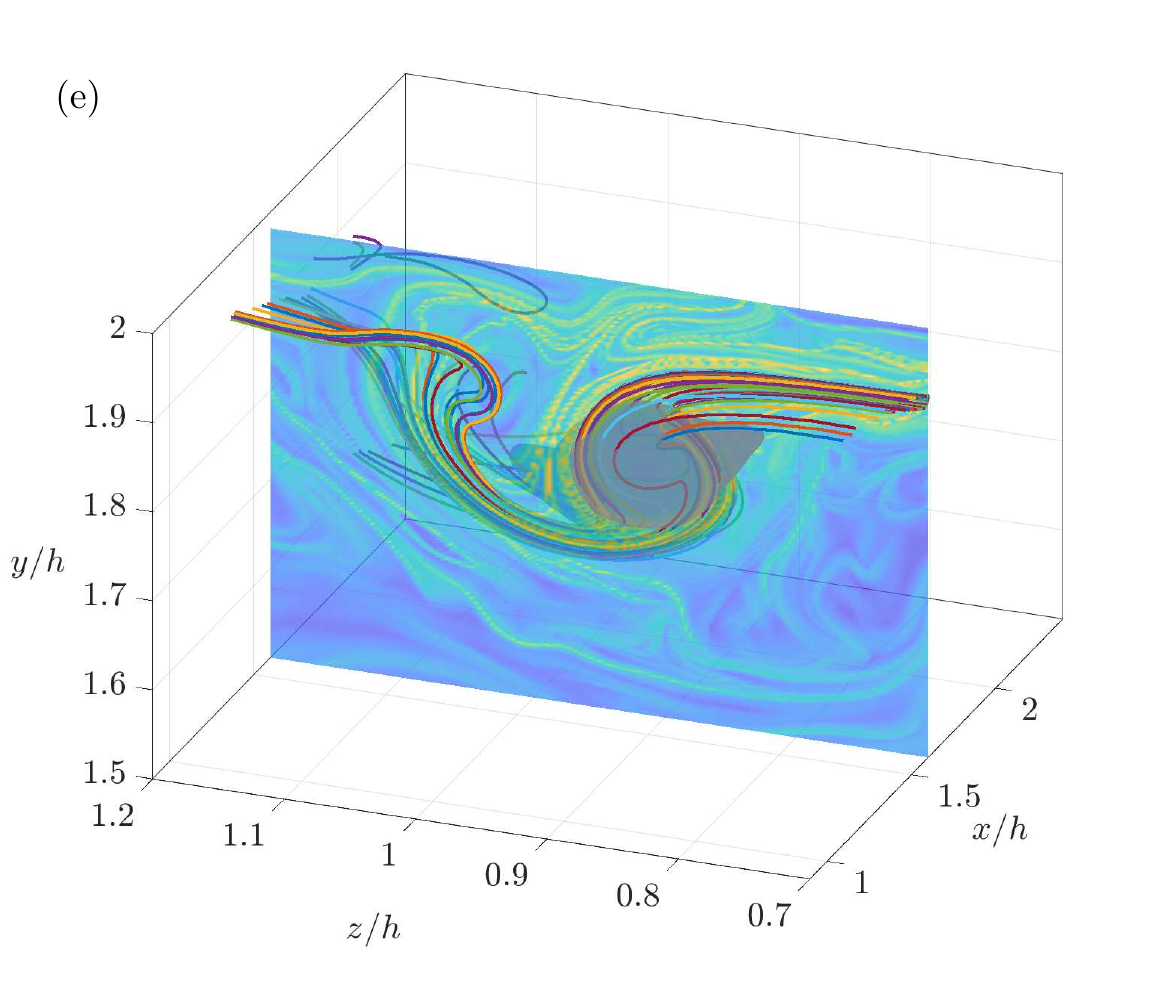}
	\includegraphics[width=0.45\linewidth]{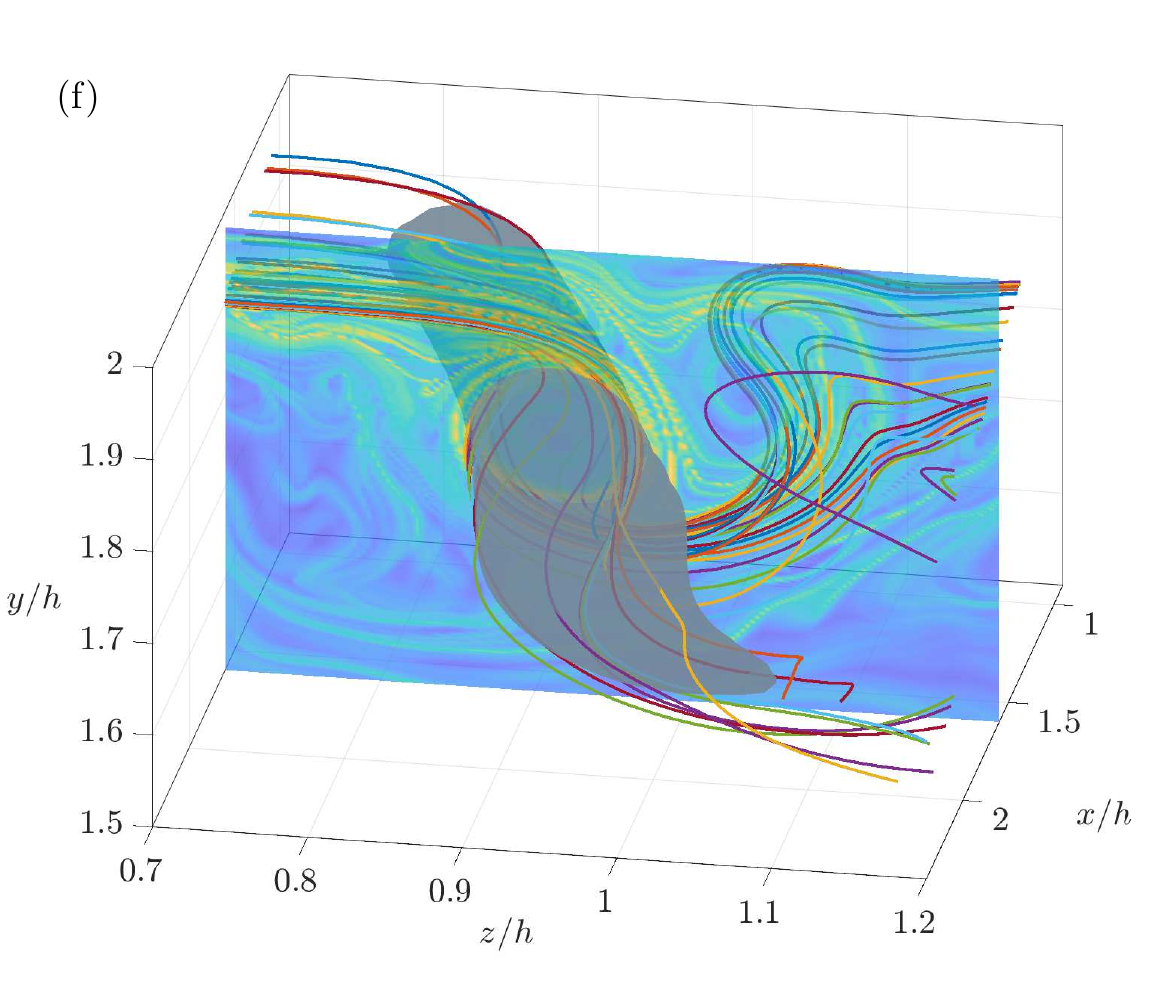}

	\caption{Streamlines of different vector fields emanating from points lying on the ridge of the aFTLE field at $x/h = 1.6$. Different views of streamlines of the momentum barrier field (a,b), velocity field (c,d) and vorticity field (e,f) at $t=0$.}
	\label{fig:d2vel_vel_omega}
\end{figure}

Finally, in Fig.~\ref{fig:d2vel_vel_omega}, we present representative streamlines for the turbulent channel flow of section \ref{turbulentChannel} based on the barrier field $\Delta \mathbf{v} (\mathbf{x})$ (a,b), the velocity field $\mathbf{v} (\mathbf{x})$ (c,d) and the vorticity field $\boldsymbol \omega (\mathbf{x})$ (e,f) at time $t=0$.
For all three fields we choose the same initial conditions, i.e., points that approximately lie on the ridge of the aFTLE field (computed based on the $\Delta \mathbf{v} (\mathbf{x})$ field) we identified in Fig.~\ref{fig:AFTLE_mush_cut}(c) at $x/h=1.6$.

Based on these calculations, we draw the following conclusions.
First, we observe the barrier streamlines being wrapped around the approximate-first-integral-based structure following even small protrusions on its surface like the one seen in Fig.~\ref{fig:d2vel_vel_omega}(a). This structure is, thus, a correct approximation to an invariant manifold of the barrier field in agreement with the best-fit streamsurface of Fig.~\ref{fig:streamsurfaces}(a,b) as well as the aFTLE landscape of Fig.~\ref{fig:AFTLE_mush_cut}. For longer integration dummy times (Fig.~\ref{fig:streamsurfaces}(c,d)), however, we observe that this delineation of the barrier surface using streamlines quickly comes to a halt. Indeed, the streamlines fail to capture a significant portion of the barrier surface as $x/h$ becomes larger, whereas for smaller $x/h$ they develop convoluted patterns that eventually fall apart resulting in elongated streamsurfaces that show no imprint on the aFTLE landscape. This constitutes the main reason why techniques aiming at the reconstruction of exact streamsurfaces through a better seed placement are not well-suited for such flows. 

Second, we note that the velocity streamlines form a misaligned (with respect to the extracted structure) tube that gives no indication of a vortical feature. Third, we remark that the vorticity streamlines correctly outline the outer shape of the mushroom-shaped structure imprinted on the aFTLE landscape without a detailed delineation, however, of the two branches that constitute it.

\begin{figure}
	\centering
	\includegraphics[width=0.45\linewidth]{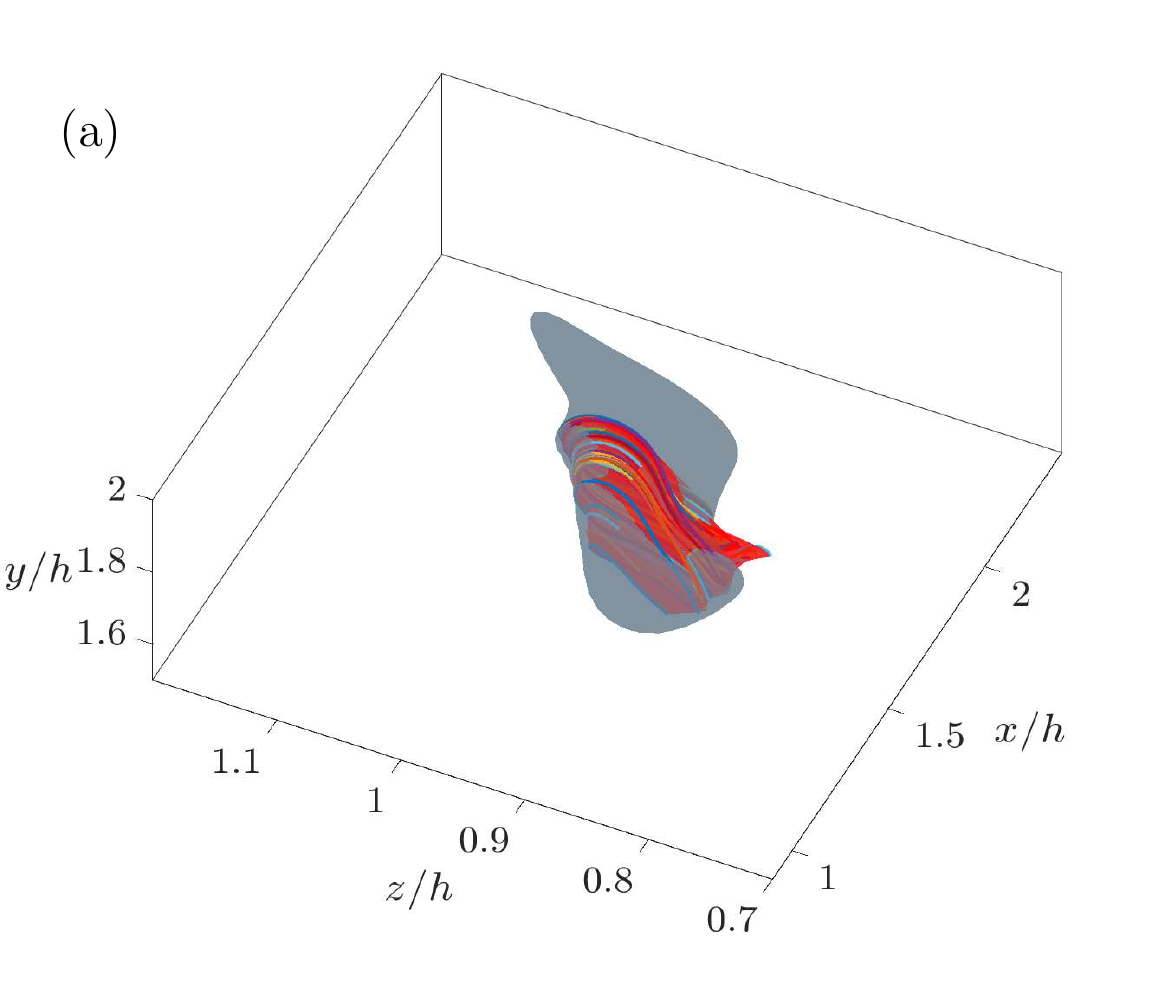}
	\includegraphics[width=0.45\linewidth]{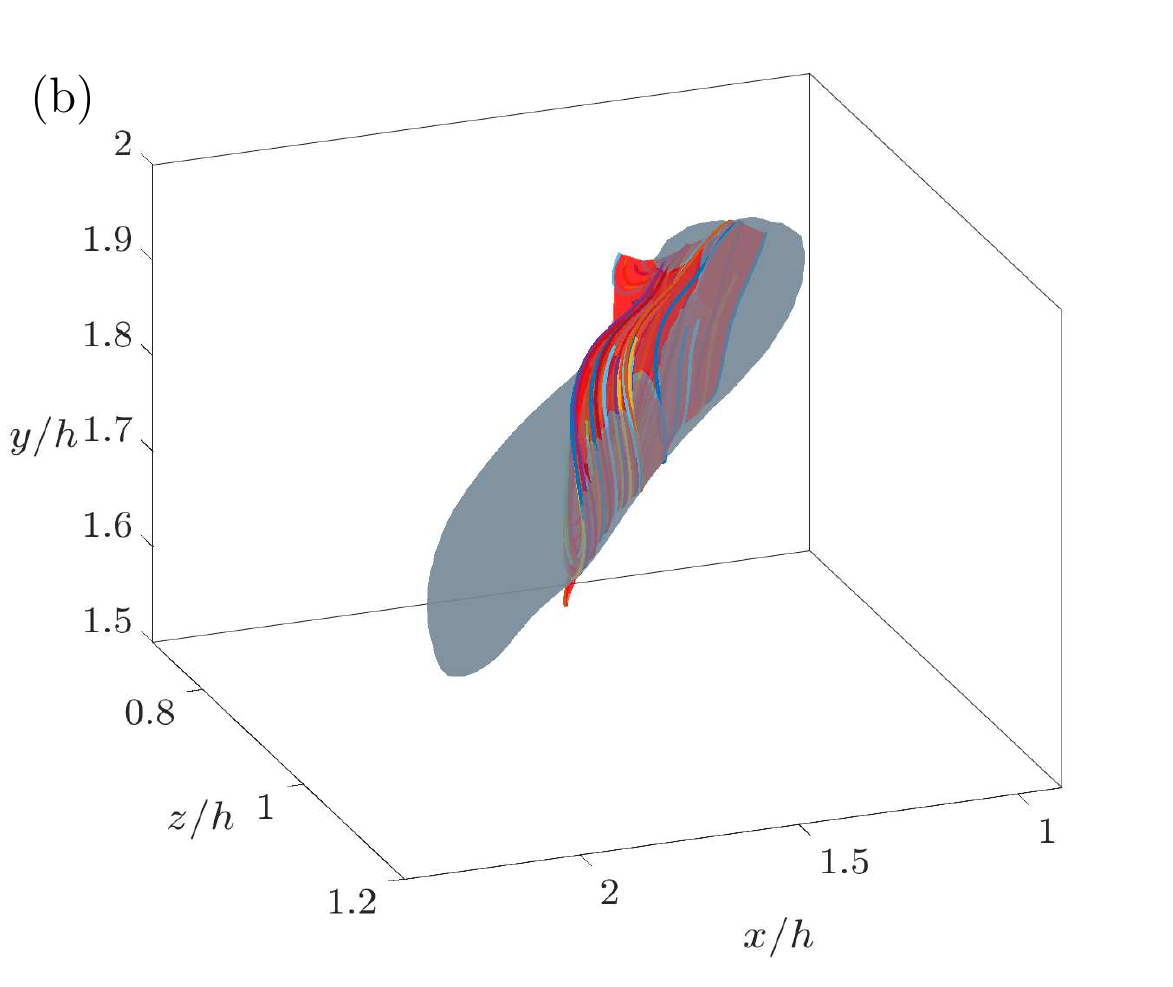}\\
	\includegraphics[width=0.45\linewidth]{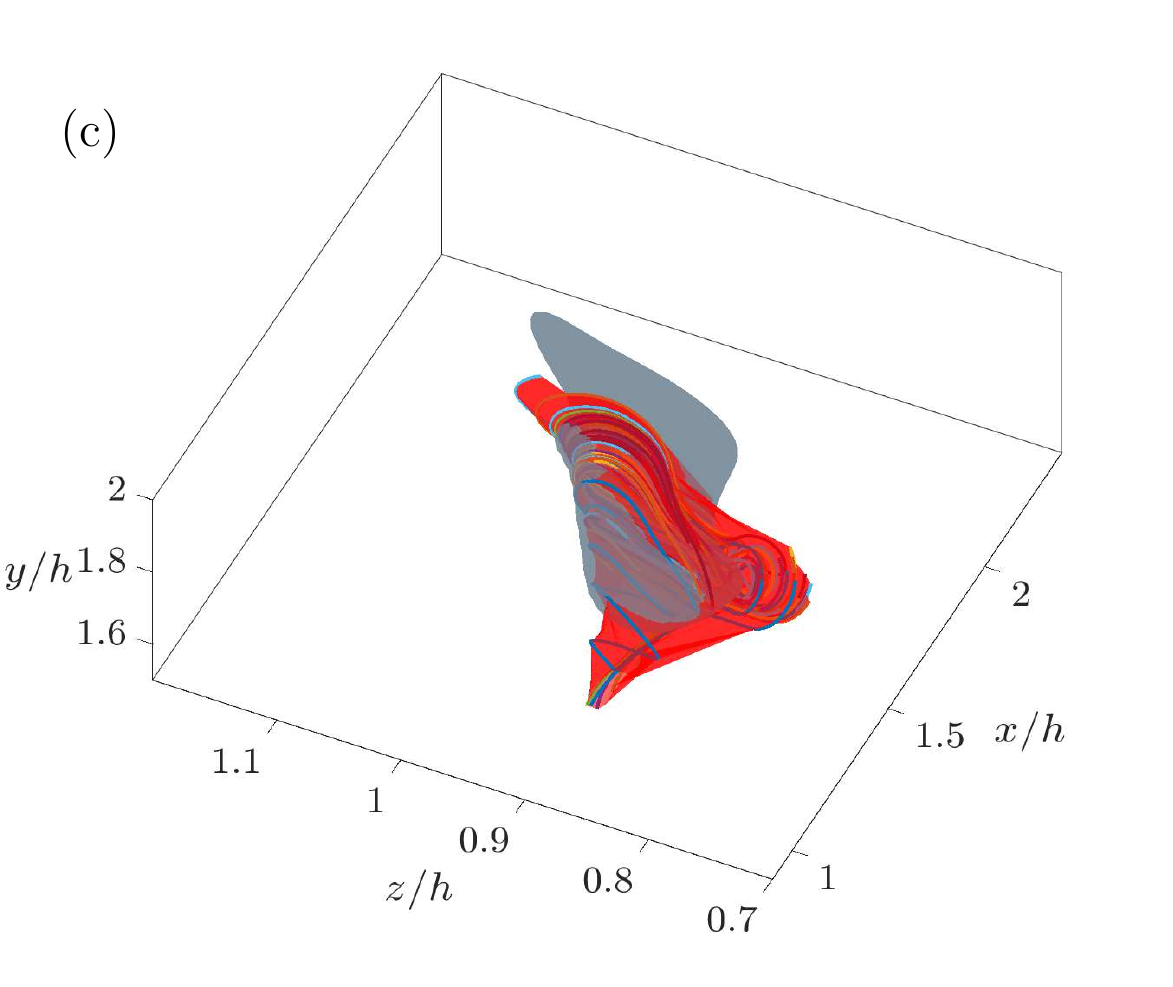}
	\includegraphics[width=0.45\linewidth]{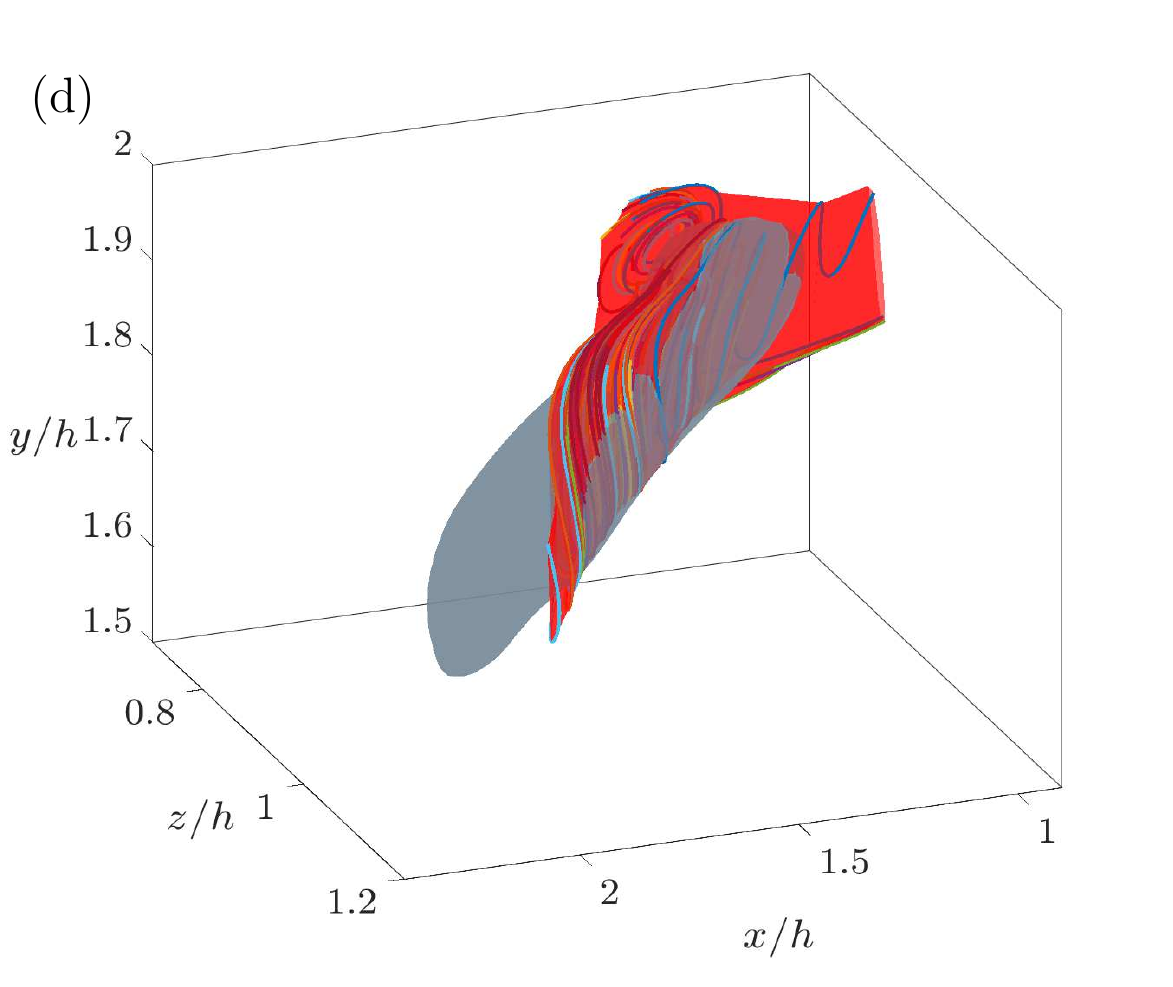}\\

\caption{Different views of streamlines of the momentum barrier field originating from points lying on the ridge of the aFTLE field presented in Fig.~\ref{fig:d2vel_vel_omega}. The integration in the subfigures (a,b) and (c,d) corresponds to different dummy final times $s_{1,(a,b)}$ and $s_{1,(c,d)}$ with $s_{1,(a,b)} < s_{1,(c,d)}$. The best-fit streamsurfaces are depicted in red and are compared against the approximate-first-integral-based structure in gray. }
	\label{fig:streamsurfaces}
\end{figure}

%% file: least_squares_homo.tex
\section{Least squares for homogeneous systems}
\label{app:least_squares_homo}

\begin{theorem}
Let $\mathbf{C}=\mathbf{U}\boldsymbol{\Sigma} \mathbf{V}^{*}$ be the singular value decomposition of a matrix $\mathbf{C}$. Let, also, $\mathbf{v}_{1}, \dotsc, \mathbf{v}_k$ be the last $k$ columns of $\mathbf{V}$ whose corresponding singular values are equal to the smallest singular value $\sigma_1$. Then, all the linear combinations of the form
\begin{equation}
\mathbf{x} = c_1 \mathbf{v}_{1} + \cdots + c_k \mathbf{v}_k
\end{equation}
with 
\begin{equation}
c_1^2 + \cdots + c_k^2 = 1
\end{equation}
are unit-norm least-squares solutions to the homogeneous system
\begin{equation}
\mathbf{C}\mathbf{x} = \mathbf{0}
\end{equation}
\end{theorem}
\begin{proof}
We want to find the solution $\mathbf{x}$ with $|\mathbf{x} | = 1$ that minimizes $| \mathbf{C} \mathbf{x} |$ or $| \mathbf{U} \boldsymbol{\Sigma} \mathbf{V}^{*} \mathbf{x}|$. Because $\mathbf{U}$ is unitary and, thus, acts as an isometry, this is equivalent to minimizing $| \boldsymbol{\Sigma} \mathbf{V}^{*} \mathbf{x}|$ or $|\boldsymbol{\Sigma} \mathbf{y}|$ with $\mathbf{y} = \mathbf{V}^{*} \mathbf{x}$. $\mathbf{V}$ is also unitary, so $|\mathbf{y} | = 1$ is equivalent to $|\mathbf{x} | = 1$. Therefore we want the unit-norm vector $\mathbf{y}$ that minimizes $|\boldsymbol{\Sigma} \mathbf{y}|$, i.e., the quantity
\begin{equation}
\sigma_1^2 y_1^2 + \cdots + \sigma_n^2 y_n^2.
\end{equation}
By the definition of SVD we have
\[
\left.\begin{array}{c}
\sigma_{1}\geq\sigma_{1}\Leftrightarrow\sigma_{1}^{2}\geq\sigma_{1}^{2}\\
\vdots\\
\sigma_{n}\geq\sigma_{1}\Leftrightarrow\sigma_{n}^{2}\geq\sigma_{1}^{2}
\end{array}\right\} \left.\Leftrightarrow\begin{array}{c}
\sigma_{1}^{2}y_{1}^{2}\geq\sigma_{1}^{2}y_{1}^{2}\\
\vdots\\
\sigma_{n}^{2}y_{n}^{2}\geq\sigma_{1}^{2}y_{n}^{2}
\end{array}\right\}
\]
and after summing up all the inequalities we obtain
\[
\sigma_{1}^{2}y_{1}^{2}+\sigma_{2}^{2}y_{2}^{2}+\ldots+\sigma_{n}^{2}y_{n}^{2}\geq\left(y_{1}^{2}+y_{2}^{2}+\ldots+y_{n}^{2}\right)\sigma_{1}^{2}=\sigma_{1}^{2},
\]
with the equality holding if
\begin{equation}
y_{n-k+1}=y_{n-k+2}=\ldots=y_{n}=0.
\label{eq:ls_homo_sol}
\end{equation}
We then have $\mathbf{x} = \mathbf{V} \mathbf{y} = y_1 \mathbf{v}_1 + \cdots + y_n \mathbf{v}_n$ and, thus, Eq. \ref{eq:ls_homo_sol} is equivalent to
\begin{equation}
\mathbf{x} = c_1 \mathbf{v}_{1} + \cdots + c_k \mathbf{v}_k
\end{equation}
with $c_1 = y_{1}, \dotsc, c_k=y_k$ and $c_1^2 + \cdots + c_k^2 = 1$.
\end{proof}

%% file: realConstraint.tex
\section{Solving for H after imposing the real constraint}
\label{app:realConstraint}
We define a complex vector $\mathbf{h} \in \mathbb{C}^n$, where $n$ represents the number of vectors $\mathbf{k} \in \mathbb{Z}^3$ such that $|\mathbf{k}|\leq N$ and we also represent the individual entries of $\mathbf{h}$ as $h_{\mathbf{k}}$. Further, in conjunction with Section \ref{algorithm_setup}, we define a $m\times n$ matrix, where m represents the total number of grid points. We are interested in minimizing $\mathbf{h}^{*} \mathbf{C}^{*} \mathbf{C} \mathbf{h}$ subject to the constraints $\mathbf{h}^{*} \mathbf{h} = 1$ and $h_{-\mathbf{k}}=h^*_{\mathbf{k}}$. The second constraint is not natural to solve using known optimization techniques, due to which we will modify the above problem rearranging $\mathbf{h}$ as a vector in  $\mathbb{R}^{2n}$.

Naively one can construct a column vector with first $n$ real entries of $\mathbf{h}$ followed by their corresponding imaginary entries
\begin{equation}
    \mathbf{v} = \begin{pmatrix} 
    \operatorname{Re}(\mathbf{h}) \\ \operatorname{Im}(\mathbf{h})
    \end{pmatrix}
\end{equation}
Similarly one can reorganize $\mathbf{A} = \mathbf{C}^{*} \mathbf{C}$ and obtain
\begin{equation}
    \mathbf{D} = \begin{bmatrix}
    \mathbf{A}_r & -\mathbf{A}_i \\
    \mathbf{A}_i & \mathbf{A}_r
    \end{bmatrix}
\end{equation}
where $\mathbf{A}_r =  \frac{\mathbf{A}+\mathbf{A}^T}{2}$ and $\mathbf{A}_i = \frac{\mathbf{A} - \mathbf{A}^T}{2i} $. Thus, we have $\mathbf{v}^T\mathbf{D}\mathbf{v}=\mathbf{h}^{*}\mathbf{A}\mathbf{h}$ subject to the constraints $\mathbf{v}^T\mathbf{v}=\mathbf{h}^{*} \mathbf{h} = 1$ and $\mathbf{U}\mathbf{v}=0$. In our case, $\mathbf{U}$ is a matrix of the form
\begin{align*} 
\mathbf{U} = \begin{bmatrix} 
\mathbf{K}_{(n\times n)} & \mathbf{J}_{(n\times n)}
\end{bmatrix}_{n \times 2n} \\
K_{ij} = \begin{cases} 1 & i=j,i<n/2,j<n/2 \\
-1 & i=n/2+j, i<n/2,j<n/2\\
0 
\end{cases} \bigg | & 
J_{ij} = \begin{cases} 1 & i=j,i<n/2,j<n/2 \\
1 & i=n/2+j, i<n/2,j<n/2\\
0 
\end{cases}
\end{align*}
The above optimization problem has a solution by \cite{golub1973some}. We state a form of the theorem that will be useful for us
\begin{theorem}
Minimizing $\mathbf{v}^T \mathbf{D}\mathbf{v}$ subject to $\mathbf{v}^T\mathbf{v}=1$,  $\mathbf{U}\mathbf{v}=0$, where $\mathbf{v}\in \mathbb{R}^{2n}$ is equivalent to minimizing $\mathbf{l}^T \mathbf{E} \mathbf{l}$ subject to $\mathbf{l}^T\mathbf{l}=1$, where $\mathbf{l}\in \mathbb{R}^n$ provided $\mathbf{U}$ is a $n\times 2n$ matrix.
\end{theorem}
\begin{proof}
Following \cite{golub1973some}, we can write 
\begin{equation}
    \mathbf{U}^T = \mathbf{QR} = \mathbf{Q}_{2n\times 2n}\begin{bmatrix} \mathbf{R}'_{n\times n} \\
    \mathbf{0}_{n\times n}
    \end{bmatrix}_{2n\times n},
\end{equation}
where $\mathbf{Q}$ is an orthogonal matrix and $\mathbf{R}'$ is an upper triangular matrix. The constraint then reads
\begin{align}
    \mathbf{R}^T\mathbf{Q}^T\mathbf{v}=0 \\
    \mathbf{R}^T\mathbf{v}_\mathbf{Q}=0.
\end{align}
The vector $\mathbf{v}_{\mathbf{Q}} = \begin{bmatrix} \mathbf{0}_{n\times 1} \\ \mathbf{l}_{n \times 1}
\end{bmatrix}$ is a solution to the above constraint. Expressing the objective function in terms of $\mathbf{v}_{\mathbf{Q}}$ we have
\begin{equation}
    \mathbf{v}^T\mathbf{D}\mathbf{v} = \mathbf{v}^T\mathbf{QQ}^{-1}\mathbf{D}(\mathbf{Q}^{-1})^T\mathbf{Q}^T\mathbf{v}=\mathbf{v}^T_{\mathbf{Q}}
    \begin{bmatrix}(\mathbf{Q}^T\mathbf{DQ})_{11} & (\mathbf{Q}^T\mathbf{DQ})_{12} \\
    (\mathbf{Q}^T\mathbf{DQ})_{21} & (\mathbf{Q}^T\mathbf{DQ})_{22}
    \end{bmatrix}\mathbf{v}_{\mathbf{Q}} = \mathbf{l}^T\mathbf{E}\mathbf{l} 
\end{equation}
subject to the constraint $\mathbf{l}^T\mathbf{l}=1$ with $\mathbf{E} = (\mathbf{Q}^T\mathbf{DQ})_{22}$. 
\end{proof}
Using this theorem we can see that the eigenvector of $\mathbf{E}$ corresponding to the minimum eigenvalue would be the solution (since $\mathbf{E}$ is a symmetric matrix). The vector of interest $\mathbf{h}$ can be reconstructed as follows
\begin{equation}
    \mathbf{h} = [\mathbf{Q}(\mathbf{v}_{\mathbf{Q}})_{min}]_{1} + i [\mathbf{Q}(\mathbf{v}_{\mathbf{Q}})_{min}]_2,\mathbf{Q}(\mathbf{v}_{\mathbf{Q}})_{min} = \begin{pmatrix}
    \{[\mathbf{Q}(\mathbf{v}_{\mathbf{Q}})_{min}]_1\}_{n\times 1} \\
    \{[\mathbf{Q}(\mathbf{v}_{\mathbf{Q}})_{min}]_2\}_{n\times 1} 
    \end{pmatrix}
\end{equation}
The results of this solution method for the non-integrable ABC flow on three different planes are depicted in Fig.~\ref{fig:closest_ABC_real}. Again, we note the very good agreement of the reconstructed level sets with both the Poincaré maps and the level sets of Fig.~\ref{fig:closest_ABC_eigen_vs_psvd}.
\begin{figure}
    \centering
    \includegraphics[]{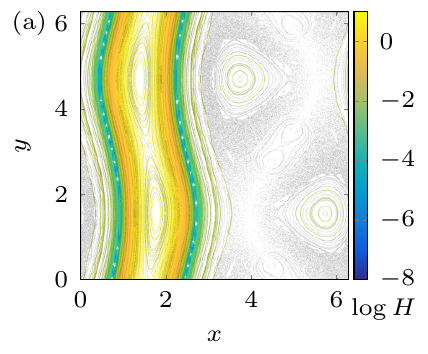}
    \includegraphics[]{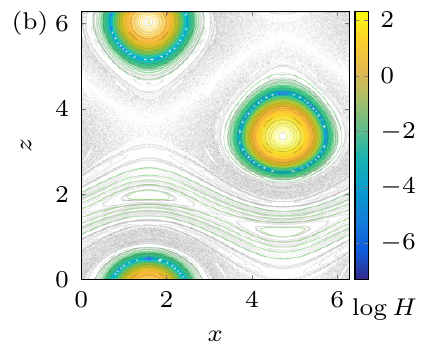}
    \includegraphics[]{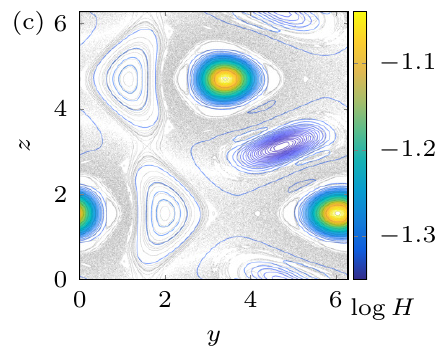}
    
    \vspace{-0.2cm}
    \caption{Analysis of the nonintegrable ABC flow using a computational grid of $100^3$ points and $9\,170$ Fourier modes after constraining $H$ to be a real scalar field. Level sets of the reconstructed first integral at $z=0$ (a), $y=0$ (b) and $x=0$ (c). The Poincaré map is overlaid on each section based on a uniform grid of $20 \times 20$ initial conditions.}
    \label{fig:closest_ABC_real}
\end{figure}

%% file: skinnySVD.tex

\section{Efficient computation of SVD for tall-skinny matrices}
\label{app:skinnySVD}
Let $\mathbf{C}$ be an $m \times n$ matrix with $m \gg n$. To make use of the tall-thin structure of this matrix, we follow an approach similar to the one described in \cite{schmidt2020survey} and factor
\begin{equation}
\mathbf{C}_{m\times n}=\mathbf{Q}_{m\times n} \mathbf{R}_{n\times n}
\end{equation}
using the so-called thin QR factorization of $\mathbf{C}$. The SVD of $\mathbf{R}$ gives
\begin{equation}
\mathbf{C}=\mathbf{Q} \mathbf{R}=\mathbf{Q}(\mathbf{U}_{R} \boldsymbol{\Sigma}_{R} \mathbf{V}_{R}^{T})=(\mathbf{Q} \mathbf{U}_{R}) \boldsymbol{\Sigma}_{R} \mathbf{V}_{R}^{T}
\end{equation}
This is in turn a singular value decomposition of $\mathbf{C}$ because
$\mathbf{Q} \mathbf{U}_{R}$ is unitary as the product of unitary matrices. To avail ourselves of this result, we can split $\mathbf{C} =\left[\begin{array}{c}
\mathbf{C}_{1}\\
\mathbf{C}_{2}
\end{array}\right]$, followed by $\mathbf{C}_{i}=\mathbf{Q}_{i}\mathbf{R}_{i} \mkern4mu (i=1,2)$ and, thus, we obtain
\begin{equation}
\mathbf{C}=\left[\begin{array}{cc}
\mathbf{Q}_{1} & 0\\
0 & \mathbf{Q}_{2}
\end{array}\right]\left[\begin{array}{c}
\mathbf{R}_{1}\\
\mathbf{R}_{2}
\end{array}\right]
\end{equation}
This is not yet a QR factorization of $\mathbf{C}$ as the right factor
is not upper triangular. We then perform another thin QR factorization
on $\left[\begin{array}{c}
\mathbf{R}_{1}\\
\mathbf{R}_{2}
\end{array}\right]$ and obtain
\vspace{-0.3cm}
\begin{equation}
\mathbf{C}=\left[\begin{array}{cc}
\mathbf{Q}_{1} & 0\\
0 & \mathbf{Q}_{2}
\end{array}\right] \mathbf{Q} \mathbf{R}
\label{eq:QR_final}
\end{equation}
which is a QR factorization of $\mathbf{C}$. Based on that, we can partition the original computation domain into smaller sub-domains, compute the $\mathbf{R}_i$ matrices for each of them, combine these to produce the final $\mathbf{R}$ matrix as in Eq. \ref{eq:QR_final} and, finally, run SVD on this $n \times n$ matrix. This procedure is readily parallelizable.

For the purposes of our computation described in the section \ref{nonintegrable_ABC}, we used $10$ cores on the Euler cluster of ETH Zurich and within each core we chose $5$ more partitions. This resulted in approximately the same memory footprint as for the solution of the eigenvalue problem. If the memory is not enough, the number of partitions can be increased with a subsequent increase in the computational time.

%% file: algorithm2.tex
\section{A different solution to the optimization problem}
\label{app:inhomoSol}
Let's assume that we start by fixing the Fourier coefficient
of particular modes. To avoid inducing unnecessary inhomogeneity to
our solution we are going to set $\hat{H}_{\mathbf{k}}=1$ for $\mathbf{k}=(1,0,0)$,
$\mathbf{k}=(0,1,0)$, $\mathbf{k}=(0,0,1)$. Similarly we impose
$\hat{H}_{\mathbf{-k}}=\hat{H}_{\mathbf{k}}^{*}=1$. We further denote
by
\begin{equation}
\begin{aligned}
\mathcal{K}= \{ 
\mathbf{k} |\mathbf{k}\in \mathbb{Z}^{3}\land |\mathbf{k}| \leq N \land & \mathbf{k}\neq(1,0,0)\land\mathbf{k}\neq(-1,0,0)\land\mathbf{k}\neq(0,1,0)\land \\         
& \mathbf{k}\neq(0,-1,0)\land\mathbf{k}\neq(0,0,1)\land\mathbf{k}\neq(0,0,-1) 
\}
\end{aligned}
\end{equation}
the set of the remaining Fourier modes. Then, the invariance
condition reads
\begin{equation}
\mathbf{C}\boldsymbol{\mathbf{h}}=\mathbf{b}
\end{equation}
where $C_{ij}=e^{i\mathbf{k}_{j}\cdot\mathbf{x}_{i}}\mathbf{k}_{j}\cdot\mathbf{v}_{i}$,
$\boldsymbol{\mathbf{h}}=\left\{ \hat{H}_{\mathbf{k}}\left|\mathbf{k}\in\right.\mathbb{Z}^{3}\land\mathbf{k}\in\mathcal{K}\right\} $
and $\mathbf{b}=\left[\begin{array}{cccccc}
b_{1} & b_{2} & \cdots & b_{l} & \cdots & b_{m}\end{array}\right]^{T}$ with
\begin{equation}
b_{l}=-2i\left(\sin x\mathbf{v}_{x}+\sin y\mathbf{v}_{y}+\sin z\mathbf{v}_{z}\right).
\end{equation}
Essentially, what we have attained here is to transform the homogeneous system of equations to an inhomogeneous one for which we can now use ordinary least squares or ridge regression. The ridge-regression solution is given as
\begin{equation}
\hat{\mathbf{h}} = \mathbf{V} \left( \boldsymbol{\Sigma}^{\top} \boldsymbol{\Sigma} + \lambda \mathbf{I}\right)^{-1} \boldsymbol{\Sigma}^{\top} \mathbf{U}^{*} \mathbf{b},
\end{equation}
whereas for $\lambda = 0$ we obtain the ordinary least-squares solution. This expression shows that the trick we employed in the homogeneous case with the QR decomposition cannot be applied here given that the solution depends on the explicit construction of $\mathbf{U}$. We are, therefore, limited with regard to the maximum size of the computational grid we can use before we run into memory issues.

Nonetheless, in Fig. \ref{fig:closest_ABC_regLS_vs_ridge} we provide the two solutions for the nonintegrable ABC flow using $60$ points per direction and the same number of Fourier modes as before. For the ridge regression we used cross-validation and kept the solution with the smallest least-squares error out of the solutions generated for $\lambda \in \{1,10^{-1},10^{-2},10^{-3},10^{-4}\}$. The results appear to be inferior to the homogeneous-system-based solutions suggesting that this approach is recommended only for numerical datasets defined over smaller grids.
\begin{figure}
	\centering
	\includegraphics[]{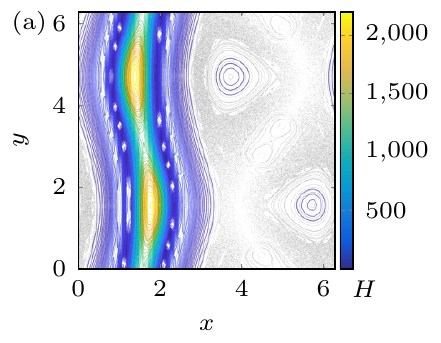}
	\includegraphics[]{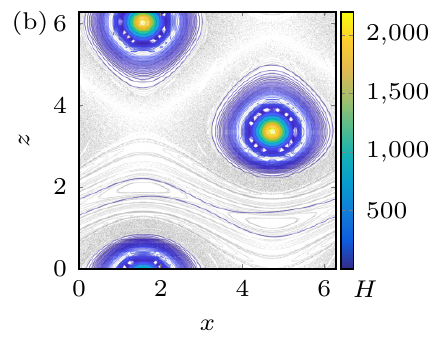}
	\includegraphics[]{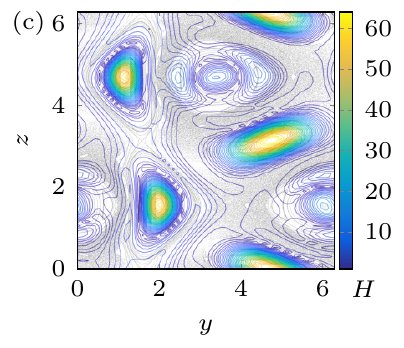}\\
	\includegraphics[]{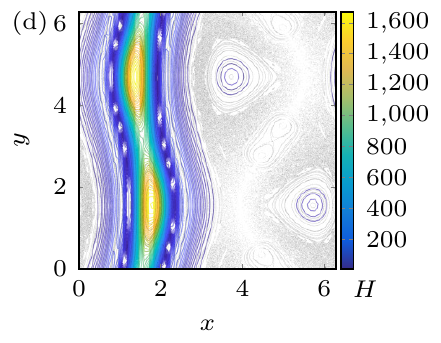}
	\includegraphics[]{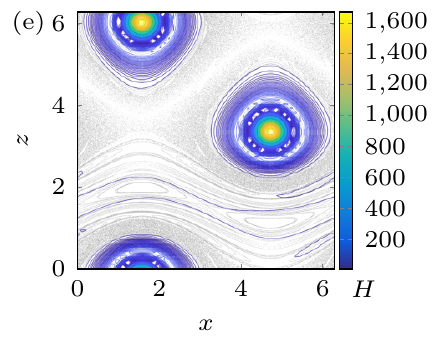}
	\includegraphics[]{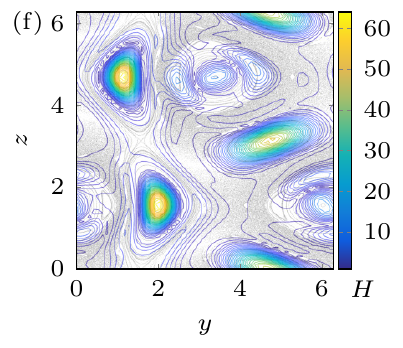}
	
	\vspace{-0.3cm}
	\caption{Analysis of the nonintegrable ABC flow using a computational grid of $60^3$ points and $9\,170$ Fourier modes. Level sets of the reconstructed first integral at $z=0$ (a and d), $y=0$ (b and e) and $x=0$ (c and f). The first row is constructed using ordinary least squares, whereas the second row is generated using ridge regression. The Poincaré map is overlaid on each section based on a uniform grid of $20 \times 20$ initial conditions.}
	\label{fig:closest_ABC_regLS_vs_ridge}
\end{figure}

%% file: closest.bbl
\providecommand{\noopsort}[1]{}\providecommand{\singleletter}[1]{#1}%
\begin{thebibliography}{64}
\expandafter\ifx\csname natexlab\endcsname\relax\def\natexlab#1{#1}\fi
\def\au#1{#1} \def\ed#1{#1} \def\yr#1{#1}\def\at#1{#1}\def\jt#1{\textit{#1}}
  \def\bt#1{#1}\def\bvol#1{\textbf{#1}} \def\vol#1{#1} \def\pg#1{#1}
  \def\publ#1{#1}\def\arxiv#1{#1}\def\org#1{#1}\def\st#1{\textit{#1}}

\bibitem[Aksamit \& Haller(2022)]{aksamit2022objective}
{\sc \au{Aksamit, N.~O.} \& \au{Haller, G.}} \yr{2022}  \at{Objective momentum
  barriers in wall turbulence}.  \jt{J. Fluid Mech.}  \bvol{941}.

\bibitem[Antuono(2020)]{antuono2020tri}
{\sc \au{Antuono, M.}} \yr{2020}  \at{Tri-periodic fully three-dimensional
  analytic solutions for the {N}avier--{S}tokes equations}.  \jt{J. Fluid
  Mech.}  \bvol{890}.

\bibitem[Arnold(1989)]{arnold89}
{\sc \au{Arnold, V.I.}} \yr{1989} {\em Mathematical {M}ethods of {C}lassical
  {M}echanics\/}.  \publ{Springer-Verlag, New York, NY}.

\bibitem[Arnold \& Khesin(1999)]{arnold1999topological}
{\sc \au{Arnold, V.~I.} \& \au{Khesin, B.~A.}} \yr{1999} {\em Topological
  methods in hydrodynamics\/}, ,  \vol{vol. 125}.  \publ{Springer Science \&
  Business Media}.

\bibitem[Ault {\em et~al.\/}(2016)Ault, Fani, Chen, Shin, Gallaire \&
  Stone]{ault2016vortex}
{\sc \au{Ault, J.~T.}, \au{Fani, A.}, \au{Chen, K.~K.}, \au{Shin, S.},
  \au{Gallaire, F.} \& \au{Stone, H.~A.}} \yr{2016}
  \at{Vortex-breakdown-induced particle capture in branching junctions}.
  \jt{Phys. Rev. Lett.}  \bvol{117}~(8),  \pg{084501}.

\bibitem[Batchelor(2000)]{batchelor2000introduction}
{\sc \au{Batchelor, G.K.}} \yr{2000} {\em An introduction to fluid dynamics\/}.
   \publ{Cambridge university press}.

\bibitem[Biskamp(2003)]{biskamp2003magnetohydrodynamic}
{\sc \au{Biskamp, D.}} \yr{2003} {\em Magnetohydrodynamic turbulence\/}.
  \publ{Cambridge University Press}.

\bibitem[Blazevski \& Haller(2014)]{blazevski2014hyperbolic}
{\sc \au{Blazevski, D.} \& \au{Haller, G.}} \yr{2014}  \at{Hyperbolic and
  elliptic transport barriers in three-dimensional unsteady flows}.
  \jt{Physica D: Nonlinear Phenomena}  \bvol{273},  \pg{46--62}.

\bibitem[Born {\em et~al.\/}(2010)Born, Wiebel, Friedrich, Scheuermann \&
  Bartz]{born10}
{\sc \au{Born, S.}, \au{Wiebel, A.}, \au{Friedrich, J.}, \au{Scheuermann, G.}
  \& \au{Bartz, D.}} \yr{2010}  \at{Illustrative stream surfaces}.  \jt{IEEE
  Transactions on Visualization and Computer Graphics}  \bvol{16}~(6),
  \pg{1329--1338}.

\bibitem[Carr {\em et~al.\/}(2003)Carr, Snoeyink \& Axen]{carr2003computing}
{\sc \au{Carr, H.}, \au{Snoeyink, J.} \& \au{Axen, U.}} \yr{2003}
  \at{Computing contour trees in all dimensions}.  \jt{Computational Geometry}
  \bvol{24}~(2),  \pg{75--94}.

\bibitem[Cheng \& Sun(1989)]{cheng89}
{\sc \au{Cheng, C.Q.} \& \au{Sun, Y.S.}} \yr{1989}  \at{Existence of invariant
  tori in three-dimensional measure-preserving mappings}.  \jt{Celestial Mech.
  Dyn. Astr.}  \bvol{47},  \pg{275--292}.

\bibitem[Chern {\em et~al.\/}(2017)Chern, Kn{\"o}ppel, Pinkall \&
  Schr{\"o}der]{chern2017inside}
{\sc \au{Chern, A.}, \au{Kn{\"o}ppel, F.}, \au{Pinkall, U.} \&
  \au{Schr{\"o}der, P.}} \yr{2017}  \at{Inside fluids: Clebsch maps for
  visualization and processing}.  \jt{ACM Transactions on Graphics (TOG)}
  \bvol{36}~(4),  \pg{1--11}.

\bibitem[Chierchia \& Gallavotti(1982)]{chierchia82}
{\sc \au{Chierchia, L.} \& \au{Gallavotti, G.}} \yr{1982}  \at{Smooth prime
  integrals for quasi-integrable {H}amiltonian systems}.  \jt{Il Nuovo Cimento
  B}  \bvol{67},  \pg{277--295}.

\bibitem[Chong {\em et~al.\/}(1990)Chong, Perry \& Cantwell]{chong1990general}
{\sc \au{Chong, M.~S.}, \au{Perry, A.~E.} \& \au{Cantwell, B.~J.}} \yr{1990}
  \at{A general classification of three-dimensional flow fields}.  \jt{Physics
  of Fluids A: Fluid Dynamics}  \bvol{2}~(5),  \pg{765--777}.

\bibitem[Dombre {\em et~al.\/}(1986)Dombre, Frisch, Greene, H{\'e}non, Mehr \&
  Soward]{dombre1986chaotic}
{\sc \au{Dombre, T.}, \au{Frisch, U.}, \au{Greene, J.~M.}, \au{H{\'e}non, M.},
  \au{Mehr, A.} \& \au{Soward, A.~M.}} \yr{1986}  \at{Chaotic streamlines in
  the {A}{B}{C} flows}.  \jt{J. Fluid Mech.}  \bvol{167},  \pg{353--391}.

\bibitem[Golub(1973)]{golub1973some}
{\sc \au{Golub, G.~H.}} \yr{1973}  \at{Some modified matrix eigenvalue
  problems}.  \jt{Siam Review}  \bvol{15}~(2),  \pg{318--334}.

\bibitem[Golub \& Pereyra(1973)]{golub1973differentiation}
{\sc \au{Golub, G.~H.} \& \au{Pereyra, V.}} \yr{1973}  \at{The differentiation
  of pseudo-inverses and nonlinear least squares problems whose variables
  separate}.  \jt{SIAM Journal on numerical analysis}  \bvol{10}~(2),
  \pg{413--432}.

\bibitem[Gottlieb \& Shu(1997)]{gottlieb1997gibbs}
{\sc \au{Gottlieb, D.} \& \au{Shu, C.-W.}} \yr{1997}  \at{On the {G}ibbs
  phenomenon and its resolution}.  \jt{SIAM review}  \bvol{39}~(4),
  \pg{644--668}.

\bibitem[Haller(2005)]{haller2005objective}
{\sc \au{Haller, G.}} \yr{2005}  \at{An objective definition of a vortex}.
  \jt{J. Fluid Mech.}  \bvol{525},  \pg{1--26}.

\bibitem[Haller(2015)]{haller2015lagrangian}
{\sc \au{Haller, G.}} \yr{2015}  \at{Lagrangian coherent structures}.
  \jt{Annu. Rev. Fluid Mech.}  \bvol{47},  \pg{137--162}.

\bibitem[Haller(2021)]{haller2021can}
{\sc \au{Haller, G.}} \yr{2021}  \at{Can vortex criteria be objectivized?}
  \jt{J. Fluid Mech.}  \bvol{908}.

\bibitem[Haller {\em et~al.\/}(2016)Haller, Hadjighasem, Farazmand \&
  Huhn]{haller16}
{\sc \au{Haller, G.}, \au{Hadjighasem, A.}, \au{Farazmand, M.} \& \au{Huhn,
  F.}} \yr{2016}  \at{Defining coherent vortices objectively from the
  vorticity}.  \jt{J. Fluid Mech.}  \bvol{795},  \pg{136--173}.

\bibitem[Haller {\em et~al.\/}(2018)Haller, Karrasch \&
  Kogelbauer]{haller2018material}
{\sc \au{Haller, G.}, \au{Karrasch, D.} \& \au{Kogelbauer, F.}} \yr{2018}
  \at{Material barriers to diffusive and stochastic transport}.
  \jt{Proceedings of the National Academy of Sciences}  \bvol{115}~(37),
  \pg{9074--9079}.

\bibitem[Haller {\em et~al.\/}(2020)Haller, Katsanoulis, Holzner, Frohnapfel \&
  Gatti]{haller2020objective}
{\sc \au{Haller, G.}, \au{Katsanoulis, S.}, \au{Holzner, M.}, \au{Frohnapfel,
  B.} \& \au{Gatti, D.}} \yr{2020}  \at{Objective barriers to the transport of
  dynamically active vector fields}.  \jt{J. Fluid Mech.}  \bvol{905}.

\bibitem[Haller \& Mezi{\' c}(1998)]{haller98}
{\sc \au{Haller, G.} \& \au{Mezi{\' c}, I.}} \yr{1998}  \at{Reduction of
  three-dimensional, volume-preserving flows with symmetry}.  \jt{Nonlinearity}
   \bvol{11}~(2),  \pg{319--339}.

\bibitem[Henon(1966)]{henon1966topologie}
{\sc \au{Henon, M.}} \yr{1966}  \at{Sur la topologie des lignes de courant dans
  un cas particulier}.  \jt{Comptes Rendus Acad. Sci. Paris A}  \bvol{262},
  \pg{312--314}.

\bibitem[Hill(1894)]{hill1894vi}
{\sc \au{Hill, M. J.~M.}} \yr{1894}  \at{On a spherical vortex}.  \jt{Philos.
  Trans. R. Soc. A}  \pg{pp. 213--245}.

\bibitem[Holmes(1984)]{holmes1984some}
{\sc \au{Holmes, P.}} \yr{1984}  \at{Some remarks on chaotic particle paths in
  time-periodic, three-dimensional swirling flows}.  \jt{Contemp. Math}
  \bvol{28},  \pg{393--404}.

\bibitem[Hultquist(1992)]{hultquist92}
{\sc \au{Hultquist, J.P.M.}} \yr{1992} Constructing stream surfaces in steady
  3d vector fields.  \bt{In {\em Proceedings Visualization '92\/}},  \pg{pp.
  171--178}.

\bibitem[Hunt {\em et~al.\/}(1988)Hunt, Wray \& Moin]{jcr1988eddies}
{\sc \au{Hunt, J.C.R.}, \au{Wray, A.} \& \au{Moin, P.}} \yr{1988}  \at{Eddies,
  stream, and convergence zones in turbulent flows}.  \jt{Center for turbulence
  research report CTR-S88}  \pg{pp. 193--208}.

\bibitem[Jeong \& Hussain(1995)]{hussain1995}
{\sc \au{Jeong, J.} \& \au{Hussain, F.}} \yr{1995}  \at{On the identification
  of a vortex}.  \jt{J. Fluid Mech.}  \bvol{285},  \pg{69--94}.

\bibitem[Katsanoulis {\em et~al.\/}(2020)Katsanoulis, Farazmand, Serra \&
  Haller]{2020vortex}
{\sc \au{Katsanoulis, S.}, \au{Farazmand, M.}, \au{Serra, M.} \& \au{Haller,
  G.}} \yr{2020}  \at{Vortex boundaries as barriers to diffusive vorticity
  transport in two-dimensional flows}.  \jt{Phys. Rev. Fluids}  \bvol{5}~(2),
  \pg{024701}.

\bibitem[Katsanoulis \& Haller(2019)]{barrierToolManual}
{\sc \au{Katsanoulis, S.} \& \au{Haller, G.}} \yr{2019} Barrier{T}ool {M}anual.
  Retrieved from \url{https://github.com/LCSETH}.

\bibitem[Kim {\em et~al.\/}(1987)Kim, Moin \& Moser]{kim1987turbulence}
{\sc \au{Kim, J.}, \au{Moin, P.} \& \au{Moser, R.}} \yr{1987}  \at{Turbulence
  statistics in fully developed channel flow at low {R}eynolds number}.  \jt{J.
  Fluid Mech.}  \bvol{177},  \pg{133--166}.

\bibitem[Llibre \& Valls(2012)]{llibre2012note}
{\sc \au{Llibre, J.} \& \au{Valls, C.}} \yr{2012}  \at{A note on the first
  integrals of the {A}{B}{C} system}.  \jt{Journal of mathematical physics}
  \bvol{53}~(2),  \pg{023505}.

\bibitem[Lorensen \& Cline(1987)]{lorensen1987marching}
{\sc \au{Lorensen, W.~E.} \& \au{Cline, H.~E.}} \yr{1987}  \at{Marching cubes:
  A high resolution 3d surface construction algorithm}.  \jt{SIGGRAPH}
  \bvol{21}~(4),  \pg{163--169}.

\bibitem[{Martinez--Esturo} {\em et~al.\/}(2013){Martinez--Esturo}, Schulze,
  R{\"o}ssl \& Theisel]{martinez13}
{\sc \au{{Martinez--Esturo}, J.}, \au{Schulze, M.}, \au{R{\"o}ssl, C.} \&
  \au{Theisel, H.}} \yr{2013}  \at{Global selection of stream surfaces}.
  \jt{Comput. Graph. Forum (Proc. Eurographics)}  \bvol{32}~(2),
  \pg{113--122}.

\bibitem[Martínez {\em et~al.\/}(2021)Martínez, Merino, Santoro, Martinez,
  Katsanoulis, Ault, Mayoral, Vazquez, Acolla, Dazzi, Mathurin, Borondics,
  Blázquez, Shauloff, Lebrón-Aguilar, Quintanilla-López, Jelinek,
  Cernicharo, Stone, O'Shea, de~Andres, Haller, Ellis \&
  Martín-Gago]{chem2021}
{\sc \au{Martínez, L.}, \au{Merino, P.}, \au{Santoro, G.}, \au{Martinez, J.},
  \au{Katsanoulis, S.}, \au{Ault, J.}, \au{Mayoral, Á.}, \au{Vazquez, L.},
  \au{Acolla, M.}, \au{Dazzi, A.}, \au{Mathurin, J.}, \au{Borondics, F.},
  \au{Blázquez, E.}, \au{Shauloff, N.}, \au{Lebrón-Aguilar, R.},
  \au{Quintanilla-López, J.}, \au{Jelinek, R.}, \au{Cernicharo, J.},
  \au{Stone, H.}, \au{O'Shea, V.}, \au{de~Andres, P.}, \au{Haller, G.},
  \au{Ellis, G.} \& \au{Martín-Gago, J.}} \yr{2021}  \at{Metal-catalyst-free
  gas-phase synthesis of long-chain hydrocarbons}.  \jt{Nat. Commun.}
  \bvol{12}~(5937).

\bibitem[Oettinger {\em et~al.\/}(2018)Oettinger, Ault, Stone \&
  Haller]{oettinger2018invisible}
{\sc \au{Oettinger, D.}, \au{Ault, J.T.}, \au{Stone, H.~A.} \& \au{Haller, G.}}
  \yr{2018}  \at{Invisible anchors trap particles in branching junctions}.
  \jt{Phys. Rev. Lett.}  \bvol{121}~(5),  \pg{054502}.

\bibitem[Oettinger {\em et~al.\/}(2016)Oettinger, Blazevski \&
  Haller]{oettinger2016global}
{\sc \au{Oettinger, D.}, \au{Blazevski, D.} \& \au{Haller, G.}} \yr{2016}
  \at{Global variational approach to elliptic transport barriers in three
  dimensions}.  \jt{Chaos}  \bvol{26}~(3),  \pg{033114}.

\bibitem[Oettinger \& Haller(2016)]{oettinger2016autonomous}
{\sc \au{Oettinger, D.} \& \au{Haller, G.}} \yr{2016}  \at{An autonomous
  dynamical system captures all lcss in three-dimensional unsteady flows}.
  \jt{Chaos}  \bvol{26}~(10),  \pg{103111}.

\bibitem[Peikert \& Sadlo(2007)]{peikert2007visualization}
{\sc \au{Peikert, R.} \& \au{Sadlo, F.}} \yr{2007} Visualization methods for
  vortex rings and vortex breakdown bubbles.  \bt{In {\em Proceedings of the
  9th Joint Eurographics/IEEE VGTC conference on Visualization\/}},  \pg{pp.
  211--218}.

\bibitem[Peikert \& Sadlo(2009)]{peikert2009topologically}
{\sc \au{Peikert, R.} \& \au{Sadlo, F.}} \yr{2009} Topologically relevant
  stream surfaces for flow visualization.  \bt{In {\em Proceedings of the 25th
  Spring Conference on Computer Graphics\/}},  \pg{pp. 35--42}.

\bibitem[Peng \& Yang(2018)]{peng2018effects}
{\sc \au{Peng, Naifu} \& \au{Yang, Yue}} \yr{2018}  \at{Effects of the mach
  number on the evolution of vortex-surface fields in compressible taylor-green
  flows}.  \jt{Physical Review Fluids}  \bvol{3}~(1),  \pg{013401}.

\bibitem[P{\"o}schel(1982)]{poschel82}
{\sc \au{P{\"o}schel, J.}} \yr{1982}  \at{Integrability of {H}amiltonian
  systems on {C}antor sets}.  \jt{Comm. Pure and Appl. Math.}  \bvol{35},
  \pg{653--696}.

\bibitem[Pullin \& Yang(2014)]{pullin2014whither}
{\sc \au{Pullin, D.I.} \& \au{Yang, Y.}} \yr{2014}  \at{Whither vortex tubes?}
  \jt{Fluid Dynamics Research}  \bvol{46}~(6),  \pg{061418}.

\bibitem[Robinson(1991)]{robinson1991coherent}
{\sc \au{Robinson, S.~K.}} \yr{1991}  \at{Coherent motions in the turbulent
  boundary layer}.  \jt{Annu. Rev. Fluid Mech.}  \bvol{23}~(1),  \pg{601--639}.

\bibitem[Sadlo \& Peikert(2007)]{sadlo07}
{\sc \au{Sadlo, F.} \& \au{Peikert, R.}} \yr{2007}  \at{Efficient visualization
  of {L}agrangian coherent structures by filtered {AMR} ridge extraction}.
  \jt{IEEE Transactions on Visualization and Computer Graphics}  \bvol{13}~(6),
   \pg{1456--1463}.

\bibitem[Schmidt(2020)]{schmidt2020survey}
{\sc \au{Schmidt, D.}} \yr{2020} A survey of {S}ingular {V}alue {D}ecomposition
  methods for distributed tall/skinny data.  \bt{In {\em 2020 IEEE/ACM 11th
  Workshop on Latest Advances in Scalable Algorithms for Large-Scale Systems
  (ScalA)\/}},  \pg{pp. 27--34}. IEEE.

\bibitem[Schulze {\em et~al.\/}(2014)Schulze, Martinez-Esturo, G{\"u}nther,
  R{\'o}ssl, Seidel, Weinkauf \& Theisel]{schulze14}
{\sc \au{Schulze, M.}, \au{Martinez-Esturo, J.}, \au{G{\"u}nther, T.},
  \au{R{\'o}ssl, C.}, \au{Seidel, H.-P.}, \au{Weinkauf, T.} \& \au{Theisel,
  H.}} \yr{2014}  \at{Sets of globally optimal stream surfaces for flow
  visualization}.  \jt{Computer Graphics Forum (Proc. EuroVis)}  \bvol{33}~(3),
   \pg{1--10}.

\bibitem[Shin {\em et~al.\/}(2015)Shin, Ault \& Stone]{shin2015flow}
{\sc \au{Shin, S.}, \au{Ault, J.} \& \au{Stone, H.}} \yr{2015}  \at{Flow-driven
  rapid vesicle fusion via vortex trapping}.  \jt{Langmuir}  \bvol{31}~(26),
  \pg{7178--7182}.

\bibitem[Sotiropoulos {\em et~al.\/}(2001)Sotiropoulos, Ventikos \&
  Lackey]{sotiropoulos2001chaotic}
{\sc \au{Sotiropoulos, F.}, \au{Ventikos, Y.} \& \au{Lackey, T.}} \yr{2001}
  \at{Chaotic advection in three-dimensional stationary vortex-breakdown
  bubbles: {\v{S}}il'nikov's chaos and the devil's staircase}.  \jt{J. Fluid
  Mech.}  \bvol{444},  \pg{257--297}.

\bibitem[Van~Wijk(1993)]{van1993implicit}
{\sc \au{Van~Wijk, J.~J.}} \yr{1993} Implicit stream surfaces.  \bt{In {\em
  Proceedings Visualization'93\/}},  \pg{pp. 245--252}. IEEE.

\bibitem[Vigolo {\em et~al.\/}(2014)Vigolo, Radl \&
  Stone]{vigolo2014unexpected}
{\sc \au{Vigolo, D.}, \au{Radl, S.} \& \au{Stone, H.~A.}} \yr{2014}
  \at{Unexpected trapping of particles at a {T} junction}.  \jt{Proc. Natl.
  Acad. Sci.}  \bvol{111}~(13),  \pg{4770--4775}.

\bibitem[Weller {\em et~al.\/}(1998)Weller, Tabor, Jasak \&
  Fureby]{weller1998tensorial}
{\sc \au{Weller, H.~G.}, \au{Tabor, G.}, \au{Jasak, H.} \& \au{Fureby, C.}}
  \yr{1998}  \at{A tensorial approach to computational continuum mechanics
  using object-oriented techniques}.  \jt{Computers in physics}  \bvol{12}~(6),
   \pg{620--631}.

\bibitem[Xiong \& Yang(2017)]{xiong2017boundary}
{\sc \au{Xiong, S.} \& \au{Yang, Y.}} \yr{2017}  \at{The boundary-constraint
  method for constructing vortex-surface fields}.  \jt{J. Comput. Phys.}
  \bvol{339},  \pg{31--45}.

\bibitem[Xiong \& Yang(2019)]{xiong2019identifying}
{\sc \au{Xiong, S.} \& \au{Yang, Y.}} \yr{2019}  \at{Identifying the tangle of
  vortex tubes in homogeneous isotropic turbulence}.  \jt{J. Fluid Mech.}
  \bvol{874},  \pg{952--978}.

\bibitem[Yang \& Pullin(2011)]{yang2011evolution}
{\sc \au{Yang, Y.} \& \au{Pullin, D.I.}} \yr{2011}  \at{Evolution of
  vortex-surface fields in viscous taylor--green and kida--pelz flows}.  \jt{J.
  Fluid Mech.}  \bvol{685},  \pg{146--164}.

\bibitem[Yang \& Pullin(2010)]{yang2010lagrangian}
{\sc \au{Yang, Y.} \& \au{Pullin, D.~I.}} \yr{2010}  \at{On lagrangian and
  vortex-surface fields for flows with taylor--green and kida--pelz initial
  conditions}.  \jt{J. Fluid Mech.}  \bvol{661},  \pg{446--481}.

\bibitem[Zhao {\em et~al.\/}(2016{\natexlab{{\em a\/}}})Zhao, Yang \&
  Chen]{zhao2016evolution}
{\sc \au{Zhao, Y.}, \au{Yang, Y.} \& \au{Chen, S.}} \yr{2016{\natexlab{{\em
  a\/}}}}  \at{Evolution of material surfaces in the temporal transition in
  channel flow}.  \jt{J. Fluid Mech.}  \bvol{793},  \pg{840--876}.

\bibitem[Zhao {\em et~al.\/}(2016{\natexlab{{\em b\/}}})Zhao, Yang \&
  Chen]{zhao2016vortex}
{\sc \au{Zhao, Yaomin}, \au{Yang, Yue} \& \au{Chen, Shiyi}}
  \yr{2016{\natexlab{{\em b\/}}}}  \at{Vortex reconnection in the late
  transition in channel flow}.  \jt{J. Fluid Mech.}  \bvol{802}.

\bibitem[Zhou {\em et~al.\/}(1999)Zhou, Adrian, Balachandar \&
  Kendall]{zhou1999mechanisms}
{\sc \au{Zhou, J.}, \au{Adrian, R.~J.}, \au{Balachandar, S.} \& \au{Kendall,
  T.M.}} \yr{1999}  \at{Mechanisms for generating coherent packets of hairpin
  vortices in channel flow}.  \jt{J. Fluid Mech.}  \bvol{387},  \pg{353--396}.

\bibitem[Ziglin(1988)]{ziglin1988splitting}
{\sc \au{Ziglin, S.L.}} \yr{1988}  \at{Splitting of the separatrices and the
  nonexistence of first integrals in systems of differential equations of
  hamiltonian type with two degrees of freedom}.  \jt{Mathematics of the
  USSR-Izvestiya}  \bvol{31}~(2),  \pg{407}.

\bibitem[Ziglin(1998)]{ziglin1998absence}
{\sc \au{Ziglin, S.L.}} \yr{1998}  \at{On the absence of a real-analytic first
  integral for {A}{B}{C} flow when {A} = {B}}.  \jt{Chaos}  \bvol{8}~(1),
  \pg{272--273}.

\end{thebibliography}
